\documentclass[11pt,a4paper]{article}
\usepackage{fullpage}
\usepackage[bookmarks]{hyperref}
\usepackage{amssymb}
\usepackage{amsmath}
\usepackage{amsthm}
\usepackage{graphicx,color,colordvi}
\usepackage{bbm}
\usepackage{varioref}
\usepackage{cleveref}
\usepackage{stmaryrd}
\usepackage[utf8]{inputenc}
\usepackage[blocks]{authblk}
\usepackage{dsfont}
\usepackage{mathtools}
\usepackage{authblk}
\newcommand{\norm}[1]{\left\| #1 \right\|}

\usepackage{wrapfig}
\usepackage[english]{babel}
\usepackage{physics}
\usepackage{braket}
\usepackage{xcolor}

\usepackage{tikz}
\usepackage{ifthen}
\usepackage{pgfplots}
\pgfplotsset{compat=1.9}
\usetikzlibrary{shapes,arrows}
\usetikzlibrary{positioning}
\usetikzlibrary{shapes.geometric}

\RequirePackage[framemethod=default]{mdframed}

\usepackage[margin=1in]{geometry}

\usepackage{comment}
\usepackage{tikz}
\usepackage{graphicx}
\usepackage[T1]{fontenc}
\usepackage{url}
\usepackage{stmaryrd}
\usepackage{hyperref}
\usepackage{varioref}
\usepackage{cleveref}
\usepackage{braket}
\usepackage{ upgreek }
\usepackage{dsfont}
\usepackage{makecell}
\usepackage{fancyref} 

\newcommand{\identity}{\ensuremath{\mathds{1}}}

\usepackage{float}
\newfloat{algorithm}{t}{lop}

\definecolor{byzantium}{rgb}{0.44, 0.16, 0.39}

\usepackage{hyperref}

\definecolor{grey}{RGB}{220,220,220}
\definecolor{deepjunglegreen}{rgb}{0.0, 0.29, 0.29}

\newmdenv[skipabove=7pt,
skipbelow=7pt,
backgroundcolor=grey!10,
innerleftmargin=5pt,
innerrightmargin=5pt,
innertopmargin=5pt,
leftmargin=0cm,
rightmargin=0cm,
innerbottommargin=5pt,
linewidth=1pt]{sBox}

\newmdenv[skipabove=7pt,
skipbelow=7pt,
backgroundcolor=grey!50,
innerleftmargin=5pt,
innerrightmargin=5pt,
innertopmargin=5pt,
leftmargin=0cm,
rightmargin=0cm,
innerbottommargin=5pt,
linewidth=1pt]{tBox}

\theoremstyle{plain}
\newtheorem{thm}{Theorem}[section]
\newtheorem{lem}[thm]{Lemma}
\newtheorem{cor}[thm]{Corollary}
\newtheorem{prop}[thm]{Proposition}
\newtheorem{defi}[thm]{Definition}

\theoremstyle{remark}
\newtheorem{rem}[thm]{Remark}
\newtheorem{stp}{Step}
\newtheorem*{stp2}{Step 3'}

\newenvironment{theo}{\begin{tBox}\begin{thm}}{\end{thm}\end{tBox}}

\newcommand{\supp}{\operatorname{supp}}

\newcommand{\diam}{\operatorname{diam}}

\newcommand*{\fancyrefthmlabelprefix}{thm}
\newcommand*{\fancyreflemlabelprefix}{lem}
\newcommand*{\fancyrefcorlabelprefix}{cor}
\newcommand*{\fancyrefdefilabelprefix}{defi}
\frefformat{plain}{\fancyreflemlabelprefix}{lemma\fancyrefdefaultspacing#1}
\Frefformat{plain}{\fancyreflemlabelprefix}{Lemma\fancyrefdefaultspacing#1}
\frefformat{plain}{\fancyrefthmlabelprefix}{theorem\fancyrefdefaultspacing#1}
\Frefformat{plain}{\fancyrefthmlabelprefix}{Theorem\fancyrefdefaultspacing#1}
\frefformat{plain}{\fancyrefcorlabelprefix}{corollary\fancyrefdefaultspacing#1}
\Frefformat{plain}{\fancyrefcorlabelprefix}{Corollary\fancyrefdefaultspacing#1}
\frefformat{plain}{\fancyrefdefilabelprefix}{definition\fancyrefdefaultspacing#1}
\Frefformat{plain}{\fancyrefdefilabelprefix}{Definition\fancyrefdefaultspacing#1}
\newcommand*{\fancyrefalglabelprefix}{alg}
\newcommand*{\frefalgname}{algorithm}
\newcommand*{\Frefalgname}{Algorithm}
\frefformat{plain}{\fancyrefalglabelprefix}{%
  \frefalgname\fancyrefdefaultspacing#1%
}%
\Frefformat{plain}{\fancyrefalglabelprefix}{%
  \Frefalgname\fancyrefdefaultspacing#1%
}%

\newcommand*{\fancyrefapplabelprefix}{app}
\newcommand*{\frefappname}{appendix}
\newcommand*{\Frefappname}{Appendix}
\frefformat{plain}{\fancyrefapplabelprefix}{%
  \frefappname\fancyrefdefaultspacing#1%
}%
\Frefformat{plain}{\fancyrefapplabelprefix}{%
  \Frefappname\fancyrefdefaultspacing#1%
}%

\usepackage{tikz}
\usepackage{ifthen}
\usepackage{pgfplots}
\pgfplotsset{compat=1.9}
\usetikzlibrary{shapes,arrows}
\usetikzlibrary{positioning}
\usetikzlibrary{shapes.geometric}

\def\Block[#1,#2,#3,#4]{

\def\r{0.3};

\ifthenelse{\NOT #4=0}{
\fill [#2] (-0.5,-0.5) rectangle ({#1-0.5},0.5);
}

\foreach \n in {1,...,#1}{ 

\shade[shading=ball, ball color=darkred] ({\n-1},0) circle (\r);

}

}

\definecolor{Green}{HTML}{00AD69}  
\definecolor{coolblue}{RGB}{0,51,102}
\definecolor{lightblue}{RGB}{102,210,255}
\definecolor{lightpurple}{RGB}{140,30,255}
\definecolor{lightpink}{RGB}{204,0,204}
\definecolor{midblue}{RGB}{0,102,204}
\definecolor{midpink}{RGB}{153,0,153}
\definecolor{darkblue}{RGB}{0,0,153}
\definecolor{cyan}{RGB}{0,204,204}
\definecolor{lightgreen}{RGB}{0,255,128}
\definecolor{midgreen}{RGB}{0,204,0}
\definecolor{midyellow}{RGB}{204,204,0}
\definecolor{darkyellow}{RGB}{153,153,0}
\definecolor{darkpurple}{RGB}{102,0,102}
\definecolor{orange}{RGB}{255,153,51}
\definecolor{darkred}{RGB}{153,0,76}
\definecolor{lightyellow}{RGB}{255,255,153}
\definecolor{lightred}{RGB}{255,153,153}

\renewcommand{\identity}{\ensuremath{\mathds{1}}}

\begin{document}

\title{\bf Strong decay of correlations\\ for Gibbs states in any dimension}

\author[1,2]{Andreas Bluhm\thanks{andreas.bluhm@univ-grenoble-alpes.fr}}
\author[3,4]{{\'A}ngela Capel\thanks{ac2722@cam.ac.uk}}
\author[5]{Antonio Pérez-Hernández\thanks{antperez@ind.uned.es}}
\affil[1]{\footnotesize Univ.\ Grenoble Alpes, CNRS, Grenoble INP, LIG, 38000 Grenoble, France}
\affil[2]{QMATH, Department of Mathematical Sciences, University of Copenhagen, Universitetsparken 5, 2100 Copenhagen, Denmark}
\affil[3]{Fachbereich Mathematik, Universit\"at T\"ubingen, 72076 T\"ubingen, Germany}
\affil[4]{Department of Applied Mathematics and Theoretical Physics, University of Cambridge, United Kingdom}
\affil[5]{Departamento de Matem\'{a}tica Aplicada I, Escuela T\'{e}cnica Superior de Ingenieros Industriales, Universidad Nacional de Educación a Distancia, calle Juan del Rosal 12, 28040 Madrid (Ciudad Universitaria), Spain}

\maketitle

\begin{abstract}

    Quantum systems in thermal equilibrium are described using Gibbs states. The correlations in such states determine how difficult it is to describe or simulate them. In this article, we show that if the Gibbs state of a quantum system satisfies that each of its marginals admits a local effective Hamiltonian  with short-range interactions, then it satisfies a mixing condition, that is, for any regions $A$, $C$ the distance of the reduced state $\rho_{AC}$ on these regions to the product of its marginals,
    \begin{equation*}
    \norm{ \rho_{AC} \rho_A^{-1} \otimes \rho_C^{-1} - \mathds{1}_{AC}} \, ,
    \end{equation*}
    decays exponentially with the distance between regions $A$ and $C$. This mixing condition is stronger than other commonly studied measures of correlation. In particular, it implies the exponential decay of the mutual information between distant regions. The mixing condition has been used, for example, to prove positive log-Sobolev constants. On the way, we prove that the the condition regarding local effective Hamiltonian is satisfied if the Hamiltonian of the system is commuting and also commutes with every marginal of the Gibbs state. The proof of these results employs a variety of tools such as Araki's expansionals, quantum belief propagation and cluster expansions.

\end{abstract}

\tableofcontents

\section{Introduction}

\subsection{Correlation measures}

\emph{Quantum Gibbs states} are used to describe quantum systems in thermal equilibrium. They are fully described by the system's Hamiltonian and the temperature. For example, for the simulation of many-body systems, it is important to know when Gibbs states allow for an efficient description. This happens for instance if the correlations between faraway regions vanish exponentially fast with the distance between the regions. We refer to \cite{alhambra2022quantum} for a topical review of this and other aspects of quantum systems in thermal equilibrium. 

There are different measures of correlations. In this work, we will focus on a measure of correlations that is called the mixing condition. In order to define it, let us consider Hamiltonians with short-range interactions, i.e., interactions whose strength decays exponentially with the distance. We show that assuming the existence of a local short-range effective Hamiltonian, the following \emph{(uniform) mixing condition} holds at sufficiently high temperature: There exist universal constants $K, \alpha \geq 0$ such that for all finite $\Lambda \subset \mathbb Z^g$, $\rho:= \rho^\Lambda_\beta$ the Gibbs state for the Hamiltonian on the region $\Lambda$ at inverse temperature $\beta>0$, and $A,C \subset \Lambda$ with $A \cap C = \emptyset$,
\begin{equation*}
   \norm{ \rho_{AC} \, \rho_A^{-1} \otimes \rho_C^{-1} - \mathds{1}_{AC}} \leq K
 {f(A,C)} \operatorname{e}^{- \alpha \, \mathrm{dist}(A,C)} \, .
\end{equation*}
Here, $f(A,C)$ is a suitable function that depends on the regions $A$, $C$, for example on their cardinality or the size of their boundaries.

 The name mixing condition comes from the study of modified logarithmic Sobolev inequalities (MLSI) \cite{CapelRouzeStilckFranca-MLSIcommuting-2020,BeigiDattaRouze-ReverseHypercontractivity-2018,art:QuantumConditionalEntropyCapel_2018}, as its homonymous classical analogue \cite{daipra2002classicalMLSI}  is a fundamental ingredient in the proof of such inequalities for classical spin systems. The relevance of MLSIs for quantum spin systems is notorious because they imply rapid mixing for quantum Markovian evolutions describing thermalizing dynamics. Additionally this comes along with a number of important consequences, such as stability under perturbations \cite{art:StabilityCubitt_2015} and the fact that it rules out the usefulness of models as self-correcting quantum memories \cite{Rev:BrownSelf}, among others. In \cite{BardetCapelLuciaPerezGarciaRouze-HeatBath1DMLSI-2019}, it was shown that the mixing condition needs to be assumed in order for heat-bath dynamics in one-dimension to have a positive MLSI constant. In one-dimension, the mixing condition was subsequently used to show that Davies generators converging to an appropriate Gibbs state have a positive MLSI constant at any positive temperature and hence exhibit rapid mixing \cite{BardetCapelGaoLuciaPerezGarciaRouze-Davies1DMLSI-2021, BardetCapelGaoLuciaPerezGarciaRouze-Davies1DMLSIshort-2021}. This  has been recently extended in \cite{kochanowski2023MLSI} to any 2-colorable graph with exponential growth, for which it has been shown that exponential decay of correlations implies rapid mixing via the mixing condition. 

The mixing condition is a very strong notion of decay of correlations. An information-theoretically well-motivated alternative way to quantify the correlations in a quantum state is by using the \emph{mutual information}. This quantity has an operational interpretation as the total amount of correlations (quantum or classical) between two subsystems, as shown in \cite{Groisman2005}. The mutual information between regions $A$ and $C$ is given as 
\begin{equation*}
I_\rho(A:C):= D(\rho_{AC}\|\rho_A \otimes \rho_C),
\end{equation*}
where $D(\rho\| \sigma):= \operatorname{Tr}[\rho(\log \rho - \log \sigma)]$ is the Umegaki relative entropy between quantum states $\rho$ and $\sigma$ \cite{Umegaki-RelativeEntropy-1962}.  We say that $H = (H_\Lambda)_{\Lambda \subset \mathbb{Z}^g}$ has \emph{exponential uniform decay of mutual information} if there exist universal constants $K', \alpha' \geq 0$ such that, given $\beta \geq 0$, for all finite $\Lambda \subset \mathbb Z^g$, for $\rho:= \rho^\Lambda_\beta$ and $A,C \subset \Lambda$ with $A \cap C = \emptyset$,
\begin{equation*}
    I_{\rho}(A:C) \leq K'
 {f'(A,C)} \operatorname{e}^{- \alpha' \mathrm{dist}(A,C)} \, .
\end{equation*}
Examples of systems that have an exponentially-decaying mutual information are those for which there is a Lindbladian that thermalizes \emph{rapidly} to them \cite{Kastoryano2013}. In particular, the mixing condition implies the exponential uniform decay of the mutual information, as shown in \cite{Bluhm2021exponential} by the present authors.

Finally, correlations in many-body systems are traditionally quantified using the \emph{operator} or \emph{covariance correlation}. For a quantum state $\rho $ in $\Lambda$ and regions $A$,$C \subset \Lambda$, it is given by
\begin{equation*}
\operatorname{Cov}_\rho(A,C) := \sup_{O_A, O_C}\left|\operatorname{Tr}[O_A \otimes O_C (\rho_{AC} - \rho_A \otimes \rho_C)]\right| \, .
\end{equation*} 
Here, the operators $O_A$ and $O_C$ have supports on $A$ and $C$, respectively, and the supremum is taken over such operators of operator norm at most $1$. The operator $\rho_X$ is the reduced density matrix of $\rho$  on $X$. \emph{Exponential uniform decay of covariance} is defined similarly as for the mutual information: There exist universal constants $K'', \alpha'' \geq 0$ such that, given $\beta \geq 0$, for all finite $\Lambda \subset \mathbb Z^g$, for $\rho:= \rho^\Lambda_\beta$ and $A,C \subset \Lambda$ with $A \cap C = \emptyset$,
\begin{equation*}
\operatorname{Cov}_\rho(A,C) \leq K''
 {f''(A,C)} \operatorname{e}^{- \alpha'' \mathrm{dist}(A,C)} \, .
\end{equation*} 
Using Pinsker's inequality \cite{Pinsker-Information-1964}, one can easily show that the mutual information upper bounds the covariance, so that decay in mutual information is stronger than decay in covariance. Thus, the mixing condition also implies exponential uniform decay of the covariance.  

This article complements a variety of works that prove decay of correlations for different measures and various setups. In one-dimensional systems, results showing exponential decay of correlations in Gibbs states at any temperature are available for all measures we have discussed: In a seminal paper in 1969 \cite{Araki1969}, Araki showed that the operator correlation of infinite quantum spin chains with local translation-invariant interactions decays exponentially fast. Building on Araki's work, the authors of the present article proved in \cite{Bluhm2021exponential} that systems in finite chains with local translation-invariant interactions satisfy a mixing condition, extending it to exponentially-decaying interactions in \cite{CapelMoscolariTeufelWessel-LPPL-2023} and \cite{Gondolf.2024.ConditionalIndependence}. Therefore, their mutual information in any finite subchain also decays exponentially fast. 

For higher dimensions, the picture is less complete: Exponential decay of the operator correlation for arbitrary graphs above a critical temperature was proved in \cite{Kliesch2014, frohlich2015some}. Contrary to the one-dimensional case, in higher dimensions exponential decay of correlations for \emph{arbitrary} systems can only hold above a critical temperature due to the possible presence of phase transitions (for example, in the classical 2D Ising model). Exponential decay of the mutual information in higher dimensions above a critical temperature for arbitrary graphs is related to the results in \cite{Kuwahara2019} on the conditional mutual information.  
Unfortunately, there is a flaw in the non-commutative cluster expansion of this paper \cite{samuel-personal-communication}. While the recent paper \cite{Kuwahara.2024} proves the decay of the conditional mutual information at any positive temperature, this result does not imply decay of the mutual information. Moreover, it does not prove the existence of an effective Hamiltonian of the form that was claimed in \cite{Kuwahara2019}.

{
\subsection{Motivation}\label{sec:main_results}
In the previous section, we have seen that there are different ways to quantify decay of correlations, using, e.g., the mixing condition, the mutual information, or the covariance. The exponential decay of covariance is the condition that is most commonly used, whereas the mixing condition has been used successfully to prove rapid mixing \cite{BardetCapelGaoLuciaPerezGarciaRouze-Davies1DMLSI-2021, BardetCapelGaoLuciaPerezGarciaRouze-Davies1DMLSIshort-2021}.
While the mixing condition implies exponential decay of the mutual information, which in turn implies exponential decay of covariance, these implications cannot be reversed in general. For example, from data-hiding, it is known that there exist states whose operator correlations are arbitrarily small, but whose mutual information is still big \cite{Hastings2007, Hayden2004}. 

However, for classical Gibbs states, it is known that all these different forms of decay of correlations are equivalent  \cite{Martinelli1999}. The main motivation for this paper is to show that this is also true in the quantum setting, i.e., to show that for Gibbs states, exponential decay of covariance implies the mixing condition. 

In previous work \cite{Bluhm2021exponential}, the present authors showed that for quantum spin chains at any positive temperature with local, finite-range, translation-invariant interactions, the three notions of decay of correlations we discussed are all equivalent. In fact, we can even add another one, namely \emph{local indistinguishability} of the Gibbs state \cite{Brandao2019}. The latter holds if there exist universal constants $K''', \alpha''' \geq 0$ such that for all $\Lambda \subset \mathbb Z^g$,  split as $\Lambda=ABC$ with $B$ shielding $A$ from $C$,  and for all local operators $O_A$ on $A$,
\begin{equation*}
\left| \operatorname{Tr}_{ABC}[\rho^{\Lambda} \, O_A ] - \operatorname{Tr}_{AB}[\rho^{AB} \, O_A ] \right| \leq  \norm{O_A}  f'''(A,C) K''' 
 \operatorname{e}^{- \alpha''' \mathrm{dist}(A,C)} \, . 
\end{equation*}

The equivalence of these measures of decay of correlations and local indistinguishability in 1D was subsequently extended to short-range interactions in \cite{CapelMoscolariTeufelWessel-LPPL-2023} and \cite{Gondolf.2024.ConditionalIndependence}, and it remains true beyond one-dimensional systems for classical or even commuting systems with finite range, as shown recently in \cite{kochanowski2023MLSI}. The caveat of the latter result though is that it presents a prefactor scaling exponentially with the size of the boundaries of $A$ and $C$.  

In this article, we make progress in the case of non-commutative interactions in systems with more than one dimension. We are inspired by the flawed proof of exponential decay of conditional mutual information in \cite{Kuwahara2019}, which relied on the existence of an effective Hamiltonian. In fact, we consider two different types of effective Hamiltonians, which we term \textit{weak} and \textit{strong}, respectively (see Section \ref{sec:effective_hamiltonian} for details).

Under the existence of a strong effective Hamiltonian, we prove that for sufficiently high temperatures, the mixing condition holds.  Using cluster expansion techniques, we can show in Section \ref{sec:comm-effective-Hamiltonian} that such a strong effective Hamiltonian exists at high enough temperature if the interactions satisfy a \emph{commuting hypothesis} (Definiton \ref{defi:CommutingHypothesis}), meaning essentially that the interactions and all their partial traces commute.

However, assuming the existence of a strong effective Hamiltonian is likely to be too restrictive for general non-commuting interactions. Therefore, we also consider a weak effective Hamiltonian, which is an effective Hamiltonian of the form claimed in \cite{Kuwahara2019}. Making use of local indistinguishability (which follows from exponential decay of covariance \cite{CapelMoscolariTeufelWessel-LPPL-2023}), we can then show that, also under this weaker hypothesis, exponential decay of covariance implies the mixing condition.

Thus, to summarize, while we cannot prove unconditional equivalence of the different measures of decay of correlations at high enough temperature in this article, we can show this equivalence assuming the existence of a local effective Hamiltonian, in two possible versions. The precise systems for which these effective Hamiltonians exist beyond the commuting case remains an open problem though, and will be addressed in future work.

\subsection{Mixing condition and proof outline} \label{sec:proof-outline}

The main results of this paper deal with the implication from local indistinguishability to mixing condition under the assumption of the existence of an effective Hamiltonian with short-range interactions. We explore separately the derivation of mixing condition in the presence of a so-called \textit{strong} (\Cref{defi:localityEffectiveHamiltonian}) or \textit{weak} (\Cref{defi:weakeffHamiltonian}) effective local Hamiltonian with short-range interactions, respectively. 

More specifically, for the \textbf{\underline{strong}} case, we prove that, given a finite lattice $\Lambda \subset \mathbb{Z}^g$ and $A, C \subset \Lambda$ such that $A$  and $C $ are ``separated enough'', and for a Gibbs state $\rho^\Lambda_\beta \equiv \rho  = \operatorname{e}^{- \beta H_\Lambda} /\Tr[\operatorname{e}^{- \beta H_\Lambda}] $ of a short-range Hamiltonian with $\beta < \beta_*$, where $\beta_*$ is some sufficiently low inverse temperature, we have
\begin{equation}\label{eq:mixing_condition_intro}
    \norm{ \rho_{AC} \, \rho_A^{-1} \otimes \rho_C^{-1} - \identity_{AC}} \leq \zeta \operatorname{e}^{- \eta \, \mathrm{dist}(A,C)} \, ,
\end{equation}
where $\zeta, \eta $ are absolute constants depending on the interactions and $\beta$, and additionally $ \zeta = \mathcal{O} \left( e^{  \operatorname{min} \{ |\partial A| ,   |\partial C| \} } , e^\beta  \right) $. 
Here, $\partial X$ is the $1$-boundary of $X$, i.e., all sites in the complement of $X$ that have distance $1$ from $X$.

We prove the assumption required, namely the existence of a strong effective Hamiltonian with short-range interactions for interactions that satisfy the Commuting Hypothesis (Definiton \ref{defi:CommutingHypothesis}). The derivation of Eq. \eqref{eq:mixing_condition_intro} is then relatively straightforward and shown in \Cref{sec:stronf_effHam_implies_mixing_condition}.

Next, we assume the existence of a \textbf{\underline{weak}} effective Hamiltonian, and prove Eq. \eqref{eq:mixing_condition_intro} in this weaker case, with modified constants $\tilde{\zeta}, \tilde{\eta} $ such that 
\begin{equation*}
\tilde{\zeta}= \mathcal{O}\left(   \operatorname{min} \{ e^{|\partial A|}(|\partial A| +|C|g(A)) ,   e^{|\partial C|}(|\partial C| +|A|g(A)) \}  , e^\beta  \right)\, ,
\end{equation*}
where the factors $g(A)$ and $g(C)$ are inherited from the notion of clustering of correlations assumed to hold. The proof of this result is quite involved and requires the use of strong machinery in the context of Gibbs states. In particular, we make use in our proof of the so-called cluster expansions, the well-known Quantum Belief Propagation (QBP) \cite{Hastings2007,kim2012perturbative,CapelMoscolariTeufelWessel-LPPL-2023} and estimates on Araki's expansionals \cite{Perez2020}. Let us sketch here the proof of this result by combining these tools. The complete proof can be found in the next sections. 

\vspace{0.2cm}

\noindent \textbf{\underline{Step 1. Construction of the effective Hamiltonian.}} 

\vspace{0.1cm}
\noindent {In a first step, motivated by the ideas of \cite{Kuwahara2019}, we assume the existence of a local effective Hamiltonian} $\widehat{H}^{L, \beta}_\Lambda$ for our original Hamiltonian $H_\Lambda$ such that, for every $L \subset \Lambda$ (cf. Section \ref{sec:effective_hamiltonian}):
\begin{equation*}
 \widehat{H}^{L, \beta}_\Lambda:=-\frac{1}{\beta} \log\left( \tr_{L^{c}}(e^{-\beta H_\Lambda}) \otimes \mathbbm{1}_{L^{c}} \right) + \frac{1}{\beta} \log[Z_{L^{c}}] \mathbbm{1} \, .
\end{equation*}
We can control the interaction terms of $\operatorname{e}^{- \beta \widehat{H}_\Lambda^{L,\beta}}$, as well as bound the expansionals of the form 
\begin{equation*}
    \operatorname{e}^{- \beta \widehat{H}_\Lambda^{AB,\beta}}  \operatorname{e}^{ \beta (\widehat{H}_\Lambda^{A,\beta} + \widehat{H}_\Lambda^{B,\beta} )  } \, .
\end{equation*}
In particular, the previous construction allows us to relate the marginals of the original Hamiltonian to the exponentials of the effective Hamiltonian in the following form (see Eq.\ \eqref{eq:relation_effrho_usualrho}):
\begin{equation*}
    \rho_{AC} \, \rho_A^{-1} \otimes \rho_C^{-1} = \operatorname{e}^{-\beta \widehat{H}_\Lambda^{AC,\beta}} \operatorname{e}^{\beta (\widehat{H}_\Lambda^{A,\beta} + \widehat{H}_\Lambda^{C,\beta})} \underbrace{Z_{ABC} Z_B Z_{AB}^{-1} Z_{BC}^{-1}}_{\kappa_{ABC}} \, ,
\end{equation*}
where $Z_{X}$ is just $\Tr_X[\operatorname{e}^{-\beta H_{X}}]$.  Therefore, we can bound
\begin{align}\label{eq:estimate_mixing_condition_informal}
   & \norm{ \rho_{AC} \rho_A^{-1} \otimes \rho_C^{-1} - \identity} \nonumber \\
   & \hspace{2cm} \leq \norm{ \operatorname{e}^{-\beta \widehat{H}_\Lambda^{AC,\beta}} \operatorname{e}^{\beta (\widehat{H}_\Lambda^{A,\beta}+ \widehat{H}_\Lambda^{C,\beta})} }  \abs{\kappa_{ABC} -1} + \norm{ \operatorname{e}^{-\beta \widehat{H}_\Lambda^{AC,\beta}} \operatorname{e}^{\beta (\widehat{H}_\Lambda^{A,\beta} + \widehat{H}_\Lambda^{C,\beta})} - \identity_{AC} }  \, .
\end{align}
Now we need to estimate each of these terms separately.

\vspace{0.2cm}

\noindent  \textbf{\underline{Step 2. Estimates on the expansionals of the effective Hamiltonian.}} 

\vspace{0.1cm}

\noindent  For estimating the last term in the RHS of Eq.\ \eqref{eq:estimate_mixing_condition_informal}, we use the estimates for Araki's expansionals for the effective Hamiltonian (as in \Cref{prop:estimates_expansionals_normal} for the original interaction), concluding:
\begin{equation*}
    \norm{ \operatorname{e}^{-\beta \widehat{H}_\Lambda^{AC,\beta}} \operatorname{e}^{\beta (\widehat{H}_\Lambda^{A,\beta} + \widehat{H}_\Lambda^{C,\beta})} } \leq \operatorname{e}^{\mathcal{K}_1\mathcal{K}_2} \, ,
\end{equation*}
for $\mathcal{K}_1$ a constant and
\begin{equation*}
    \mathcal{K}_2 \leq \mathcal{O}\left( {\operatorname{min} \{ |\partial A| , |\partial C| \} } \operatorname{e}^{- \mathrm{dist}(A,C)} \right) \, .
\end{equation*}
We can similarly estimate the first term in the RHS above, obtaining:
\begin{equation*}
    \norm{ \operatorname{e}^{-\beta \widehat{H}_\Lambda^{AC,\beta}} \operatorname{e}^{\beta (\widehat{H}_\Lambda^{A,\beta} + \widehat{H}_\Lambda^{C,\beta})} - \identity_{AC} } \leq \operatorname{e}^{\mathcal{K}_1\mathcal{K}_2} -1 \, .
\end{equation*}

\noindent  \textbf{\underline{Step 3. Estimates on partition functions of the original Hamiltonian.} } 

\vspace{0.1cm}

\noindent  The remaining term from Eq.\ \eqref{eq:estimate_mixing_condition_informal} to be bounded is $\abs{\kappa_{ABC} -1}$.  We bound it in \Cref{lem:clusteringImpliesPartitionFunction}  using the result of local indistinguishability from \Cref{thm:local_indistinguishability_kastoryanobrandao} as well as the estimates for Araki's expansionals for the original Hamiltonian from \Cref{prop:estimates_expansionals_normal}, obtaining thus:
\begin{equation*}
   \abs{\kappa_{ABC} -1} \leq    \min \{ \mathcal{O}  (|\partial A| + |C| g(A) ) e^{\mathcal{O}(|\partial A|) },\mathcal{O}  (|\partial C| + |A| g(C) ) e^{\mathcal{O}(|\partial C|) }    \}  \, , 
\end{equation*}
where the factors $g(A)$ and $g(C)$ are inherited from the notion of clustering of correlations assumed to hold. 
Note that, in the proof of \Cref{thm:local_indistinguishability_kastoryanobrandao}, we additionally make use of the Quantum Belief Propagation. 

\vspace{0.2cm}

\section{Setting and beyond}\label{sec:setting}

\subsection{Notations and model}\label{subsec:model}

Let $G=(V, E)$ be a possibly infinite graph with vertices $V$ and edges $E$. We endow the graph with a metric $\mathrm{dist}: V \times V \to \mathbb R_+$, for example the shortest path distance on the graph. This fixes the set $V$ our quantum systems live on. For the distance between sets $X$, $Y \subset V$, let 
\begin{equation*}
    \mathrm{dist}(X,Y) := \inf_{x \in X} \inf_{y \in Y} \mathrm{dist}(x,y)\, .
\end{equation*}
We write the double inclusion $X \subset \subset V$ to indicate that the subset $X$ is finite. The set of finite subsets of $V$ will be denoted by $\mathcal{P}_{f}(V)$. 

The \textit{diameter} of a finite subset $X$ of $V$ is given by $\diam(X) = \max_{x,y \in X}\mathrm{dist}(x,y)$. For $A \subset \subset V$ and $r>0$, we denote by $\partial_{r} A$ the subset of $A$ made of all sites whose distance from $A^{c}:=V \setminus A$ is less than or equal to $r$. In particular, we will write $ \partial A:=\partial_{1} A$.

Let us now come to the Hilbert space associated to the quantum spin system. At each site $x \in V$ we set a local Hilbert space $\mathcal{H}_{x} \equiv \mathbb{C}^{D}$ of dimension $D \in \mathbb N$. For each $X \in \mathcal{P}_{f}(V)$ we then have the space of states $\mathcal{H}_{X} = \otimes_{v \in X}\mathcal{H}_{v} \equiv (\mathbb{C}^{D})^{\otimes |X|}$ of dimension $D_{X} = D^{|X|}$ and the algebra of observables $\mathfrak{A}_{X} = \mathcal{B}(\mathcal{H}_{X})$. As usual, for two finite subsets $X$, $Y$ of $V$ such that $X \subset Y$ we identify $\mathfrak{A}_{X} \subset \mathfrak{A}_{Y}$ via the canonical linear isometry $\mathfrak{A}_{X} \to \mathfrak{A}_{Y}$ given by $Q \mapsto Q \otimes \mathbbm{1}_{Y \setminus X}$. This allows to define the algebra of local observables  as the inductive limit $\mathfrak{A}_{loc}:=\bigcup_{X \in \mathcal{P}_{f}(V)} \mathfrak{A}_{X}$.  We will say that a local observable $Q \in \mathfrak{A}_{loc}$ is supported in $X \in \mathcal{P}_{f}(V)$, if $Q$ belongs to $\mathfrak{A}_{X}$.

Next, we will describe which notation we will use for different versions of the trace. For each $X \in \mathcal P_{f}(V)$, we will denote by $\Tr_{X}:\mathfrak{A}_{X} \longrightarrow \mathbb{C}$ the full (unnormalized) trace over $X$. For the partial trace over $X$ on any $ X' \in  P_{f}(V) $ with $X \subset X'$, we will write
\begin{equation*}
    \operatorname{tr}_{X}:= \mathrm{Tr}_X \otimes \mathrm{id}_{X^\prime \setminus X}: \mathfrak{A}_{X'} \longrightarrow \mathfrak{A}_{X' \setminus X}
\end{equation*}
and combine this map with the above canonical isometries. For instance, for a state $\sigma \in \mathfrak{A}_{X^\prime}$ and $X \subset X^\prime$, we can write $\operatorname{tr}_{X}(\sigma) = \operatorname{tr}_{X}(\sigma) \otimes \mathds{1}_{ X} \in \mathfrak A_{X^\prime}$. In particular, if $Q \in \mathfrak{A}_{X}$, then we will deal with $\operatorname{tr}_{X}(Q)$ as a multiple of identity, $\operatorname{tr}_{X}(Q) = \Tr_X(Q) \mathds{1}_X$. The normalized version of the partial trace, which is a conditional expectation, will be denoted $\mathbb E_{X^c} := \tr_{X}/D_{X}$. In this case, given $Q \in \mathfrak{A}_{X}$ and two subsets $Y, Y' \in \mathcal P_{f}(V)$ with $Y, Y' \supset X$, we can identify $\mathbb E_{Y^c}(Q) = \mathbb E_{(Y')^c}(Q)$. In terms of norms, we will denote by $\|Q\|$ the operator norm of $Q \in \mathfrak A_X$, and by $\| Q \|_1 = \operatorname{Tr}_X(| Q |)$ its trace norm.

Let us present now the kind of Hamiltonians we will consider. By a \textit{local interaction}, we refer to a family $\Phi = (\Phi_{X})_{X \in \mathcal{P}_{f}(V)}$,  where $\Phi_X \in \mathfrak{A}_X$ and $\Phi_X = \Phi_X^\ast$ for every $X \in \mathcal{P}_{f}(V)$. To quantify the decay of the interactions, we introduce for each $\lambda, \mu>0$ the following notation
\begin{equation}\label{eq:norm_interaction}
    \| \Phi\|_{\lambda, \mu}:= \sup_{x \in V} \sum_{X \ni x} \| \Phi_{X}\| e^{\lambda |X| + \mu \diam(X)} \in [0, \infty]\,.
\end{equation}
We will say that $\Phi$ has \textit{finite range} $r > 0$ and \textit{strength} $J > 0$ if $\| \Phi_{X}\| = 0$ whenever $X$ has diameter greater than $r$ and $\| \Phi_{X}\| \leq J$ for all $X \in \mathcal{P}_{f}(V)$. Moreover, we will say that  $\Phi$ has \textit{short range}, or it is \textit{exponentially decaying}, if $\norm{\Phi}_{\lambda,\mu} < \infty$. As usual, we denote for every finite subset $Y \subset \subset V$ the corresponding \textit{Hamiltonian} by
\[ H_{Y} := \sum_{X \subset Y}{\Phi_{X}} \, , \]
the \textit{time-evolution operator} (with possibly complex-valued time) by
\[ \Gamma^{s}_{H_Y}(Q) = e^{isH_Y}Qe^{-isH_Y}  \quad , \quad s \in \mathbb{C}\,,\]
and the \textit{Gibbs state} at inverse temperature $\beta >0$ by
\begin{equation*}
    \rho^{Y}_\beta := \frac{e^{-\beta H_Y}}{\Tr_{Y}[e^{-\beta H_Y}]} \,.
\end{equation*}
Moreover, for any $X \subset Y$, we denote by $\rho_{\beta,X}^Y$ the marginal in $X$ of the Gibbs state of $H_Y$ at inverse temperature $\beta >0$, namely
\[ \rho_{\beta,X}^Y := \tr_{Y \setminus X} \left[\rho^Y_\beta \right] \, . \]
Note that seen as an element in $\mathfrak{A}_Y$, $\rho_{\beta,X}^Y$ is no longer a quantum state, because its trace is no longer normalized. We will frequently drop the superindex $Y$ when it is clear from the context, as well as the subindex $\beta$ when we are fixing the temperature.

\subsection{Locality and time evolution }\label{subsec:time_evolution}

We devote this subsection to deriving some estimates on the norm of the time-evolution operator for short-range interactions. We provide below both universal estimates on such time-evolution operators, as well as decay estimates on the difference between pairs of them.  

\begin{figure}[ht]
\begin{center}

\begin{tikzpicture}[scale=0.4]


\fill [darkyellow!50!orange!25!white] (-3.5,-1.5) rectangle (18.5,13.5);


\foreach \n in {-3,-2,...,18}{
\foreach \m in {-1,0,1,...,13}{
\shade[shading=ball, ball color=darkred!10!white] (\n,\m) circle (0.2);
}
 }


\draw [black, very thick] (5.5,4.5) rectangle (8.5,7.5);
\fill [blue, opacity=0.05] (5.5,4.5) rectangle (8.5,7.5);
\node at (6.5,6.5) {\small $\textbf  Z$};


\draw [blue, very thick] (0.5,0.5) rectangle (12.5,11.5);
\node at (3.5,10.5) {\small $\textbf  Y$};


\draw [red, very thick] (-2.5,-0.5) rectangle (13.5,12.5);
\node at (-1.5,10.5) {\small $\textbf  Y'$};


\draw[<->,  red!80!black, very thick] (8.5,5.5) -- (12.5,5.5);
\node[red!80!black] at (10.5,6.5) {\tiny $\operatorname{dist}(Z,Y^{c})$};

\end{tikzpicture}

 \caption{Example of configuration of regions $Z \subset Y \subset Y'$ in Proposition \ref{theo:localityEstimates}. If the local interaction of the system is exponentially decaying, then the evolutions of an observable supported in $Z$ under $H_{Y}$ and $H_{Y'}$, respectively, are exponentially close to each other in the distance from $Z$ to the complement of $Y$.}
  \end{center}
\end{figure}
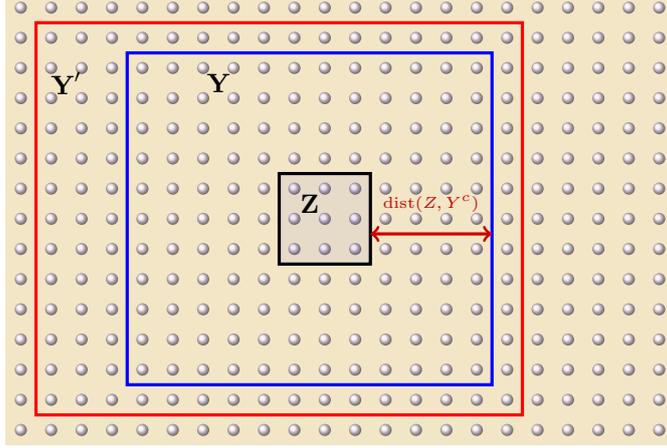

\begin{prop}\label{theo:localityEstimates}
Let  $\Phi$ be an interaction on $V$ satisfying for some constants $\lambda, \mu \in [0, \infty)$ that
\begin{equation*} 
\| \Phi\|=\| \Phi\|_{\lambda, \mu}:=\sup_{x \in V} \sum_{X \ni x} \| \Phi_{X}\| e^{\lambda |X| + \mu \operatorname{diam}(X)} < \infty   \,, 
\end{equation*}
and let $Q$ be an observable having support in a finite subset $Z$ of $V$. If $Y \in \mathcal{P}_{f}(V)$, then for every $s \in \mathbb{C}$ with $|s| < \lambda/(2 \| \Phi\|)$ 
\begin{equation}\label{equa:time_evolution_main_1}
\|\Gamma_{H_{Y}}^{s}(Q)\| \leq \| Q\| \, e^{\lambda |Z|}  \, \frac{\lambda}{\lambda - 2 \| \Phi\| \, |s|}\, . 
\end{equation}
Moreover, if $Y' \in \mathcal{P}_{f}(V)$ and $Z \subset Y \subset Y'$, then for every $s \in \mathbb{C}$ with $|s| < \lambda/(2 \| \Phi\|)$,
\begin{equation}\label{equa:time_evolution_main_2}
\left\| \Gamma_{H_{Y'}}^{s}(Q) - \Gamma_{H_{Y}}^{s}(Q) \right\| \leq \| Q\| \,e^{\lambda |Z|} \,  \frac{2 \| \Phi\| \, |s| \lambda}{(\lambda - 2 \| \Phi\| \, |s|)^{2}} \,  e^{- \mu \operatorname{dist}(Z, V \setminus Y)} \,.
 \end{equation}
\end{prop}

Before proving this result, let us mention that Eq.\ \eqref{equa:time_evolution_main_1} should be compared to  \cite[Theorem 6.2.4]{BraRob97}, to which our estimate reduces whenever  $\mu=0$. Moreover, Eq.\ \eqref{equa:time_evolution_main_2} for $\mu=0$ is essentially in \cite[after Theorem 6.2.4]{BraRob97} and can be interpreted as a manifestation of locality of the interactions. Note that we have introduced the weight  $e^{\mu \diam(X)}$ to control the decay of the interactions with the diameter, following the approach of \cite{Ueltschi2004}.

\begin{proof}
The time-evolution operator can be written in terms of derivations $\delta_{H_{Y}}(Q) = i[H_{Y}, Q]$:
\begin{equation}\label{equa:timeEvolutionAux1} 
 \Gamma_{H_{Y}}^{s}(Q) = \sum_{m=0}^{\infty} \frac{s^{m}}{m!} \delta_{H_{Y}}^{m}(Q)\,.
 \end{equation}
In turn, each $\delta_{H_{Y}}(Q)$ can be expanded as a sum of $\delta_{X}(Q):=[\Phi_{X}, Q]$, leading to
\[ \delta_{H_{Y}}^{m}(Q) = \delta_{H_{Y}} \circ \ldots \circ \delta_{H_{Y}}(Q) = \sum_{X_{1} \cap Z \neq \emptyset} \ldots \sum_{X_{m} \cap S_{m-1} \neq \emptyset} \, \delta_{X_{m}} \circ \delta_{X_{m-1}} \circ \ldots \circ \delta_{X_{1}}(Q)\,, \]
where $S_{0}:=Z$,  $S_{j} := Z \cup X_{1} \cup \ldots \cup X_{j}$ for $1 \leq j \leq m$, and the sums are extended over subsets $X_{j} \subset Y$. We next follow an argument inspired by the proof of \cite[Theorem 6.2.4]{BraRob97} to estimate
\begin{equation}\label{equa:timeEvolutionAux2}
\| \delta_{H_{Y}}^{m}(Q)\|  \leq 2^{m} \| Q\|  \sum_{X_{1} \cap Z \neq \emptyset} \ldots \sum_{X_{m} \cap S_{m-1} \neq \emptyset} \,\,\, \prod_{j=1}^{m} \| \Phi_{X_{j}}\| \,.
\end{equation}
Next, let us rewrite
\[ 
\prod_{j=1}^{m} \| \Phi_{X_{j}}\| =  e^{-\lambda ( |X_{1}| + \ldots |X_{m}|)} \, \prod_{j=1}^{m} \| \Phi_{X_{j}}\| e^{\lambda |X_{j}|} \leq e^{\lambda |Z|} e^{- \lambda |S_{m}|} \, \prod_{j=1}^{m} \| \Phi_{X_{j}}\| e^{\lambda |X_{j}|}  \, .
\]
Applying the inequality $e^{-\lambda x} \leq \frac{m!}{\lambda^{m} x^{m}}$ valid for every $\lambda$, $x >0$ with $x=|S_{m}|$ in the previous expression,  we moreover get 
\begin{equation}\label{equa:timeEvolutionAux3}
\prod_{j=1}^{m} \| \Phi_{X_{j}}\| \leq e^{\lambda |Z|} \frac{m!}{|S_{m}|^{m} \lambda^{m}} \, \prod_{j=1}^{m} \| \Phi_{X_{j}}\| e^{\lambda |X_{j}|} \leq e^{\lambda |Z|} \frac{m!}{\lambda^{m}} \, \prod_{j=1}^{m} \frac{1}{|S_{j-1}|} \| \Phi_{X_{j}}\| e^{\lambda |X_{j}|}  \, ,
\end{equation}
and thus
\begin{equation}\label{equa:timeEvolutionAux4}
\| \delta_{H_{Y}}^{m}(Q)\| \leq  2^{m} \| Q\| \, e^{\lambda |Z|} \frac{m!}{\lambda^{m}} \sum_{X_{1} \cap Z \neq \emptyset} \ldots \sum_{X_{m} \cap S_{m-1} \neq \emptyset} \,\,\, \prod_{j=1}^{m} \frac{1}{|S_{j-1}|} \| \Phi_{X_{j}}\|e^{\lambda |X_{j}|} \, .
\end{equation}
Finally, note that we can bound for each finite subset $Y$ of the lattice
\begin{equation}\label{equa:timeEvolutionAux5}
\sum_{X \cap Y \neq \emptyset} \| \Phi_{X}\| e^{\lambda |X| + \mu \operatorname{diam}(X)} \leq \sum_{v \in Y} \sum_{X \ni v}\| \Phi_{X}\| e^{\lambda |X| + \mu \operatorname{diam}(X)}\leq |Y| \| \Phi\|\,. 
\end{equation}
Applying Eq.\ \eqref{equa:timeEvolutionAux5}  iteratively,  we can estimate
\[ \sum_{X_{1} \cap Z \neq \emptyset} \ldots \sum_{X_{m} \cap S_{m-1} \neq \emptyset} \,\,\, \prod_{j=1}^{m} \frac{1}{|S_{j-1}|} \| \Phi_{X_{j}}\|e^{\lambda |X_{j}|} \leq \| \Phi\|^{m} \,,\]
so that 
\begin{equation}\label{equa:EvolutionNormAux2}
\| \delta_{H_{Y}}^{m}(Q)\| \leq \left(\frac{2\| \Phi\|}{\lambda}\right)^{m} \| Q\| \, e^{\lambda |Z|} \, m! \,.
\end{equation}
Applying Eq.\ \eqref{equa:EvolutionNormAux2} in Eq.\ \eqref{equa:timeEvolutionAux1}, 
\[ \|\Gamma_{H_{Y}}^{s}(Q)\| \leq \| Q\| \, e^{\lambda |Z|}  \, \sum_{m=0}^{\infty} \left( \frac{2\| \Phi\|\,|s|}{\lambda}\right)^{m} \, . \]
Using finally the formula $\frac{1}{1-x} = \sum_{m=0}^{\infty} x^{m}$ for $|x|<1$ we arrive at Eq.\ \eqref{equa:time_evolution_main_1}. 

Next, we prove the other estimate. Using Eq.\ \eqref{equa:timeEvolutionAux1}, we can again upper bound
\begin{equation}\label{equa:timeEvolutionComparisonAux3} 
\| \Gamma^{s}_{H_{Y'}}(Q) - \Gamma^{s}_{H_{Y}}(Q)\| \leq \sum_{m=1}^{\infty} \frac{|s|^{m}}{m!} \left\|\delta_{H_{Y'}}^{m}(Q) - \delta_{H_Y}^{m}(Q)\right\| \, .
\end{equation}
Each summand can be bounded following a similar strategy to the first inequality of the theorem. Let us denote by $\Phi^{Y}$ the local interaction on $V$ given by $\Phi_{X}^{Y} = \Phi_{X}$ if $X \subset Y$ and $\Phi_{X}^{Y} = 0$ if $X \nsubseteq Y$. Then,
\begin{multline}\label{equa:timeEvolutionComparisonAux4} 
\| \delta_{H_{Y'}}^{m}(Q)  - \delta_{H_{Y}}^{m}(Q)  \| \\ 
\leq 2^{m} \| Q\| \sum_{j=1}^{m} \sum_{X_{1} \cap Z \neq \emptyset} \ldots \sum_{X_{m} \cap S_{m-1} \neq \emptyset}  \left( \prod_{\substack{i=1}}^{j-1} \| \Phi^Y_{X_{i}}\| \right) \cdot \| \Phi_{X_{j}} - \Phi^Y_{X_{j}}\| \cdot \left( \prod_{\substack{i=j+1}}^{m} \| \Phi_{X_{i}}\| \right)\,,
\end{multline}
where the sums run over subsets $X_{i}$ of $Y^\prime$ satisfying $X_{i} \cap S_{i-1} \neq \emptyset$.

To deal with the previous term, we argue as with Eq.\ \eqref{equa:timeEvolutionAux2}, but adding one additional intermediate step. More specifically, we first estimate as in Eq.\ \eqref{equa:timeEvolutionAux3} to get for each $j \in \{ 1, \ldots, m\}$
\begin{equation*}
    \begin{split}
        & \left( \prod_{\substack{i=1}}^{j-1} \| \Phi^Y_{X_{i}}\| \right) \cdot \| \Phi^Y_{X_{j}} - \Phi_{X_{j}}\| \cdot \left( \prod_{\substack{i=j+1}}^{m} \| \Phi_{X_{i}}\| \right) \leq \\[2mm]
        & \hspace{3cm} \leq  e^{\lambda |Z|} \frac{m!}{\lambda^{m}} \,\, \left( \prod_{\substack{i=1}}^{j-1} \frac{\| \Phi^Y_{X_{i}}\| e^{\lambda |X_{i}|}}{|S_{i-1}|} \right) \cdot \frac{\| \Phi^Y_{X_{j}} - \Phi_{X_{j}}\| e^{\lambda |X_{j}|}}{|S_{j-1}|} \cdot \left( \prod_{\substack{i=j+1}}^{m} \frac{\| \Phi_{X_{i}}\| e^{\lambda |X_{i}|}}{|S_{i-1}|} \right) \, .
    \end{split}
\end{equation*}
Then, applying Eq.\ \eqref{equa:timeEvolutionAux5} to the last $m-j$ terms iteratively
\begin{equation}\label{equa:timeEvolutionAux6}
\begin{split}
& \sum_{X_{1} \cap Z \neq \emptyset} \ldots \sum_{X_{m} \cap S_{m-1} \neq \emptyset}  \left( \prod_{\substack{i=1}}^{j-1} \| \Phi^Y_{X_{i}}\| \right) \cdot \| \Phi^Y_{X_{j}} - \Phi_{X_{j}}\| \cdot \left( \prod_{\substack{i=j+1}}^{m} \| \Phi_{X_{i}}\| \right) \\[2mm]
& \hspace{1cm} \leq \| \Phi\|^{m-j} \cdot e^{\lambda |Z|} \frac{m!}{\lambda^{m}} \sum_{X_{1} \cap Z \neq \emptyset} \ldots \sum_{X_{j} \cap S_{j-1} \neq \emptyset}  \left( \prod_{\substack{i=1}}^{j-1} \frac{\| \Phi^Y_{X_{i}}\| e^{\lambda |X_{i}|}}{|S_{i-1}|} \right) \cdot \frac{\| \Phi^Y_{X_{j}} - \Phi_{X_{j}}\| e^{\lambda |X_{j}|}}{|S_{j-1}|} \, .
\end{split}
\end{equation}
Let us observe that, by definition, $\Phi_{X} -\Phi_{X}^{Y} = 0$ if $X \subset Y$ and $\Phi_{X} -\Phi_{X}^{Y} = \Phi_{X}$ if $X \nsubseteq Y$ (i.e. $X \cap (V \setminus Y) \neq \emptyset$). Thus, in the above expression we can restrict the sum over $X_{j}$ with $X_{j} \cap S_{j-1} \neq \emptyset$ to sets $X_{j}$ that also satisfy $X_{j} \cap (V \setminus Y) \neq \emptyset$ and simplify $\Phi_{X_j} -\Phi_{X_j}^{Y} = \Phi_{X}$. Thus, the upper bound from Eq.\ \eqref{equa:timeEvolutionAux6} can be rewritten as
\begin{equation}\label{equa:timeEvolutionAux7}
\| \Phi\|^{m-j} \cdot e^{\lambda |Z|} \frac{m!}{\lambda^{m}} \sum_{X_{1} \cap Z \neq \emptyset} \ldots \sum_{X_{j-1} \cap S_{j-2} \neq \emptyset}\sum_{\substack{X_{j} \cap S_{j-1} \neq \emptyset \\ X_{j} \cap (V \setminus Y) \neq \emptyset}}  \left( \prod_{\substack{i=1}}^{j-1} \frac{\| \Phi^Y_{X_{i}}\| e^{\lambda |X_{i}|}}{|S_{i-1}|} \right) \cdot \frac{\|  \Phi_{X_{j}}\| e^{\lambda |X_{j}|}}{|S_{j-1}|}\,. 
\end{equation}
Note that the conditions on $X_{1}, \ldots, X_{j}$ yield that
\[ \operatorname{dist}(Z, V \setminus Y) \leq \sum_{i=1}^{j} \diam(X_{i})\,. \]
Hence, we can introduce a factor:
\begin{multline}
\left( \prod_{\substack{i=1}}^{j-1} \frac{\| \Phi^Y_{X_{i}}\| e^{\lambda |X_{i}|}}{|S_{i-1}|} \right) \cdot \frac{\|  \Phi_{X_{j}}\| e^{\lambda |X_{j}|}}{|S_{j-1}|}\\
 \leq 
\left( \prod_{\substack{i=1}}^{j-1} \frac{\| \Phi^Y_{X_{i}}\| e^{\lambda |X_{i}| + \mu \operatorname{diam}(X_{i})}}{|S_{i-1}|} \right) \cdot \frac{\|  \Phi_{X_{j}}\| e^{\lambda |X_{j}| + \mu \operatorname{diam}(X_{j})}}{|S_{j-1}|} e^{- \mu \operatorname{dist}(Z, V \setminus Y)} \, .
\end{multline}
Inserting the last expression in Eq.\  \eqref{equa:timeEvolutionAux7}, and using again Eq.\ \eqref{equa:timeEvolutionAux5} considering the fact that $\| \Phi^{Y}\| \leq \| \Phi\|$, we show that
\[
\sum_{X_{1} \cap Z \neq \emptyset} \ldots \sum_{X_{j-1} \cap S_{j-2} \neq \emptyset}\sum_{X_{j} \cap S_{j-1} \neq \emptyset }  \left( \prod_{\substack{i=1}}^{j-1} \frac{\| \Phi^Y_{X_{i}}\| e^{\lambda |X_{i}| + \mu \operatorname{diam}(X_{i})}}{|S_{i-1}|} \right) \cdot \frac{\|  \Phi_{X_{j}}\| e^{\lambda |X_{j}| + \mu \operatorname{diam}(X_{j})}}{|S_{j-1}|} \leq \| \Phi\|^{j}\,.
\]
Thus, we deduce from Eq.\ \eqref{equa:timeEvolutionComparisonAux4} 
\begin{align*}
\| \delta_{H_{Y'}}^{m}(Q)  - \delta_{H_{Y}}^{m}(Q)  \| & \leq 2^{m} \| Q\| m \, \frac{m!}{\lambda^{m}} \, \| \Phi\|^{m}  \, e^{\lambda |Z|} \,  e^{-\mu \operatorname{dist}(Z, V \setminus Y)}\,.
\end{align*}
Replacing this estimate in Eq.\ \eqref{equa:timeEvolutionComparisonAux3}, we conclude that
\[ \left\| \Gamma_{H_{Y'}}^{s}(Q) - \Gamma_{H_{Y}}^{s}(Q) \right\| \leq \| Q\| \,e^{\lambda |Z|} \,  \sum_{m=1}^{\infty} m \left( \frac{2\| \Phi\| \, |s|}{\lambda}\right)^{m}  e^{-\mu \operatorname{dist}(Z, V \setminus Y)} \,.\]
Finally, we apply the formula $x/(1-x)^{2} = \sum_{m=1}^{\infty}mx^{m}$ with $x = 2 \| \Phi\|\, |s|/ \lambda < 1$ to get the desired result.
\end{proof}

\subsection{Araki's expansionals} 

In this subsection, we present  some estimates on Araki's expansionals \cite{Araki1969} for a Hamiltonian with short-range interactions. We use the following notation for the expansionals:

\begin{equation*}
\begin{aligned}
        E_{X,Y}(s)  := & \operatorname{e}^{- s H_{XY}}\operatorname{e}^{s H_{X} + s H_Y} \,  \quad \text{ for every } s \in \mathbb{C} \, . \\
        E_{X,Y}  :=&  E_{X,Y}(1) = \operatorname{e}^{-  H_{XY}}\operatorname{e}^{ H_{X} +  H_Y} \, .
\end{aligned}
\end{equation*}
 \noindent First let us recall that for every pair of observables $H$ and $W$ and every real value $\beta \geq 0$ we have the following expansion in terms of the time-evolution operator (see \cite[Eq.\ (5.1)]{Araki1969})

\[
e^{\beta(H+W)}e^{-\beta H} = \sum_{m=0}^{\infty} \,\, \int_{0}^{\beta}d t_{1} \int_{0}^{t_{1}}dt_{2} \ldots \int_{0}^{t_{m-1}} dt_{m} \,\, \prod_{j}^{m \to 1} \Gamma_{H}^{-it_{j}}(W) \, ,
\]
 where we are denoting $\underset{j}{\overset{m \to n}{\prod}}Q_{j} := Q_{m} Q_{m-1} \ldots Q_{n}$ for every $m \geq n$, and where we recall that $\Gamma_H^{-it}(W) = e^{tH}We^{-tH}$. Changing the signs $H \mapsto -H$ and $W \mapsto - W$ we can then rewrite
\begin{equation} \label{eq:Duhamel_formula}
 e^{-\beta(H+W)}e^{\beta H} = \sum_{m=0}^{\infty} (-1)^{m}\,\, \int_{0}^{\beta}d t_{1} \int_{0}^{t_{1}}dt_{2} \ldots \int_{0}^{t_{m-1}} dt_{m} \,\, \prod_{j}^{m \to 1} \Gamma_{H}^{it_{j}}(W) \,.
\end{equation}

\begin{prop}\label{prop:estimates_expansionals_normal}
Let $A,B,C$ be disjoint finite subsets of $V$ and let $\Phi$ be a local interaction on $V$ satisfying for some $\lambda, \mu >0$ 
\begin{equation*} 
\| \Phi\|=\| \Phi\|_{\lambda, \mu}:=\sup_{x \in V} \sum_{X \ni x} \| \Phi_{X}\| e^{\lambda |X| + \mu \operatorname{diam}(X)} < \infty   \,, 
\end{equation*} 
Then, for every real number $\beta$ with $|\beta|<  \frac{\lambda}{2 \| \Phi\|}$ we have
\begin{equation}\label{eq:estimate_expansional_normal}
    \norm{E_{A,B}(\beta)} \leq \exp{\frac{\| \Phi\|\, |\beta| \lambda}{\lambda-2 \| \Phi\| \, |\beta| }   \, \textstyle\sum_{v \in A} e^{-\mu \operatorname{dist}(v,B)}  }   \, ,
\end{equation}
and
\begin{equation}\label{eq:estimate_differences_expansional_normal}
   \textstyle \norm{E_{A,BC}(\beta) - E_{A,B}(\beta)} \leq \exp{ \frac{ \| \Phi\|\, |\beta| \lambda}{\lambda - 2 \| \Phi\|\, |\beta|} \sum_{v \in A} e^{-\mu \operatorname{dist}(v,BC)}} \cdot \frac{|\beta| \, \| \Phi\|^2 (\lambda + |\beta|)^{2} }{(\lambda - 2 \| \Phi\| |\beta|)^{2}} \sum_{v \in A} e^{- \mu \operatorname{dist}(v,C)}\,.
\end{equation}
\end{prop}

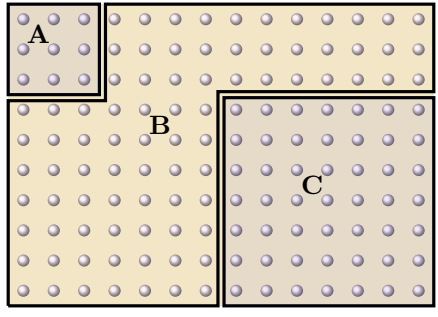
\begin{figure}[ht]
\begin{center}

\begin{tikzpicture}[scale=0.4]


\fill [darkyellow!50!orange!25!white] (-0.5,-0.5) rectangle (13.5,9.5);


\foreach \n in {0,1,...,13}{
\foreach \m in {0,1,...,9}{
\shade[shading=ball, ball color=darkred!10!white] (\n,\m) circle (0.2);
}
 }


\draw [black, very thick] (-0.5,6.5) rectangle (2.5,9.5);
\fill [blue, opacity=0.05] (-0.5,6.5) rectangle (2.5,9.5);
\node at (0.5,8.5) {\small $\textbf A$};


\draw [black, very thick] (-0.5,-0.5) -- (6.4,-0.5) -- 
(6.4, 6.6) -- (13.5, 6.6) --
(13.5, 9.5) -- (2.7,9.5) --
(2.7,6.3) -- (-0.5,6.3) --
(-0.5,-0.5);
\node at (4.5,5.5) {\small  $\textbf B$};


\draw [black, very thick] (6.6,-0.5) rectangle (13.5,6.4);
\fill [blue, opacity=0.05] (6.6,-0.5) rectangle (13.5,6.4);
\node at (9.5,3.5) {\small $\textbf C$};

\end{tikzpicture}

 \caption{Example of configuration of the three disjoint regions $A,B,C$ in Proposition \ref{prop:estimates_expansionals_normal}}
  \label{fig:three_disjoint_regions}
  \end{center}
\end{figure}

\begin{rem}
    In fact, the exponential growth in $\beta$ is unavoidable, as one can see from the commutative case \cite{kochanowski2023MLSI}.
\end{rem}

{
\begin{rem} \label{rem:cupcake}
    In Eq. \eqref{eq:estimate_expansional_normal} we could have also written $\sum_{v \in B} e^{- \mu \operatorname{dist}(v,A)}$, since $E_{A,B}(\beta) = E_{B,A}(\beta)$. Thus, a sharper upper bound would be to take the minimum of both quantities.
\end{rem}

\begin{rem} \label{rem:onion}
In both inequalities \eqref{eq:estimate_expansional_normal} and \eqref{eq:estimate_differences_expansional_normal}, we find expressions of the form $\sum_{v \in X} e^{- \mu \operatorname{dist}(v,Y)}$, that can be understood as a way of measuring the size of the boundary of $X$ with respect to $Y$. Indeed, let us consider the case $V=\mathbb{Z}^{g}$ with the distance induced by any of the $\| \cdot\|_{p}$ norm ($1 \leq p \leq \infty$). Given two finite and disjoint subsets $A$ and $X$ of $V$, we can estimate for example 
\[ \sum_{v \in A} e^{- \mu \operatorname{dist}(v, X)} = \sum_{k=1}^{\infty} |\{ v \in A \colon k-1 < \operatorname{dist}(v, X) \leq k \}| e^{-\mu k} \,. \]
On the one hand, the summands corresponding to $k < \operatorname{dist}(A,X)$ are equal to zero, since no element $v \in A$ satisfies $\operatorname{dist}(v, X) \leq k$ in this case. Thus, using that $X \subset A^{c}$, we can estimate
\[ \sum_{v \in A} e^{- \mu \operatorname{dist}(v, X)} \leq \sum_{k \geq  \operatorname{dist}(A,X) } |\{ v \in A \colon \operatorname{dist}(v, A^{c}) \leq k \}| e^{-\mu k} \,. \]
On the other hand, this space has the property that if $v \in A$ satisfies $\operatorname{dist}(v,A^{c}) \leq k$, then the open ball $B_{d}(v,k)$ centered at $v$ with radius $k$ intersects $\partial A$, i.e. $B_{d}(v,k) \cap \partial A \neq \emptyset$. Therefore,
\[ \{ v \in A \colon \operatorname{dist}(v, A^{c}) \leq k \} \subset \cup_{u \in \partial A} B_{d}(u,k)\,, \]
and so we can upper estimate
\[ |\{ v \in A \colon \operatorname{dist}(v, A^{c}) = k\}| \leq |\partial A| \sup_{v} |B_{d}(v,k)| \leq |\partial A| (2k+1)^{g}\,. \]
We then conclude that there is a constant $\nu =  \sup_{k \in \mathbb{N}} (2k+1)^{g} e^{-\mu k /2} (\sum_{j \geq 0} e^{-\mu j/2})$ depending on $\mu$ such that
\[ \sum_{v \in A} e^{- \mu \operatorname{dist}(v, X)} \leq |\partial A| \, \sum_{k \geq  \operatorname{dist}(A,X) } (2k+1)^{g}  e^{-\mu k} \leq |\partial A| \nu e^{-(\mu/2)  \operatorname{dist}(A,X) } \,. \]
\noindent As a consequence, we can simplify the estimates in Eq. \eqref{eq:estimate_expansional_normal} by 
\begin{equation}\label{eq:simplified_estimates_expansionals}
    \norm{E_{A,B}(\beta)}  \leq \exp \{ |\beta| K  \min \{ |\partial A|, |\partial B| \}   \} \, ,
\end{equation}
where we used Remark \ref{rem:cupcake}, and 
\begin{equation}\label{eq:simplified_estimates_difference_expansionals}
  \norm{E_{A,BC}(\beta)-E_{A,B}(\beta)}  \leq \exp \{ |\beta| K |\partial A|    \} K ' |\partial A| e^{- (\mu/2) \operatorname{dist}(A,C) }   \, ,
\end{equation}
 for certain constants $K=K(\lambda, \mu, \| \Phi\|, \beta)$ and $K'=K'(\lambda, \mu, \| \Phi\|, \beta)$ depending on $\lambda, \mu, \| \Phi\|, \beta$.
\end{rem}
}

\begin{proof}[Proof of \Cref{prop:estimates_expansionals_normal}]
 
We can restrict the proof of the theorem to values $\beta \geq 0$, since for values $\beta <0$ we can rewrite
\[
E_{X,Y}(\beta) = e^{-\beta H_{XY}}e^{\beta H_{X} + \beta H_{Y}} = e^{ -|\beta| \hat{H}_{XY}} e^{|\beta| \hat{H}_{X} + |\beta | \hat{H}_{Y}}\,,
\]
where $\hat{H}$ is the Hamiltonian associated to the new interaction  \mbox{$\hat{\Phi} = - \Phi$}, which satisfies $\| \hat{\Phi}\| = \|\Phi\|$. Let us start with an observation that will be useful at several points of the proof. For an arbitrary pair of disjoint subsets $A$,$B \in \mathcal P_f(V)$ we can estimate
\begin{align}
\nonumber \underset{ \hbox{\tiny${\begin{array}{c}
         Z \subset V  \\
         Z \cap A \neq \emptyset \\
         Z \cap B \neq \emptyset
    \end{array}}$}}{\sum } \norm{ \Phi_Z} \operatorname{e}^{\lambda \abs{Z}} 
    &\nonumber \leq \sum_{v \in A} \underset{ \hbox{\tiny ${\begin{array}{c}
         Z \ni v  \\
         Z \cap B \neq \emptyset
    \end{array}}$}}{\sum } \norm{ \Phi_Z} \operatorname{e}^{\lambda \abs{Z}}\\
    &\nonumber = \sum_{v \in A} \underset{ \hbox{\tiny${\begin{array}{c}
         Z \ni v  \\
         Z \cap B \neq \emptyset
    \end{array}}$}}{\sum } \norm{ \Phi_Z} \operatorname{e}^{\lambda \abs{Z}} \operatorname{e}^{\mu \, \text{diam}(Z)} \operatorname{e}^{-\mu \, \text{diam}(Z)} \\
    &\nonumber \leq \sum_{v \in A} \underset{ \hbox{\tiny${\begin{array}{c}
         Z \ni v  \\
         Z \cap B \neq \emptyset
    \end{array}}$}}{\sum } \norm{ \Phi_Z} \operatorname{e}^{\lambda \abs{Z}} \operatorname{e}^{\mu \, \text{diam}(Z)} \operatorname{e}^{-\mu \, \text{dist}(v,B)} \\
    & \label{eq:estimates_expansionals_normal_Aux_1}  \leq \| \Phi\| \sum_{v \in A} e^{-\mu \operatorname{dist}(v,B)}
\end{align}

\noindent To prove Eq.\ \eqref{eq:estimate_expansional_normal}, let us consider $H\equiv  H_{A} + H_{B}$ and $W \equiv W_{A,B}=  H_{AB} -   H_A -  H_B$ in Eq.\ \eqref{eq:Duhamel_formula}, yielding for every $\beta \geq 0$: 

\begin{align*}
    E_{A,B}(\beta) & = \sum_{m=0}^{\infty} (-1)^{m} \,\, \int_{0}^{\beta}dt_{1} \int_{0}^{t_{1}}dt_{2} \ldots \int_{0}^{t_{m-1}} dt_{m} \,\, \prod_{j}^{m \to 1} \Gamma_{H_{A} + H_{B}}^{it_{j}}(W_{A,B})  \, .
\end{align*}
Therefore, we can estimate 
\begin{equation}\label{eq:estimate_prod_exponentials_normal_minus_identity}
   \left\|  E_{A,B}(\beta)  \right\| \leq 1 + \sum_{m=1}^{\infty} \frac{\beta^{m}}{m!} \, \left(\sup_{0 \leq t \leq \beta}\| \Gamma^{it}_{H_A + H_B}(W_{A,B})\|\right)^{m} \, .
\end{equation}
Now, let us recall that
\begin{equation}\label{eq:estimate_expansional_normal_Aux2}
    \Gamma^{it}_{H_A + H_B}(W_{A,B}) = \underset{ \hbox{\tiny${\begin{array}{c}
         Z \subset AB  \\
         Z \cap A \neq \emptyset \\
         Z \cap B \neq \emptyset
    \end{array}}$}}{\sum }\Gamma^{it}_{H_A + H_B}(\Phi_Z)   \, .
\end{equation}
Then, using Proposition \ref{theo:localityEstimates} and the fact that $0 \leq t \leq \beta <\lambda /(2\|\Phi\|)$,
\begin{equation*}
    \norm{  \Gamma^{it}_{H_A + H_B}(W_{A,B}) }  \leq \underset{ \hbox{\tiny${\begin{array}{c}
         Z \subset AB  \\
         Z \cap A \neq \emptyset \\
         Z \cap B \neq \emptyset
    \end{array}}$}}{\sum }\norm{\Gamma^{it}_{H_A + H_B}(\Phi_Z) } \leq \frac{\lambda}{\lambda - 2 \| \Phi\| \beta} \underset{ \hbox{\tiny${\begin{array}{c}
         Z \subset AB  \\
         Z \cap A \neq \emptyset \\
         Z \cap B \neq \emptyset
    \end{array}}$}}{\sum } \norm{ \Phi_Z} \operatorname{e}^{\lambda \abs{Z}}  \, .
\end{equation*} 
At this point, we can make use of the observation Eq.\ \eqref{eq:estimates_expansionals_normal_Aux_1}  to obtain  the upper bound
\begin{equation}\label{eq:estimate_expansional_normal_Aux3}
\norm{  \Gamma^{it}_{H_A + H_B}(W_{A,B}) } \leq \frac{\lambda \| \Phi\|}{\lambda - 2 \| \Phi\| \beta} \,  \sum_{v \in A} e^{-\mu \operatorname{dist}(v,B)} \,. 
\end{equation}
Finally, applying this upper estimate to Eq.\ \eqref{eq:estimate_prod_exponentials_normal_minus_identity}, we conclude that Eq.\  \eqref{eq:estimate_expansional_normal} holds. 

Let us prove now Eq.\ \eqref{eq:estimate_differences_expansional_normal} following similar ideas. First, note that
\begin{align*}
      & E_{A,B}(\beta) - E_{A,BC}(\beta)  = \operatorname{e}^{-\beta H_{AB} }\operatorname{e}^{  \beta H_A  + \beta H_B }  - \operatorname{e}^{-\beta H_{ABC} }\operatorname{e}^{ \beta H_A  + \beta H_{BC} } \\
      & =  \sum_{m=1}^{\infty} \,\, \int_{0}^{\beta}dt_{1} \int_{0}^{t_{1}}dt_{2} \ldots \int_{0}^{t_{m-1}} dt_{m} \,\, \left[\underset{j}{\overset{m \to 1}{\prod}} \Gamma_{H_A + H_B}^{it_{j}}(W_{A,B}) - \underset{j}{\overset{m \to 1}{\prod}} \Gamma_{H_A + H_{BC} }^{it_{j}}(W_{A,BC}) \right] 
\end{align*}
Then,
\begin{equation}\label{equa:differenceBound}
    \norm{E_{A,B}(\beta) - E_{A,BC}(\beta)}\leq  \sum_{m=1}^{\infty}  \frac{\beta^m}{m!} \underset{|t_m| , \ldots , |t_1|\leq \beta}{\text{sup}} \norm{\underset{j}{\overset{m \to 1}{\prod}} \Gamma_{H_A + H_B}^{it_{j}}(W_{A,B}) - \underset{j}{\overset{m \to 1}{\prod}} \Gamma_{H_A + H_{BC} }^{it_{j}}(W_{A,BC}) } \, .
\end{equation}
Moreover, note that 
\begin{align*}
    & \norm{\underset{j}{\overset{m \to 1}{\prod}} \Gamma_{H_A + H_B}^{it_{j}}(W_{A,B}) - \underset{j}{\overset{m \to 1}{\prod}}  \Gamma_{H_A + H_{BC} }^{it_{j}}(W_{A,BC}) } \nonumber \\
    & = \norm{ \underset{\alpha=1}{\overset{m}{\sum}} \,  \underset{j}{\overset{m \to (\alpha + 1)}{\prod}}  \Gamma_{H_A + H_B}^{it_{j}}(W_{A,B})  \left(  \Gamma_{H_A + H_B}^{it_{\alpha}}(W_{A,B})- \Gamma_{H_A + H_{BC}}^{it_{\alpha}}(W_{A,BC})  \right) \underset{j}{\overset{(\alpha -1) \to 1}{\prod}}  \Gamma_{H_A + H_{BC}}^{it_{j}}(W_{A,BC})  } \nonumber \\
    & \leq \sum_{\alpha = 1}^{m} \underbrace{\norm{\Gamma_{H_A + H_B}^{it_{\alpha}}(W_{A,B})- \Gamma_{H_A + H_{BC}}^{it_{\alpha}}(W_{A,BC})}}_{(I)}
    \prod_{j=\alpha + 1}^{m} \underbrace{\norm{\Gamma_{H_A + H_B}^{it_{j}}(W_{A,B})}}_{(II)} 
    \prod_{j=1}^{\alpha-1} \underbrace{\norm{\Gamma_{H_A + H_{BC}}^{it_{j}}(W_{A,BC})}}_{(III)}\,.
\end{align*}
We have classified the factors on the previous expression into three types (I), (II) and (III). Factors of type (II) can be upper estimate using Eq.\ \eqref{eq:estimate_expansional_normal_Aux3}. The same estimate can be applied to factors of type (III) adapted to the pair $A$ and $BC$ instead of $A$ and $B$. But since $\operatorname{dist}(v,B) \geq \operatorname{dist}(v,BC)$ for every $v \in A$, we can actually use the following common upper bound for both type of factors: 
\begin{equation}\label{eq:estimate_expansional_normal_Aux4}
\norm{\Gamma_{H_A + H_B}^{it}(W_{A,B})} , \norm{\Gamma_{H_A + H_{BC}}^{it}(W_{A,BC})} \leq \frac{\lambda \| \Phi\|}{\lambda - 2 \| \Phi\| \beta}  \sum_{v \in A} e^{-\mu \operatorname{dist}(v,BC)} 
\end{equation} 
To deal with the factor of type (I), let us split
\begin{align*}
    & \norm{\Gamma_{H_A + H_B}^{it}(W_{A,B})- \Gamma_{H_A + H_{BC}}^{it}(W_{A,BC})  }  \\
    & \leq \underbrace{\underset{ \hbox{\tiny${\begin{array}{c}
         Z \subset AB  \\
         Z \cap A \neq \emptyset \\
         Z \cap B \neq \emptyset
    \end{array}}$}}{\sum } \norm{  \Gamma_{H_A + H_B}^{it}(\Phi_Z)- \Gamma_{H_A + H_{BC}}^{it}(\Phi_Z)  }}_{(III.1)} + \underbrace{\underset{ \hbox{\tiny${\begin{array}{c}
         Z \subset ABC  \\
         Z \cap A \neq \emptyset \\
         Z \cap C \neq \emptyset 
    \end{array}}$}}{\sum } \norm{ \Gamma_{H_A + H_{BC}}^{it}(\Phi_Z) }}_{(III.2)} \, .
\end{align*}
For the second sum (III.2), we again use 
\begin{align}
\nonumber \underset{ \hbox{\tiny${\begin{array}{c}
         Z \subset ABC  \\
         Z \cap A \neq \emptyset \\
         Z \cap C \neq \emptyset
    \end{array}}$}}{\sum }\norm{\Gamma^{it}_{H_A + H_{BC}}(\Phi_Z) }
    & \label{eq:estimates_expansionals_normal_Aux_8} \leq \frac{\lambda \| \Phi\|}{\lambda - 2 \| \Phi\| \beta}
    \underset{ \hbox{\tiny${\begin{array}{c}
         Z \subset ABC  \\
         Z \cap A \neq \emptyset \\
         Z \cap C \neq \emptyset
    \end{array}}$}}{\sum}
    \norm{\Phi_{Z}} e^{\lambda |Z|}\\
    & \leq  \frac{\lambda \|\Phi\|^{2}}{\lambda - 2 \| \Phi\| \beta}  \sum_{v \in A} e^{- \mu \operatorname{dist}(v,C)}\,,
\end{align}
where in the last inequality we have again used Eq.\ \eqref{eq:estimates_expansionals_normal_Aux_1} for the pair $A$ and $C$. For the first sum (III.1), however, we are going to use inequality Eq.\ \eqref{equa:time_evolution_main_2} from Proposition \ref{theo:localityEstimates} with the interactions $\Phi^{A,B}$ and $\Phi^{A,BC}$ on $V = ABC$ giving $H_{A} + H_{B}$ and $H_{A}+H_{BC}$, respectively, and $Y=AB$. Note that $\|\Phi^{A,B}\|, \| \Phi^{A,BC}\| \leq \|\Phi\|$ since both interactions coincide with $\Phi$ or are zero on every subset $X \in \mathcal{P}_{f}$. Then
\begin{align*}
\underset{ \hbox{\tiny${\begin{array}{c}
         Z \subset AB  \\
         Z \cap A \neq \emptyset \\
         Z \cap B \neq \emptyset
    \end{array}}$}}{\sum } \norm{  \Gamma_{H_A + H_B}^{it}(\Phi_Z)- \Gamma_{H_A + H_{BC}}^{it}(\Phi_Z)  } 
     &\leq 
    \frac{2 \| \Phi\| \beta \lambda}{(\lambda - 2 \| \Phi\| \beta)^{2}}
     \underset{ \hbox{\tiny${\begin{array}{c}
         Z \subset AB  \\
         Z \cap A \neq \emptyset \\
         Z \cap B \neq \emptyset
    \end{array}}$}}{\sum } 
    \| \Phi_{Z}\| e^{\lambda |Z|} e^{- \mu \operatorname{dist}(Z, C)}\,,\\
\end{align*}
where the last sum can be estimated by
\begin{align*} 
\underset{ \hbox{\tiny${\begin{array}{c}
         Z \subset AB  \\
         Z \cap A \neq \emptyset \\
         Z \cap B \neq \emptyset
    \end{array}}$}}{\sum } 
    \| \Phi_{Z}\| e^{\lambda |Z|} e^{- \mu \operatorname{dist}(Z,C)} 
    & = 
    \underset{ \hbox{\tiny${\begin{array}{c}
         Z \subset AB  \\
         Z \cap A \neq \emptyset \\
         Z \cap B \neq \emptyset
    \end{array}}$}}{\sum } 
    \| \Phi_{Z}\| e^{\lambda |Z| + \mu \operatorname{diam}(Z)} e^{- \mu(\operatorname{diam}(Z)+ \operatorname{dist}(Z,C))} \\
    & \leq 
    \sum_{v \in A} \sum_{Z \ni v}
    \| \Phi_{Z}\| e^{\lambda |Z| + \mu \operatorname{diam}(Z)} e^{- \mu(\operatorname{diam}(Z)+ \operatorname{dist}(Z,C))} \\
    & \leq 
    \sum_{v \in A} \sum_{Z \ni v}
    \| \Phi_{Z}\| e^{\lambda |Z| + \mu \operatorname{diam}(Z)} e^{- \mu \operatorname{dist}(v,C)} \\
    & \leq \| \Phi\| \sum_{v \in A} e^{- \mu \operatorname{dist}(v,C)}\,.
\end{align*}
Here, we have used that $\operatorname{diam}(Z) + \operatorname{dist}(Z,C) \geq \operatorname{dist}(v,C)$ for any $v \in Z$, which holds by the triangle inequality. Combining Eq.\ \eqref{eq:estimates_expansionals_normal_Aux_8} with the previous bounds, we obtain the following bound for (I)
\begin{align*} 
\norm{\Gamma_{H_A + H_B}^{it}(W_{A,B})- \Gamma_{H_A + H_{BC}}^{it}(W_{A,BC})  } 
& \leq   \left(\frac{2 \| \Phi\|^2 \beta \lambda}{(\lambda - 2 \| \Phi\| \beta)^{2}} + \frac{\lambda \| \Phi\|^{2}}{\lambda - 2 \| \Phi\| \beta} \right) \sum_{v \in A} e^{- \mu \operatorname{dist}(v,C)}\\
& = \frac{2 \| \Phi\|^2 \beta \lambda + \lambda^{2}\| \Phi\|^2 -2 \lambda \| \Phi\|^{3} \beta}{(\lambda - 2 \| \Phi\| \beta)^{2}} \sum_{v \in A} e^{- \mu \operatorname{dist}(v,C)}\\
& \leq \frac{\| \Phi\|^2 (\lambda + \beta)^{2} }{(\lambda - 2 \| \Phi\| \beta)^{2}} \sum_{v \in A} e^{- \mu \operatorname{dist}(v,C)}\,.
\end{align*}
Combining the upper bounds for (I), (II) and (III) we conclude that
\begin{multline}
 \norm{\underset{j}{\overset{m \to 1}{\prod}} \Gamma_{H_A + H_B}^{it_{j}}(W_{A,B}) - \underset{j}{\overset{m \to 1}{\prod}}  \Gamma_{H_A + H_{BC} }^{it_{j}}(W_{A,BC}) }  \\
\leq   m \left(\frac{\lambda \| \Phi\|}{\lambda - 2 \| \Phi\| \beta}  \sum_{v \in A} e^{-\mu \operatorname{dist}(v,BC)}  \right)^{m-1}  \frac{\| \Phi\|^2 (\lambda + \beta)^{2} }{(\lambda - 2 \| \Phi\| \beta)^{2}} \sum_{v \in A} e^{- \mu \operatorname{dist}(v,C)}\,.
\end{multline}
Inserting this expression in Eq.\ \eqref{equa:differenceBound} we conclude that
\begin{align*}
\norm{E_{A,B}(\beta) - E_{A,BC}(\beta)} 
& \leq 
\sum_{m=1}^{\infty} \frac{\beta^{m}}{(m-1)!}
\left(\frac{\lambda \| \Phi\|}{\lambda - 2 \| \Phi\| \beta}  \sum_{v \in A} e^{-\mu \operatorname{dist}(v,BC)}  \right)^{m-1}   \cdot \\
& \hspace{2cm} \cdot \frac{\| \Phi\|^2 (\lambda + \beta)^{2} }{(\lambda - 2 \| \Phi\| \beta)^{2}} \sum_{v \in A} e^{- \mu \operatorname{dist}(v,C)}\\
& = \exp{ \frac{\lambda \| \Phi\|\beta}{\lambda - 2 \| \Phi\| \beta}  \sum_{v \in A} e^{-\mu \operatorname{dist}(v,BC)}  }\\ 
& \hspace{2cm} \cdot \frac{\beta \| \Phi\|^2 (\lambda + \beta)^{2} }{(\lambda - 2 \| \Phi\| \beta)^{2}} \sum_{v \in A} e^{- \mu \operatorname{dist}(v,C)}\,.
\end{align*}
 This finishes the proof of the inequality.
\end{proof}

Based on this proposition, we can derive estimates for various expressions on the expansionals, in the spirit of, e.g., those from \cite[Corollary 4.4]{Bluhm2021exponential}. We only provide here the bounds required in the proof of the main result, Theorem \ref{thm:Weakimpliesmixingcondition}, but some other bounds would follow from Proposition \ref{prop:estimates_expansionals_normal} analogously. However, unlike in \cite{Bluhm2021exponential}, we can only recover bounds in which we take full traces of expansionals and states.

\begin{cor}\label{cor:estimate_expansional}
    Under the conditions of Proposition \ref{prop:estimates_expansionals_normal} and for $V=\mathbb{Z}^g$ endowed with the Euclidean distance, for $0 \leq \beta < \frac{\lambda}{2 \| \Phi\|}$,  and denoting by $\rho^{AB}_{\beta}$ the Gibbs state in $AB$ at inverse temperature $\beta$, we have 
    \begin{equation}
        \Big| \Tr_{AB}\Big[ \rho^{AB}_\beta E_{A,B}^{\ast \, -1}(\beta)\Big]^{-1}  \Big| \leq e^{{\beta} \, K \, \mathrm{min} \{ \abs{\partial A},  \abs{\partial B} \}  } \, . 
    \end{equation}
    where $K=K(\lambda, \mu, \| \Phi\|, \beta)$ is the constant from Eq. \eqref{eq:simplified_estimates_expansionals}.
\end{cor}

\begin{proof}
Note that we can write 
\begin{align}
\Tr_{AB} \Big[ \rho^{AB}_\beta E_{A,B}^{\ast \, -1}(\beta) \Big]& =\Tr_{AB} \left[ \rho^{AB}_\beta E_{A,B}^{\ast \, -1}\left(\frac{\beta}{2}\right) E_{A,B}^{ -1}\left(\frac{\beta}{2}\right) \right]  \, ,
\end{align}
where $E_{A,B}^{ -1}\left(\frac{\beta}{2}\right) = e^{\frac{-\beta (H_A+H_B)}{2}}e^{\frac{\beta H_{AB}}{2}}$ and $E_{A,B}^{* -1}\left(\frac{\beta}{2}\right) = e^{\frac{\beta H_{AB}}{2}}e^{\frac{-\beta (H_A+H_B)}{2}}$. Let us denote $Q:=  E_{A,B}^{\ast \, -1}\left(\frac{\beta}{2}\right) E_{A,B}^{ -1}\left(\frac{\beta}{2}\right)$. Since $Q$ is a positive and invertible operator, the following inequality holds:
\begin{equation}
    Q \geq \norm{Q^{-1}}^{-1} \mathds{1} \, .
\end{equation} 
Next, note that 
\begin{equation}
    \norm{Q^{-1}} \leq \norm{ E_{A,B}\left(\frac{\beta}{2}\right)E_{A,B}^{\ast }\left(\frac{\beta}{2}\right)} \leq \norm{ E_{A,B}^{\ast }\left(\frac{\beta}{2}\right)}^2 \leq \operatorname{e}^{{\beta} \, K \, \text{min} \{ \abs{\partial A},  \abs{\partial B} \}  } \, ,
\end{equation}
where we are using the estimates from \Cref{prop:estimates_expansionals_normal}, and specifically the simplification from Eq. \eqref{eq:simplified_estimates_expansionals}. Then,  
\begin{align}
    \Big| \Tr_{AB}\Big[ \rho^{AB}_\beta E_{A,B}^{\ast \, -1}(\beta)\Big]^{-1}  \Big| & \leq \norm{E_{A,B}^{\ast }\left(\frac{\beta}{2}\right) E_{A,B}\left(\frac{\beta}{2}\right)} \leq \operatorname{e}^{{\beta} \, K \, \text{min} \{ \abs{\partial A},  \abs{\partial B} \} } \, .
\end{align}
\end{proof}

\section{Local effective Hamiltonian}\label{sec:effective_hamiltonian}

Another tool we will need in order to prove our main result is the existence of an effective Hamiltonian. Let us depart from a quantum spin system defined on a (possibly infinite) metric space $(V, \operatorname{dist})$ and a local interaction $\Phi$. Given any two finite subsets $L \subset \Lambda$ of $V$ and some fixed (inverse) temperature $\beta >0$, we can consider the Hermitian operator given by
\[ \widetilde{H}^{L, \beta}_{\Lambda} := - \frac{1}{\beta}\log \mathbb E_L[e^{- \beta H_{\Lambda}}] \,. \]
This allows us to represent the normalized marginal  of the Gibbs state
$(\rho^{\Lambda}_{\beta})_{L}$ as the Gibbs state of this new (so-called \emph{effective}) Hamiltonian:
\[ e^{-\beta \widetilde{H}^{L, \beta}_{\Lambda}} = \mathbb E_L[e^{-\beta H_{\Lambda}}] \,. \]
One might expect that the inherent locality of the original Hamiltonian $H_{\Lambda}$ manifests in some form of locality for the new one. We may even speculate that if $\Phi$ possesses a strong decaying condition (e.g. finite or short range), then the local interactions defining the effective Hamiltonian should also have some strong form of decay (exponential or even faster). 

One initial exploratory avenue to seek evidence supporting our statements is within the ``(very) high-temperature regime'' where, traditionally, the locality of the interactions manifests as locality properties of the Gibbs state (e.g.\ decay of correlations) in very general settings, as it has been formally proved in quite a number of results. From a purely mathematical point of view, we can consider a complex variable instead of merely a positive temperature. 

Let us explain the broad idea behind this approach. One considers the complex vector-valued function
\[ z \mapsto \rho^{\Lambda}_{z} = e^{-z H_{\Lambda}}/\Tr_{\Lambda}[e^{-z H_{\Lambda}}]\,. \]
Since $\Lambda$ is finite, and therefore, $H_{\Lambda}$ is bounded, we can deduce the existence of an open disk around $z=0$ where the above function is well-defined and analytic. Moreover, using the analyticity properties around $z=0$, one can infer locality properties for small values of $\beta$. {This is the route taken in \cite{Kuwahara2019}, where cluster expansions are used to control the number of terms appearing in the expansion around $z=0$ and their norms.  Unfortunately, there is a flaw \cite{samuel-personal-communication, Kuwahara.2024} in this part of \cite{Kuwahara2019}, such that the status of the effective Hamiltonian derived in this paper is presently unclear. Therefore, we will take a different route in this paper.

In Section \ref{sec:prop-high-T}, we will follow the approach explained above to get an idea of what properties one might expect for an effective Hamiltonian, culminating in the definition of a local effective Hamiltonian in Section \ref{sec:loc-effective-H}. More specifically, we will examine the problem at \emph{very} high temperatures.  By this, we mean values of $\beta$ that are smaller than a constant $\beta_{c}$ depending on the size of $\Lambda$, opposing to simply high temperatures, where the constant $\beta_{c}$ does not depend on the size on $\Lambda$ but maybe on local properties of the underlying graph $V$. This subtle difference makes the latter a considerable harder problem that we will address in Section \ref{sec:comm-effective-Hamiltonian}. In Section \ref{sec:comm-effective-Hamiltonian} we show that such effective Hamiltonians exist in the case where all marginals of the interactions are commuting. From then on, we assume the existence of such a local effective Hamiltonian to prove the mixing condition.}

\subsection{Properties exhibited at very high temperatures} \label{sec:prop-high-T}

 Fix a finite subset $\Lambda$ of $V$. To quantify the decay of the local interaction $\Phi$ on $V$, we are going to introduce now a new class of functions. We will assume there is a function $\mathbf{b}: \mathcal{P}_{f}(V) \to [0,\infty)$ being \emph{subbaditive}, namely $\mathbf{b}(X \cup Y) \leq \mathbf{b}(X) + \mathbf{b}(Y)$ for every $X, Y \in \mathcal{P}_{f}(V)$, such that
 \[ \|\Phi\|_{\mathbf{b}} = \sum_{x \in V} \sum_{X \ni x} \|\Phi_{X}\| e^{\mathbf{b}(X)} < \infty\,. \]
Observe that the previous quantification of decay encompasses the case $\mathbf{b}(X) = \lambda |X| + \mu \operatorname{diam}(X)$ for any fixed constants $\lambda, \mu \geq 0$.

Next, let us introduce for each $X \subset \Lambda$ a complex variable $z_{X} \in \mathbb{C}$, and consider the vector-valued holomorphic map 
\[ \mathbb{C}^{\mathcal{P}_{f}(\Lambda)} \to \mathfrak{A}_{\Lambda} \quad , \quad z= (z_{X})_{X} \mapsto H_{\Lambda}(z) = \sum_{X \subset \Lambda} z_{X} \Phi_{X}\,. \]
Then, for each (inverse) temperature $\beta>0$ and every subset $L \subset \Lambda$, the composite map 
\begin{equation}\label{equa:expectationExpansion1} 
\mathbb{C}^{\mathcal{P}_{f}(\Lambda)} \to \mathfrak{A}_{L} \quad , \quad z \mapsto \mathbb E_L[e^{-\beta H_{\Lambda}(z)}]\,, 
\end{equation}
defines again an entire function.  

Let us recall a few facts on holomorphic functions in several variables with values in a Banach space \cite{Mujica1986}. Given a \emph{multiradius}, that is, $\mathbf{r} = (r_{X})_{X \in \mathcal{P}_{f}(\Lambda)}$ with $r_{X} >0$ for every $X$, we define the open \emph{polydisc} with multiradius $r$ as
\[ r \mathbb{D}^{\mathcal{P}_{f}(\Lambda)} = \{ z \in \mathbb{C}^{\mathcal{P}_{f}(\Lambda)} \colon |z_{X}| < r_{X} \text{ for every } X \in \mathcal{P}_{f}(\Lambda) \}\,. \]
Analogously, one defines the closed polydisk $r \overline{\mathbb{D}}^{\mathcal{P}_{f}(\Lambda)}$ replacing the condition $|z_{X}| < r_{X}$ with $|z_{X}| \leq r_{X}$ for every $X \in \mathcal{P}_{f}(\Lambda)$. 

A holomorphic function on the open polydisc $f:\mathbb{C}^{\mathcal{P}_{f}(\Lambda)} \to \mathfrak{A}_{L}$ is characterized by the existence of a unique power (or monomial) series expansion
\[  f(z) = \sum_{\alpha: \mathcal{P}_{f}(\Lambda) \to \mathbb{N}_{0}} c_{\alpha} z^{\alpha} \quad \mbox{ where } \quad {z^{\alpha}} :=  \prod_{X \in \mathcal{P}_{f}(\Lambda)} z_{X}^{\alpha_{X}}\,,  \]
that is absolutely convergent on every closed polydisc $s \overline{\mathbb{D}}^{\mathcal{P}_{f}(\Lambda)}$ with $0\leq s_{X} < r_{X}$ for every $X \in \mathcal{P}_{f}(\Lambda)$, see \cite[Corollary 7.8]{Mujica1986}. Moreover, the coefficients $c_{\alpha}$ can be computed using integral Cauchy formulas, and also in terms of the partial derivatives
\[ c_{\alpha} =  (\partial^{\alpha} f) (0) = \big(\textstyle \prod_{X \in \mathcal{P}_{f}(\Lambda)} \partial^{\alpha_{X}}_{z_{X}}f\big) (0) \,,\]
where we do not specify a specific order in the concatenated application of the partial derivatives since the value is independent of it. 

We will however use another notation for the power series expansion that has been already used in \cite{wild2023classical}. It consists of identifying each map $\alpha: \mathcal{P}_{f}(\Lambda) \to \mathbb{N}_{0}$ appearing in the power series expansion with the multiset $\mathbf{W} = \mathbf{W}_{\alpha}$ containing each $X \in \mathcal{P}_{f}(\Lambda)$ a number $\alpha_{X}$ of times. Thus we can rewrite the power series expansion of $f$ as
\[ f(z) = \sum_{\mathbf{W}} c_{\mathbf{W}} z^{\mathbf{W}}  \quad \mbox{ where } \quad z^{\mathbf{W}}:= \prod_{X \in \mathbf{W}}  z_{X}\,, \]
and each coefficient as
\[ c_{\mathbf{W}} = D_{\mathbf{W}}|_{z=0}f(z) = \big( \prod_{X \in \mathbf{W}} \partial_{X} f\big)(0)\,. \]

The function defined in \eqref{equa:expectationExpansion1} thus admits a power series expansion that is absolutely convergent on every polydisc. It can be computed explicitly by expanding the exponential term:
\[ \mathbb E_L[e^{-\beta H_{\Lambda}(z)}] = \sum_{m=0}^{\infty} \frac{(-\beta)^m}{m!} \sum_{X_{1},  \ldots ,X_{m} \in \mathcal{P}_{f}(\Lambda)} \mathbb{E}_{L}[\Phi_{X_{1}} \cdot \ldots \cdot \Phi_{X_{m}}] \prod_{i=1}^{m} z_{X_{i}} =\sum_{\mathbf{W}} c_{\mathbf{W}} z^{\mathbf{W}}\,, \]
where  $c_{\mathbf{W}} = \frac{(-\beta)^{m}}{m!} \sum_{X_{1}, \ldots, X_{m} \in \mathcal{P}_{f}(\Lambda) \colon [X_{1}, \ldots, X_{m}] = \mathbf{W}} \mathbb{E}_{L}[\Phi_{X_{1}} \cdot \ldots \cdot \Phi_{X_{m}}]$ for each multiset $\mathbf{W}$ with $m$ elements. Note that the summands in the previous expression may be different from each other as the local interaction is not necessarily commuting.

Recall that a sufficient condition for the existence of a holomorphic logarithm of a given holomorphic function on an open domain  $f: \Omega \to \mathfrak{A}$ , is that $\|f(z) - \mathbbm{1} \| < 1$ for every $z \in \Omega$. In the case of \eqref{equa:expectationExpansion1}, the above expansion yields that taking the multiradius $\mathbf{r} = e^{\mathbf{b}}$, for every $z \in e^{\mathbf{b}} \overline{\mathbb D}^{\mathcal{P}_{f}(\Lambda)}$
\[ \| \mathbb E_L[e^{-\beta H_{\Lambda}(z)}]  - \mathbbm{1}\|  \leq \sum_{m=1}^{\infty} \frac{\beta^{m}}{m!} \sum_{X_{1}, \ldots, X_{m} \in \mathcal{P}_{f}(\Lambda)} \,\,\prod_{i=1}^{m} \| \Phi_{X_{i}}\| e^{\mathbf{b}(X_{i})} \leq e^{|\beta| \| \Phi\|_{\mathbf{a}} |\Lambda|} - 1\,.  \]
Therefore, if we restrict to values $0 < \beta < \frac{\log{(2)}}{\|\Phi\|_{\mathbf{b}} |\Lambda|}$, the previous norm is smaller than one, and therefore  we have a holomorphic function with power series expansion
\begin{equation} \label{eq:power-series-expansion}
 -\frac{1}{\beta}\log \mathbb E_L[e^{-\beta H_{\Lambda}(z)}] = \sum_{\mathbf{W}} a_{\mathbf{W}}(\beta) z^{\mathbf{W}}\,
 \end{equation}
that is absolutely convergent on the polydisc  $z \in e^{\mathbf{b}} \overline{\mathbb D}^{\mathcal{P}_{f}(\Lambda)}$, and whose coefficientes are given by
\[ a_{\mathbf{W}}(\beta)=-\frac{1}{\beta} D_{\mathbf{W}}|_{z=0} \log \mathbb E_L[e^{-\beta H_{\Lambda}(z)}] = -\frac{1}{\beta} D_{\mathbf{W}}|_{z=0} \log \mathbb E_L[e^{-\beta H_{\Lambda, \mathbf{W}}(z)}]\,. \]
Here, $H_{\Lambda,\mathbf{W}}(z)$ corresponds to $H_{\Lambda}(z)$ with the variables corresponding to sets $X \notin \mathbf{W}$ particularized to $z_{X} = 0$. Note, however, that we can omit the subindex $\Lambda$ since if $\Lambda'$ is another finite set containing $\cup \mathbf{W}$, then  $H_{\Lambda', \mathbf{W}}(z) = H_{\Lambda, \mathbf{W}}(z)$. Thus, we will simply write $H_{\mathbf{W}}(z)$.  

We can further simplify the sum in Eq.\ \eqref{eq:power-series-expansion} when $z$ is held constantly equal to one, which is the case we are most interested in because it recovers our Hamiltonian:
\begin{equation*} 
 -\frac{1}{\beta}\log \mathbb E_L[e^{-\beta H_{\Lambda}}] = \sum_{\mathbf{W}} a_{\mathbf{W}}(\beta) \,, 
 \end{equation*}
Next, we are going to rearrange the summands of the power series expansion in the following way: Let us define for each subset $X \in \mathcal{P}_{f}(\Lambda)$
\begin{equation}\label{equa:definingPhiSummand}
 \widetilde \Phi^{L, \beta}_{X}(z) = -\frac{1}{\beta} \sum_{\mathbf{W} \colon X = \cup \mathbf{W}} D_{\mathbf{W}}{|_{z=0}} \log \mathbb E_L[e^{-\beta H_{\mathbf{W}}(z)}] \cdot z^{\mathbf{W}}\,,
\end{equation}
where the sum is extended over all multisets $\mathbf{W}$ such that the union of its elements is equal to $X$. Note that absolute convergence ensures that this sum is well defined, at least on the polydisc $e^{\mathbf{b}} \overline{\mathbb D}^{\mathcal{P}_{f}(\Lambda)}$, and thus  we can rewrite
\[ -\frac{1}{\beta} \log \mathbb E_L[e^{-\beta H_{\Lambda}(z)}] = \sum_{X \in \mathcal{P}_{f}(\Lambda)} \widetilde \Phi_{X}^{L, \beta}(z)\,. \]
We can easily check the following properties for every $X \subset \Lambda$:
\begin{enumerate}
\item[(i)] $\widetilde \Phi^{L,\beta}_{X}(z)$ is supported in $X \cap L$: Indeed, for each $\mathbf{W}$ we have that $H_{\mathbf{W}}$ is supported in $\cup\mathbf{W}$, and thus
\[   D_{\mathbf{W}} \log \mathbb E_L[e^{-\beta H_{\mathbf{W}}(z)}] \]
is supported in $\cup\mathbf{W} \cap L$.
\item[(ii)] If $X \subset L$, then $\widetilde \Phi^{L, \beta}_{X}(z) = \Phi_{X}z_{X}$: Indeed, for every multiset $\mathbf{W}$ such that the union of its elements is equal to $X$, since $X \subset L$, we have
\[ \mathbb E_L[e^{-\beta H_{\mathbf{W}}(z)}] = e^{-\beta H_{\mathbf{W}}(z)}\,, \]
so that
\begin{align*}
D_{\mathbf{W}} \log \mathbb E_L[e^{-\beta H_{\Lambda}(z)}] 
& = D_{\mathbf{W}}  \log \mathbb E_L[e^{-\beta H_{ \mathbf W}(z)}]\\[2mm] 
& = D_{\mathbf{W}} \log e^{-\beta H_{\mathbf W}(z)}\\[2mm]
& =- \beta \, D_{\mathbf{W}}  H_{\mathbf W}(z)\,, 
\end{align*}
Now $D_{\mathbf{W}}|_{z=0}H_{\mathbf{W}}(z)$ is zero if the multiset $\mathbf{W}$ has cardinality greater than one (i.e. if it contains two different elements, or one element with multiplicity larger than two). Therefore, the only nonzero summand in Eq.\ \eqref{equa:definingPhiSummand} corresponds to $\mathbf{W}=[X]$, and $D_{\mathbf{W}}|_{z=0} H_{\mathbf{W}}(z) =  \Phi_{X}$.
\item[(iii)] If $L \subset L' \subset \Lambda$ and $X \cap (L' \setminus L) = \emptyset$, then ${\widetilde \Phi}^{L,\beta}_{X}(z) = {\widetilde \Phi}^{L',\beta}_{X}(z)$: Indeed, using that $(L^\prime \setminus L)^c = (L^\prime)^c \cup L$,  
\begin{align*} 
\mathbb E_L[e^{-\beta H_{\Lambda,\mathbf{W}}(z)}] & = \mathbb E_{L'} \, \mathbb E_{L \cup (L')^c}[e^{-\beta H_{\Lambda, \mathbf{W}}(z)}]\\[2mm]
&  = \mathbb E_{L'} \, \mathbb E_{ (L'\setminus L)^c}[e^{-\beta H_{\Lambda, \mathbf{W}}(z)}]\\[2mm] 
& =  \mathbb E_{L'}[e^{-\beta H_{\Lambda, \mathbf{W}}(z)}] \,.
\end{align*}
Therefore, the summands in Eq.\ \eqref{equa:definingPhiSummand} for both $\widetilde \Phi^{L, \beta}_X(z)$ and $\widetilde \Phi^{L', \beta}_X(z)$ coincide.
\end{enumerate}

\subsection{Locality of the effective Hamiltonian} \label{sec:loc-effective-H}

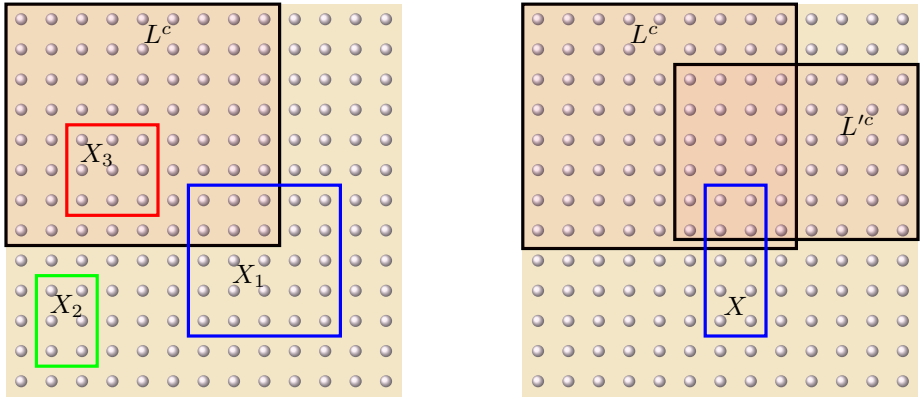
\begin{figure}[ht]
\begin{center}

\begin{tikzpicture}[scale=0.4]


\fill [darkyellow!50!orange!25!white] (-0.5,-0.5) rectangle (12.5,12.5);


\foreach \n in {1,...,13}{
\foreach \m in {1,...,13}{
\shade[shading=ball, ball color=darkred!10!white] (\n-1,\m-1) circle (0.2);
}
 }


\draw [black, very thick] (-0.5,4.5) rectangle (8.5,12.5);
\fill [red, opacity=0.05] (-0.5,4.5) rectangle (8.5,12.5);

\node at (4.5,11.5) {\small $L^{c}$};


\draw [blue, very thick] (5.5,1.5) rectangle (10.5,6.5);
\node at (7.5,3.5) {\small $X_{1}$};

\draw [green, very thick] (0.5,0.5) rectangle (2.5,3.5);
\node at (1.5,2.5) {\small $X_{2}$};

\draw [red, very thick] (1.5,5.5) rectangle (4.5,8.5);
\node at (2.5,7.5) {\small $X_{3}$};

\begin{scope}[xshift=17cm]

\fill [darkyellow!50!orange!25!white] (-0.5,-0.5) rectangle (12.5,12.5);


\foreach \n in {1,...,13}{
\foreach \m in {1,...,13}{
\shade[shading=ball, ball color=darkred!10!white] (\n-1,\m-1) circle (0.2);
}
 }
 

\draw [black, very thick] (-0.5,4.4) rectangle (8.5,12.5);
\fill [red, opacity=0.05] (-0.5,4.4) rectangle (8.5,12.5);

\node at (3.5,11.5) {\small $L^{c}$};


\draw [black, very thick] (4.5,4.7) rectangle (12.5,10.5);
\fill [red, opacity=0.05] (4.5,4.7) rectangle (12.5,10.5);
\node at (10.5,8.5) {\small $L'^{c}$};


\draw [blue, very thick] (5.5,1.5) rectangle (7.5,6.5);
\node at (6.5,2.5) {\small $X$};

\end{scope}

\end{tikzpicture}

 \caption{On the left picture, we represent three distinct dispositions of a subset $X$ with respect to the tracing region $L^{c}$. In particular, the local interactions of the effective Hamiltonian must satisfy $\widetilde{\Phi}^{L, \beta}_{X_{2}} = \Phi_{X_{2}}$, while $\widetilde{\Phi}^{L, \beta}_{X_{3}}$ is a multiple of the identity. On the right picture, the coincidence $X \cap L = X \cap L'$ yields that $\widetilde{\Phi}^{L, \beta}_{X} = \widetilde{\Phi}^{L', \beta}_{X}$.}
  \end{center}
\end{figure}

Taking inspiration from the properties that we have noticed in the previous section for extremely high temperatures, we introduce the following definition for the given quantum spin system on the metric space $V$ with local interaction $\Phi$.

\begin{defi}[Strong form]\label{defi:localityEffectiveHamiltonian}
Let us say that the above quantum spin system has (strong) \emph{local effective Hamiltonians  at (inverse) temperature $\beta >0$} if it satisfies the following property: for every  $L\subset V$, there exists a local interaction $\widetilde{\Phi}^{L,\beta}$ on $V$ satisfying
\begin{enumerate}
\item[(i)]  $\widetilde{\Phi}_{X}^{L, \beta}$ is supported in $X \cap L$ for every  $X \in \mathcal{P}_{f}(V)$.
\item[(ii)] If $L' \subset V$, then $\widetilde{\Phi}^{L, \beta}_{X} = \widetilde{\Phi}^{L', \beta}_{X}$ for all  $X \in \mathcal{P}_{f}(V)$ such that $X \cap L' = X \cap L$.
\item[(iii)] For every finite subset $\Lambda \subset V$
\begin{equation}
\label{equa:localityEffectiveHamiltonian}
\widetilde{H}^{L, \beta}_{\Lambda}:=\log \mathbb{E}_{L}[e^{-\beta H_{\Lambda}}] =  \sum_{X \subset \Lambda} \widetilde{\Phi}^{L, \beta}_{X}\,. 
\end{equation}
\end{enumerate}
We will say that $\widetilde{\Phi}^{L, \beta}$ is the local effective interaction of the marginals $(\rho^{\Lambda}_{\beta})_{L}$ on $L$.
\end{defi}

\begin{rem}
It is noteworthy that the same local interaction $\widetilde{\Phi}^{L, \beta}$ provides the family of Hamiltonians $\widetilde{H}_{\Lambda}^{L, \beta}$ for every $\Lambda \in \mathcal{P}_{f}(V)$. In particular, taking $L = V$ we have that the local interaction $\widetilde{\Phi}^{V, \beta}$ satisfies that for every finite subset $\Lambda \subset V$
\[ \sum_{X \subset \Lambda} \widetilde{\Phi}^{V, \beta}_{X} = - \frac{1}{\beta} \log[e^{-\beta H_{\Lambda}}] = H_{\Lambda} = \sum_{X \subset \Lambda} \Phi_{X}\,. \]
By an easy induction argument on the cardinal of $\Lambda$, the above equality implies that   $\widetilde{\Phi}^{V, \beta}_{X} = \Phi_{X}$ for every finite subset $X \subset V$. Consequently, applying condition (ii), for every $L \subset V$ and every finite $X \subset L$, we have $\widetilde{\Phi}^{L, \beta}_{X} =  \widetilde{\Phi}^{V, \beta}_{X} = \Phi_{X}$.
\end{rem}

The above definition does not include any condition on the decay of the local effective interaction $\widetilde{\Phi}^{L, \beta}$. One might conjecture that if the original interaction $\Phi$ satisfies a decay condition given in terms of a subbaditive function $\mathbf{b}:\mathcal{P}_{f}(V) \to [0,\infty)$, namely
\[ \| \Phi\|_{\mathbf{b}} = \sup_{x \in V} \sum_{X \ni x} \| \Phi_{X}\| e^{\mathbf{b}(X)} < \infty\,, \]
then $\widetilde{\Phi}^{L, \beta}$ should satisfy a similar decay condition. We will explore this in Section \ref{sec:comm-effective-Hamiltonian}. We will actually obtain a general result where the decay of the local effective interaction is slightly weaker than that of $\Phi$.

Although the effective Hamiltonian has appeared in previous works \cite{Anshu2021, Kuwahara2019, Bilgin2010}, we were unable to find an explicit definition that incorporates locality similar to the conventional definition of local Hamiltonians, namely in terms of a local interaction defined on the lattice. The closest approach we are aware of was explored by Kuwahara et al. \cite{Kuwahara2019}, who used cluster expansions to study the locality properties of the effective Hamiltonian in arbitrary lattices at high temperatures. However, a gap was found in their proof, leaving the validity of their result unknown. Nevertheless, it is worthwhile to compare their locality description given in \cite[Theorem 11]{Kuwahara2019} with the above definition. They apply cluster expansion ideas to analyze the locality properties of
\begin{equation}\label{equa:KuwaharaLocalEffectiveHamiltonian} 
-\frac{1}{\beta} \log\left( \tr_{\Lambda \setminus  L}[e^{-\beta H_\Lambda}]  \right) + \frac{1}{\beta} \log[Z_{\Lambda \setminus L}] \mathbbm{1}\,,
\end{equation}
that they rewrite as $H_{L}$ plus a sum of local terms localized around the boundary of $L$ with exponential decay. Note that \eqref{equa:KuwaharaLocalEffectiveHamiltonian}  can be rewritten in terms of the conditional expectation as
\[ -\frac{1}{\beta} \log\left( \mathbb{E}_{L}[e^{-\beta H_\Lambda}] \right) + \frac{1}{\beta} \log \mathbb{E}_{L}[e^{-\beta H_{\Lambda \setminus L}}] \mathbbm{1}\,. \]
Hence, if the quantum system admits a local effective Hamiltonian at inverse temperature $\beta$ according to Definition \ref{defi:localityEffectiveHamiltonian}, then using property (ii) and the remark after the definition, we have that
\[ -\frac{1}{\beta} \log\left( \mathbb{E}_{L}[e^{-\beta H_\Lambda}]  \right) + \frac{1}{\beta} \log \mathbb{E}_{L}[e^{-\beta H_{\Lambda \setminus L}}] \mathbbm{1} = \sum_{\substack{X \subset \Lambda \\ X \cap L \neq \emptyset}} \widetilde{\Phi}^{L, \beta}_{X} = H_{L} + \sum_{\substack{X \subset \Lambda \\ X \cap L \neq \emptyset \\ X \cap L^{c} \neq \emptyset}} \widetilde{\Phi}^{L, \beta}_{X}\,. \]
Thus, taking inspiration from the approach of Kuwahara et al. \cite{Kuwahara2019}, an alternative definition for the existence of a local effective Hamiltonian is the following:
\begin{defi}[Weak version]\label{defi:weakeffHamiltonian}
Let us say that the above quantum spin system has (weak)\emph{ local effective Hamiltonians at (inverse) temperature $\beta >0$} if it satisfies the following property: for every  subset $L\subset V$ there exists a local interaction $\widehat{\Phi}^{L,\beta}$ on $V$ such that
\begin{enumerate}
\item[(i)]  $\widehat{\Phi}_{X}^{L, \beta}$ is supported in $X \cap L$ for every  $X \subset \subset V$.
\item[(ii)] $\widehat{\Phi}^{L, \beta}_{X} = \widehat{\Phi}^{L', \beta}_{X}$ for all finite subset $X \subset V$ and $L' \subset V$ satisfying $X \cap L' = X \cap L$.
\item[(iii)] For every finite subset $\Lambda \subset V$
\[ \widehat{H}^{L, \beta}_{\Lambda}:=-\frac{1}{\beta} \log\left( \tr_{\Lambda \setminus  L}[e^{-\beta H_\Lambda}] \right) + \frac{1}{\beta} \log[Z_{\Lambda \setminus L}] \mathbbm{1} = \sum_{X \subset \Lambda, X \cap L \neq \emptyset} \widehat{\Phi}_{X}^{L, \beta}\,. \]
\end{enumerate}
\end{defi}

We have shown before that the strong version  (Definition \ref{defi:localityEffectiveHamiltonian}) implies the weak version (Definition \ref{defi:weakeffHamiltonian}). However,  we do not have a proof neither a counterexample for the reverse implication.

\begin{rem}\label{rem:relation_effrho_usualrho}
    For $H$ and $\widehat{H}^{L,\beta}_\Lambda$ defined as in \Cref{defi:weakeffHamiltonian}, and for fixed $\beta >0$, note that     \begin{equation}\label{eq:relation_effrho_usualrho}
      \rho^\Lambda_{\beta,L} \equiv \rho_L =  e^{-\beta \widehat{H}^{L,\beta}_\Lambda} \frac{Z_{\Lambda \setminus L}}{Z_\Lambda}\,.
    \end{equation}
    Thus, bounding products of exponentials of $\widehat{H}$ allows us to bound products of marginals of $\rho$. We can also see that the weak version of the effective Hamiltonian is designed to be able to write marginals of a Gibbs state as Gibbs states of effective Hamiltonians. It is the form that the claimed effective Hamiltonians in the paper \cite{Kuwahara2019} had.

    The strong version of the effective Hamiltonian is designed to write marginals of $\exp(-\beta H_\Lambda)$ as exponentials of the effective Hamiltonian. We will see in the remainder of \Cref{sec:effective_hamiltonian} that for Hamiltonians satisfying a Commuting Hypothesis (Definiton \ref{defi:CommutingHypothesis}), we can prove the existence of an effective Hamiltonian in the strong from.
\end{rem}

\subsection{Effective Hamiltonian in the commuting case} \label{sec:comm-effective-Hamiltonian}

{
We will now show in some cases that effective Hamiltonians with short-range interactions exist, as it is well-known for the Ising model \cite{Sznajd.1984.EffectiveHamiltonianIsing}.  Let us state the main assumption on the local interaction $\Phi$ that will be needed.

\begin{defi}[Commuting Hypothesis] \label{defi:CommutingHypothesis}
Let us say that a local interaction $\Phi$ on $V$ satisfies the \emph{Commuting Hypothesis} if there is a commuting algebra $\mathcal{A} \subset \mathfrak{A}_{V}$ such that  $\Phi_{X} \in \mathcal{A}$ for every $X \in \mathcal{P}_{f}(V)$,  and moreover, for every $L \subset V$ the conditional expectation $\mathbb E_{L}[\cdot]$ satisfies $\mathbb E_{L}[\mathcal{A}] \subset \mathcal{A}$.
\end{defi}

 We next state the main result of this section. 

\begin{thm}\label{thm:GeneralLocalEffectiveInteractionExistence}
Let us consider a quantum spin system with local interaction $\Phi$ on $V$ satisfying the Commuting Hypothesis (Definition \ref{defi:CommutingHypothesis}) and such that for some $\varepsilon >0$ and a subadditive function $\mathbf{b}:\mathcal{P}_{f}(V) \to [0,\infty)$
\[ \| \Phi\|_{\varepsilon, \mathbf{b}} = \sup_{x \in V} \sum_{X \ni x} \| \Phi_{X}\| e^{\varepsilon |X| + \mathbf{b}(X)} < \infty \,.\]
Then, for every $\beta \in \mathbb{C}$ with $|\beta| \leq \varepsilon /(2\| \Phi\|_{\varepsilon, \mathbf{b}})$ there are (strong) local effective Hamiltonians, namely for every $L\subset V$ there exists a local interaction $\widetilde{\Phi}^{L, \beta}$ on $V$ satisfying Definition \ref{defi:localityEffectiveHamiltonian}(i)-(iii), and moreover

\[ \| \widetilde{\Phi}^{L, \beta}\|_{\mathbf{b}} = \sup_{x \in V} \sum_{X \ni x} \| \widetilde{\Phi}^{L, \beta}_{X}\| e^{\mathbf{b}(X)} < \frac{\varepsilon}{2}\,.  \]
\end{thm}

Observe that the decay of the local effective interaction is slightly weaker than the decay of $\Phi$. In particular, if we wanted to ensure that the decay of the effective interaction $\widetilde{\Phi}^{L, \beta}$ satisfies $\| \widetilde{\Phi}^{L, \beta}\|_{\mathbf{b}} < \infty$ for the subadditive function $\mathbf{b}(X) = \lambda |X| + \mu \diam(X)$, we would need that the original interaction $\Phi$ to decay as in $\| \Phi\|_{\mathbf{b}'} < \infty$ with $\mathbf{b}'(X) = (\lambda + \varepsilon) |X| + \mu \diam(X)$ for a positive value $\varepsilon >0$. This fact may be a limitation produced by our techniques. 

Under stronger assumptions on the decay of $\Phi$, however, it is possible to obtain a better result where the local effective Hamiltonian has the same type of decay. Let us denote by $\mathcal{S}$ the set of finite elements $X \subset V$ with $\Phi_{X} \neq 0$.

\begin{defi}
We say that our quantum spin system has finite degree $\mathfrak{d} \in \mathbb{N}$ if for every $X \in \mathcal{S}$, the number of subsets $Y \in \mathcal{S}$ such that $Y \cap X \neq \emptyset$ is at most $\mathfrak{d}$.
\end{defi}

Our primary example is the case of a quantum spin system over $V = \mathbb{Z}^{g}$ endowed with the supremum distance (norm) and with a local interaction $\Phi$ having finite range $r >0$. In this case, it is not difficult to check that it has finite degree, since every ball of radius $r$ intersects at most a number $\mathfrak{d} = \mathfrak{d}(r,g)$ of balls of radius $r$. For this type of interactions, we can prove the following result

\begin{thm}\label{thm:effectiveTemperatureFiniteDegree}
Let us consider a quantum spin system with local interaction $\Phi$ on $V$ satisfying the Commuting Hypothesis (Definition \ref{defi:CommutingHypothesis}), such that for a subadditive $\mathbf{b}:\mathcal{P}_{f}(V) \to [0,\infty)$
\[ \| \Phi\|_{\mathbf{b}} = \sup_{x \in V}\sum_{X \ni x}\| \Phi_{X}\|e^{\mathbf{b}(X)} < \infty\,, \]
and having finite degree $\mathfrak{d}$. Then, for every $|\beta| <\frac{1}{\mathfrak{d}(1+\mathfrak{d})e^{2} \| \Phi\|_{\mathbf{b}}}$, there are (strong) local effective Hamiltonians, namely for every $L\subset V$ there exists a local interaction $\widetilde{\Phi}^{L, \beta}$ on $V$ satisfying Definition \ref{defi:localityEffectiveHamiltonian}(i)-(iii), and moreover
\[ \| \widetilde{\Phi}^{L, \beta}\|_{\mathbf{b}} = \sup_{x \in V}\sum_{X \ni x}\| \widetilde{\Phi}^{L, \beta}_{X}\|e^{\mathbf{b}(X)} \leq 1\,. \]
\end{thm}

To establish the above results, we will employ cluster expansion techniques through the theory of abstract polymer models \cite{FriedliVelenik2018, Ueltschi2004, FernandezProcacci2007, KoteckyPreiss2003}. However, we are going to introduce a novelty that, to the best of our knowledge, has not been considered elsewhere, namely vector-valued polymer models. 
}
\subsubsection{Polymer models and cluster expansions}

A polymer model is described in terms of three elements denoted as $(\mathbb{P}, \xi, \mathbf{w})$. Here, $\mathbb{P}$ is a (nonempty) set  whose elements are referred as \emph{polymers}. There is also an \emph{interaction function} $\xi: \mathbb{P} \times \mathbb{P} \to \mathbb{R}$, which is assumed to satisfy
\[ \xi(\gamma, \gamma') = \xi(\gamma', \gamma) \quad \text{and} \quad |1+\xi(\gamma, \gamma')| \leq 1 \quad , \quad \text{ for all } \gamma, \gamma' \in \mathbb{P}\,. \]
\noindent Lastly, we have a function $\mathbf{w}: \mathbb{P} \to \mathcal{A}$, called the \emph{weight function}, taking values in a (complex) commuting Banach algebra $\mathcal{A}$, and satisfying
\begin{equation}\label{equa:boundedWeight}
|\mathbf{w}|:= \sum_{\gamma \in \mathbb{P}} \| \mathbf{w}(\gamma)\| < \infty\,. 
\end{equation}
In the literature, the weight function is typically assumed to take real or complex values. Although transitioning to vector-valued functions may introduce new challenges when attempting to extend results from scalars to this setting, the results that we will be using can be straightforwardly reproved along the same lines in this case, thanks to the commutativity of the Banach algebra. Investigating the nonconmutative case appears to be an interesting line of research that we will not be pursuing here.
 
Under the above conditions, the  \emph{(polymer) partition function} associated with this model is defined by
\begin{align*} 
\mathcal{Z} & := 1+ \sum_{m=1}^{\infty} \frac{1}{m!} \sum_{(\gamma_{1}, \ldots, \gamma_{m}) \in \mathbb{P}^{m}} \,\, \prod_{j=1}^{m} \mathbf{w}(\gamma_{j}) \, \prod_{1 \leq i <j \leq n} (1+\xi(\gamma_{i}, \gamma_{j}))\,.
\end{align*}
Then, subject to certain conditions on the weight function, we can write the logarithm of the partition function $\mathcal{Z}$ in terms of the Mayer expansion \cite{GruberKunz1971}
\begin{equation}\label{equa:MayerExpansion}
\mathcal{Z}= \exp\left(\sum_{m=1}^{\infty}  \sum_{(\gamma_{1}, \ldots, \gamma_{m}) \in \mathbb{P}^{m}} \phi(\gamma_{1}, \ldots, \gamma_{m}) \prod_{j=1}^{m} \mathbf{w}(\gamma_{j}) \right)
\end{equation}
where the functions $\phi: \cup_{m}\mathbb{P}^{m} \to \mathbb{R}$ are the so-called \emph{Ursell} functions (see e.g. \cite[Eq.\ (2.4)]{FernandezProcacci2007}). To explicitly define them, denote by $\mathcal{G}_{n}$ the set of all graphs with $n$ vertices, that we identify with $\{1,\ldots, n\}$. The edge connecting two vertices $i$ and $j$ will be denoted by $\{ i,j\}$ and to claim that a given graph $G \in \mathcal{G}_{n}$ contains this edge, we will write $\{ i,j\} \in G$ in an abuse of notation. Then, for every $(\gamma_{1}, \ldots, \gamma_{m}) \in \mathbb{P}^{m}$
\begin{equation}\label{equa:UrsellFunctionDefi} 
\phi(\gamma_{1}, \ldots, \gamma_{m})= \frac{1}{m!} \sum_{\substack{G \in \mathcal{G}_{m}\\ \text{ connected}}} \prod_{\{ i,j\} \in G}\xi(\gamma_{i}, \gamma_{j})\,.
\end{equation}
We will say that the sequence $(\gamma_{1}, \ldots, \gamma_{m})$ is a \emph{cluster} if the graph $G \in \mathcal{G}_{m}$, that contains an edge $\{ i,j\}$ if and only if $\xi(\gamma_{i}, \gamma_{j}) \neq 0$, is connected. Observe that $\phi(\gamma_{1}, \ldots, \gamma_{m}) = 0$ if $(\gamma_{1}, \ldots, \gamma_{m})$ is not connected, as every summand in the right hand-side of \eqref{equa:UrsellFunctionDefi} is going to be null.

Several sufficient conditions for the absolute convergence of the series in \eqref{equa:MayerExpansion} have been provided by e.g. Koteck\'{y} and Preiss, Dobrushin, and more recently by Fern\'{a}ndez and Procacci \cite{FernandezProcacci2007}. We are going to base on a criterion that appears in \cite[Theorem 5.4 and Lemma 5.6]{FriedliVelenik2018}, which in turn is based on a more general approach by Ueltschi \cite[Theorem 1]{Ueltschi2004}.

\begin{thm}\label{thm:convergenceAbstractPolymers}
Let $(\mathbb{P}, \xi, \mathbf{w})$ be a polymer model where the weight function $\mathbf{w}$ takes values in a commutative Banach algebra $\mathcal{A}$. Let us assume that there is a function  $\mathbf{a}:\mathbb{P} \to [0, \infty)$ such that 
\begin{enumerate}
\item[(i)] $\sum_{\gamma \in \mathbb{P}}\| \mathbf{w}(\gamma)\| e^{\mathbf{a}(\gamma)} < \infty$,
\item[(ii)] $\sum_{\gamma \in \mathbb{P}} \|\mathbf{w}(\gamma)\| \, |\xi(\gamma, \gamma^{\ast})| e^{\mathbf{a}(\gamma)} \leq \mathbf{a}(\gamma^*)$  for every $\gamma^* \in \mathbb{P}$.
\end{enumerate}
Then, the power series given in Eq.\ \eqref{equa:MayerExpansion} is absolutely convergent, namely
\[ \sum_{m=1}^{\infty}  \sum_{(\gamma_{1}, \ldots, \gamma_{m}) \in \mathbb{P}^{m}} \left|\phi(\gamma_{1}, \ldots, \gamma_{m})\right| \, \prod_{j=1}^{m} \left\| \mathbf{w}(\gamma_{j})\right\| \leq \sum_{\gamma \in \mathbb{P}} \| \mathbf{w}(\gamma)\| e^{\mathbf{a}(\gamma)} < \infty\,. \]
Moreover, for every $\gamma^{\ast} \in \mathbb{P}$
\begin{equation}\label{equa:convergenceAbstractPolymersAux2}
\sum_{m = 1}^{\infty} \sum_{(\gamma_{1}, \ldots, \gamma_{m}) \in \mathbb{P}^{m}} \left( \sum_{j=1}^{m} |\xi(\gamma^{\ast}, \gamma_{j})| \right) |\phi(\gamma_{1}, \gamma_{2}, \ldots, \gamma_{m})| \prod_{j=1}^{m} \| \mathbf{w}(\gamma_{j})\| \leq \mathbf{a}(\gamma^{\ast})\,. 
\end{equation}
\end{thm}

The proof of this result follows the lines of \cite[Theorem 5.4 and Lemma 5.6]{FriedliVelenik2018}, and its original source \cite[Theorem 1]{Ueltschi2004}. Although these proofs are developed in the scalar case, they can be straightforwardly reproduced in the commuting vector-valued case. Let us also observe that in \cite{FriedliVelenik2018} the set of polymers $\mathbb{P}$ is assumed to be finite. This is not a major issue, and actually in the original source \cite{Ueltschi2004} it is permitted that the set of polymers can be infinite. One just needs to add two extra conditions on the weight function that are omitted in the proofs of \cite[Theorem 5.4 and Lemma 5.6]{FriedliVelenik2018}, as they are superfluous under finiteness assumption. The first condition is that the weight function $\mathbf{w}$ satisfies \eqref{equa:boundedWeight}, which is tantamount to the condition in \cite{Ueltschi2004} that the complex measure has bounded total variation; and the second condition is that $\mathbf{a}$ in Theorem \ref{thm:convergenceAbstractPolymers} satisfies condition ($ii$), whose analogue is \cite[Eq. (3) in Theorem 1]{Ueltschi2004}.

\subsubsection{Cluster expansion for the effective Hamiltonian}

Let $\Lambda$ be a finite subset of $V$ and let also $L \subset V$. Next, we want to find a (vector-valued) polymer model $(\mathbb{P}, \xi, \mathbf{w})$ for which we can rewrite
\[ \mathbb E_{L}[e^{-\beta H_{\Lambda}}]\,, \]
as the associated (polymer) partition function, so that we can apply the cluster expansion techniques that allow us to describe its logarithm as a convergent power series. 

Let us denote by $\mathcal{S} = \mathcal{S}(\Lambda)$ the set of all (finite) subsets $X \subset \Lambda$ such that $\Phi_{X} \neq 0$. Then, we can expand
\begin{equation}\label{equa:logPartitionAux1}
\mathbb{E}_{L}[e^{-\beta H_{\Lambda}}] = 1+\sum_{k=1}^{\infty} \frac{(-\beta)^{k}}{k!} \sum_{(X_{1}, \ldots, X_{k}) \in \mathcal{S}^{k}} \mathbb{E}_{L}\left[ \Phi_{X_{1}} \ldots \Phi_{X_{k}} \right]\,.  
\end{equation}
 We next establish an equivalence relation on $\mathcal{S}^{k}$, by defining $\mathbb{X} = (X_{i})_{i=1}^{k} \sim \mathbb{Y}=(Y_{i})_{i=1}^{k}$ if there is a permutation $\pi$ on $\{ 1, \ldots, k\}$ such that $X_{\pi(i)} = Y_{i}$ for every $1 \leq i \leq k$. Let us denote $\mathbb{S}_{k} = \mathbb{S}_{k}(\Lambda) := \mathcal{S}^{k}/\sim$ and $\mathbb{S}= \mathbb{S}(\Lambda):=\cup_{k}\mathbb{S}_{k}$ the set of all equivalence classes. We also define for each $\gamma \in \mathbb{S}_{k}$
 \[ \mathbf{w}_{\beta}(\gamma) := \frac{(-\beta)^{k}}{k!}\sum_{(X_{1}, \ldots, X_{k}) \in \gamma} \mathbb{E}_{L}\left[ \Phi_{X_{1}} \ldots \Phi_{X_{k}} \right] \,, \]
so that Eq.\ \eqref{equa:logPartitionAux1} can be rewritten as
 \begin{equation} \label{equa:logPartitionAux2}
 \mathbb{E}_{L}[e^{-\beta H_{\Lambda}}] = 1+\sum_{k=1}^{\infty}   \sum_{\gamma \in \mathbb{S}_{k}} \mathbf{w}_{\beta}(\gamma) \,.
 \end{equation}
 We can associate (identify) the elements of $\mathbb{S}$ with (nonempty) multisets by considering the equivalence class of $(X_{1}, \ldots, X_{k})$ as the multiset $[ X_{1}, \ldots, X_{k}]$. In this way, each element in $\mathbb{S}$ corresponds to a unique multiset and vice versa. Given two multisets $\gamma = [ X_{1}, \ldots, X_{k}]$ and $\gamma'=[ Y_{1}, \ldots, Y_{l}]$, they are said to be \emph{disjoint}, denoted $\gamma \wedge \gamma'=\emptyset$,  if $X_{i} \cap Y_{j} = \emptyset$ for every $i, j$. Otherwise, we will write $\gamma \wedge \gamma' \neq \emptyset$. The \emph{sum} of the previous multisets $\gamma$ and $\gamma'$ is defined as the new multiset $\gamma \vee \gamma' := [X_{1}, \ldots, X_{k}, Y_{1}, \ldots, Y_{l}]$. We say that a multiset $\gamma$ is \emph{disconnected} if it can  be written as $\gamma = \gamma_{1} \vee \gamma_{2}$ where $\gamma_{1}$ and $\gamma_{2}$ are disjoint (nonempty) multisets. In this case, it is very easy to check that
 \begin{equation}\label{equa:logPartitionAux3} 
 \mathbf{w}_{\beta}(\gamma) = \mathbf{w}_{\beta}(\gamma_{1}) \mathbf{w}_{\beta}(\gamma_{2})\,. 
 \end{equation}
 If $\gamma$ is not disconnected, we will say that it is \emph{connected}. Note that an equivalent way to formulate that a polymer $\gamma$ is connected, is that we can order its elements as $\gamma =[X_{1}, \ldots, X_{k}]$ so that $X_{j} \cap (X_{1} \cup X_{2} \cup \ldots \cup X_{j-1}) \neq \emptyset$ for every $j=2,\ldots, k$. See Figure \ref{fig:connected_polymers}. We can define the function $\chi: \mathbb{S} \to \{ 0,1\}$ given for each $\chi = [X_{1}, \ldots, X_{k}] \in \mathbb{S}_{k}$
 \[
 \chi(\gamma) =\chi(X_{1}, \ldots, X_{k}) = 
 \begin{cases}
1 &  \text{$\gamma$ is conneced}\\
0 & \text{$\gamma$ is disconnected}
 \end{cases}
 \]

\begin{figure}[ht]
\begin{center}

\begin{tikzpicture}[scale=0.4]


\fill [darkyellow!50!orange!25!white] (-0.5,-0.5) rectangle (12.5,12.5);


\foreach \n in {1,...,13}{
\foreach \m in {1,...,13}{
\shade[shading=ball, ball color=darkred!10!white] (\n-1,\m-1) circle (0.2);
}
 }


\draw [black, very thick] (2.5,2.5) rectangle (6.5,6.5);
\node at (4.5,4.5) {\small $X_{1}$};


\draw [black, very thick] (5.5,4.5) rectangle (8.5,9.5);
\node at (6.5,8.5) {\small $X_{2}$};


\draw [black, very thick] (4.5,5.5) rectangle (0.5,7.5);
\node at (1.5,6.5) {\small  $X_{4}$};


\draw [black, very thick] (4.5,3.5) rectangle (9.5,1.5);
\node at (7.5,2.5) {\small  $ X_{3}$};


\draw [black, very thick] (8.5,2.5) rectangle (11.5,0.5);
\node at (10.5,1.5) {\small  $X_{5}$};



\begin{scope}[xshift=17cm]

\fill [darkyellow!50!orange!25!white] (-0.5,-0.5) rectangle (12.5,12.5);


\foreach \n in {1,...,13}{
\foreach \m in {1,...,13}{
\shade[shading=ball, ball color=darkred!10!white] (\n-1,\m-1) circle (0.2);
}
 }

\begin{scope}[yshift=2cm]
\draw [black, very thick] (2.5,2.5) rectangle (6.5,6.5);
\node at (4.5,4.5) {\small $X_{1}$};
\end{scope}

\begin{scope}[yshift=2cm]
\draw [black, very thick] (5.5,4.5) rectangle (8.5,9.5);
\node at (6.5,8.5) {\small $X_{2}$};
\end{scope}

\draw [black, very thick] (4.5,3.5) rectangle (9.5,1.5);
\node at (7.5,2.5) {\small $X_{3}$};

\begin{scope}[yshift=2cm]
\draw [black, very thick] (4.5,5.5) rectangle (0.5,7.5);
\node at (1.5,6.5) {\small $X_{4}$};
\end{scope}

\draw [black, very thick] (8.5,2.5) rectangle (11.5,0.5);
\node at (10.5,1.5) {\small $X_{5}$};

\end{scope}

\end{tikzpicture}

 \caption{On the left-hand side, a multiset $\gamma =[X_{1}, X_{2}, X_{3}, X_{4}, X_{5}]$ that is connected (polymer). On the right-hand side, an homonymous multiset that is disconnected, as it can be decomposed as $\gamma = \gamma_{1} \vee \gamma_{2}$ where $\gamma_{1} = [X_{1}, X_{2}, X_{4}]$ and $\gamma_{2} = [X_{3}, X_{5}]$ satisfy $\gamma_{1} \wedge \gamma_{2} = \emptyset$.}
  \label{fig:connected_polymers}
  \end{center}
\end{figure}
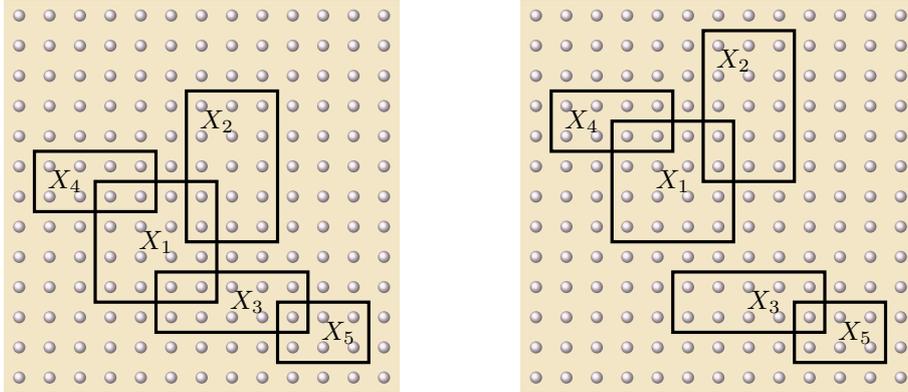
 
 Let  $\mathbb{P}_{k} = \mathbb{P}_{k}(\Lambda)$ be the subset of $\mathbb{S}_{k}$ made of all multisets that are connected, and $\mathbb{P} = \mathbb{P}(\Lambda):= \cup_{k} \mathbb{P}_{k}$. Note that every $\gamma \in \mathbb{S}$ can be decomposed in a unique way as a sum of connected multisets (i.e.\ polymers) $\gamma = \gamma_{1} \vee \ldots \vee \gamma_{m}$ with $\gamma_{i} \wedge \gamma_{j} = \emptyset$ whenever $i \neq j$, let us call them its \emph{connected components}, so that by Eq.\ \eqref{equa:logPartitionAux3} 
 \[ \mathbf{w}_{\beta}(\gamma) = \prod_{i=1}^{m} \mathbf{w}_{\beta}(\gamma_{i})\,. \]
 Using this fact on Eq.\ \eqref{equa:logPartitionAux2}, and rearranging summands according to the number of connected components, we get
 \[ \mathbb{E}_{L}[e^{-\beta H_{\Lambda}}] = 1 + \sum_{m=1}^{\infty} \frac{1}{m!} \sum_{\substack{(\gamma_{1}, \ldots, \gamma_{m}) \in \mathbb{P}^{m} \\ \gamma_{i} \wedge \gamma_{j} = \emptyset \,,\, \forall i \neq j}} \prod_{i=1}^{m} \mathbf{w}_{\beta}(\gamma_{i})  \,.  \]
Defining $\xi: \mathbb{P} \times \mathbb{P} \to \{ 0,-1\}$ as 
\[
\xi(\gamma, \gamma') = -\chi(\gamma \vee \gamma') = 
\begin{cases}
-1 & \gamma \wedge \gamma' \neq \emptyset\\
0 & \gamma \wedge \gamma' = \emptyset
\end{cases}\, ,
\]
we can rewrite again
\begin{equation}\label{equa:CondExpectAsPartitionFunction}
\mathbb{E}_{L}[e^{-\beta H_{\Lambda}}] = 1 + \sum_{m=1}^{\infty} \frac{1}{m!} \sum_{\substack{(\gamma_{1}, \ldots, \gamma_{m}) \in \mathbb{P}^{m}}} \,\, \prod_{i=1}^{m} \mathbf{w}_{\beta}(\gamma_{i}) \prod_{1 \leq i < j \leq m} (1+\xi(\gamma_{i}, \gamma_{j})) 
\end{equation}
This form is consistent with the polymer partition function associated with the polymer model where $\mathbb{P}$ serves as the set of polymers, and it employs the disjointness relation along with the weight function $\mathbf{w}_{\beta}: \mathcal{G} \to \mathfrak{A}_{V}$. However, in order to apply Theorem \ref{thm:convergenceAbstractPolymers} and to get an explicit description of the logarithm, we need to assume that the weight takes values in a commutative algebra.It is at this point we will need that our local interactions $\Phi$ satisfy the Commuting Hypothesis given in Definition \ref{defi:CommutingHypothesis}.

Let us state now the main result from which we will prove the existence of a local effective Hamiltoniam. Recall that a function $\mathbf{c}:\mathcal{P}_{f}(V) \to [0,\infty)$ is subaditive if $\mathbf{c}(X \cup Y) \leq \mathbf{c}(X) + \mathbf{c}(Y)$ for every $X,Y \in \mathcal{P}_{f}(V)$. From such a function, we can construct a function on the set of polymers $\mathbb{P}$ that we denote in the same way $\mathbf{c}:\mathbb{P} \to [0, \infty)$ by defining $\mathbf{c}(\gamma) = \sum_{X \in \gamma} \mathbf{c}(X) =\sum_{i=1}^{m}\mathbf{c}(X_{i})$ if $\gamma = [X_{1}, \ldots, X_{m}]$.

\begin{thm}\label{thm:generalExistenceEffectiveHamiltonianCommuting}
Let us consider a quantum spin system with local interaction $\Phi$ satisfying the Commuting Hypothesis (Definiton \ref{defi:CommutingHypothesis}). Assume that for a given $\beta \in \mathbb{R}$ there exist subadditive maps $\mathbf{a}, \mathbf{b}: \mathcal{P}_{f}(V) \to [0,\infty)$ satisfying that for every $Z \in \mathcal{S}$
\begin{equation}\label{equa:generalExistenceEffectiveHamiltonianCommutingAux1} 
\sum_{k=1}^{\infty} \, |\beta|^{k}\,\sum_{[X_{1}, \ldots, X_{k}] \in \mathbb{P}_{k}} \chi(Z,X_{1}, \ldots, X_{k}) \, \prod_{j=1}^{k}  \| \Phi_{X_{j}}\| e^{\mathbf{a}(X_{j}) + \mathbf{b}(X_{j})} \leq \mathbf{a}(Z)\,. 
\end{equation}
Then, for every $L \subset V$ there exists a local interaction $\widetilde{\Phi}^{L, \beta}$ on $V$ satisfying the following properties:
\begin{enumerate}
\item[(i)]  $\widetilde{\Phi}_{X}^{L, \beta}$ is supported in $X \cap L$ for every  $X\in \mathcal{P}_{f}(V)$.
\item[(ii)] If $L' \subset V$, then $\widetilde{\Phi}^{L, \beta}_{X} = \widetilde{\Phi}^{L', \beta}_{X}$ for all  $X \in \mathcal{P}_{f}(V)$ such that $X \cap L' = X \cap L$.
\item[(iii)] For every finite subset $\Lambda \subset V$
\[ \log \mathbb{E}_{L}[e^{-\beta H_{\Lambda}}] =  \sum_{X \subset \Lambda} \widetilde{\Phi}^{L, \beta}_{X}\,. \]
\item[(iv)] For every $x \in V$
\begin{equation}\label{equa:effectiveInteractionDecayCondition}
\sum_{X \ni x} \| \widetilde{\Phi}^{L, \beta}_{X} \| e^{\mathbf{b}(X)} \leq \mathbf{a}(\{ x\})
\end{equation}

\end{enumerate} 
\end{thm}

\begin{proof}
Let us start by fixing some finite subsets $\Lambda, L \subset V$. Recall the discussion preceding this theorem, where we found that $\mathbb{E}_{L}[e^{-\beta H_{\Lambda}}]$ can be rewritten as the partition function of a polymer model, see \eqref{equa:CondExpectAsPartitionFunction}. The Commuting Hypothesis (Definiton \ref{defi:CommutingHypothesis}) implies that the weight function $\mathbf{w}_{\beta}: \mathbb{P} \to \mathcal{A}$ takes values in a commutative Banach algebra $\mathcal{A}$. Thus, the first assumption of Theorem \ref{thm:convergenceAbstractPolymers} is satisfied.

 Next, we are going to check that conditions (i) and (ii) of Theorem \ref{thm:convergenceAbstractPolymers}  are satisfied when considering as the weight function $\gamma \mapsto \mathbf{w}_{\beta}(\gamma) e^{\mathbf{b}(\gamma)}$. Observe that, as a consequence, the same conditions (i) and (ii) will be satisfied if considering as the weight function only $\gamma \mapsto \mathbf{w}_{\beta}(\gamma) $. We have to consider however this more stringent condition in order to ensure that the last condition on the decay of the effective interaction, eq. \eqref{equa:effectiveInteractionDecayCondition}, is satisfied. On the one hand, observe that we can always bound
 \begin{equation}\label{equa:generalExistenceEffectiveHamiltonianCommutingAux2} 
 |\xi(\gamma, \gamma^\ast)| \leq \sum_{Z \in \gamma^{\ast}} |\xi(\gamma,[Z])|\,. 
 \end{equation}
Combining inequality \eqref{equa:generalExistenceEffectiveHamiltonianCommutingAux2} with the definition of $\xi$ in terms of $\chi$ and also with the hypothesis \eqref{equa:generalExistenceEffectiveHamiltonianCommutingAux1}, we deduce that
\begin{align*}   
\sum_{\gamma \in \mathbb{P}(\Lambda)} |\xi(\gamma, \gamma^*)| \, \| \mathbf{w}_{\beta}(\gamma)\| & e^{\mathbf{a}(\gamma) + \mathbf{b}(\gamma)} 
 \leq \sum_{Z \in \gamma^\ast} \sum_{\gamma \in \mathbb{P}(\Lambda)} |\xi([Z], \gamma)|\, \| \mathbf{w}_{\beta}(\gamma)\| e^{\mathbf{a}(\gamma) + \mathbf{b}(\gamma)}\\
& \leq \sum_{Z \in \gamma^\ast}  \sum_{k=1}^{\infty} \,\, \sum_{[X_{1}, \ldots, X_{k}] \in \mathbb{P}_{k}(\Lambda)} \chi(Z, X_{1}, \ldots, X_{k}) \prod_{j=1}^{k}|\beta| \, \| \Phi_{X_{j}}\| e^{\mathbf{a}(X_j)+\mathbf{b}(X_j)} \\[2mm]
& \leq \sum_{Z \in \gamma^{\ast}} \mathbf{a}(Z) = \mathbf{a}(\gamma^{\ast})\,.
\end{align*}
This shows that Theorem \ref{thm:convergenceAbstractPolymers}.(ii) holds.  With the same idea, we can also argue that Theorem \ref{thm:convergenceAbstractPolymers}.(i) is satisfied, since due to the fact that $\mathcal{S}(\Lambda)$ is finite (because $\Lambda$ is finite), 
\[  
\sum_{\gamma \in \mathbb{P}(\Lambda)} \| \mathbf{w}_{\beta}(\gamma)\| e^{\mathbf{a}(\gamma) + \mathbf{b}(\gamma)} \leq \sum_{Z \in \mathcal{S}(\Lambda)} \sum_{\gamma \in \mathbb{P}(\Lambda)} |\xi(\gamma, [Z])| \, \|\mathbf{w}_{\beta}(\gamma)\| e^{\mathbf{a}(\gamma)+\mathbf{b}(\gamma)} \leq \sum_{Z \in \mathcal{S}(\Lambda)} \mathbf{a}(Z)< \infty\,.
\]
Thus, as a consequence of Theorem \ref{thm:convergenceAbstractPolymers}, we conclude that
\[
\sum_{m=1}^{\infty} \sum_{(\gamma_{1}, \ldots, \gamma_{m}) \in \mathbb{P}^{m}(\Lambda)} |\phi(\gamma_{1}, \ldots, \gamma_{m})| \prod_{j=1}^{m} \|\mathbf{w}_{\beta}(\gamma_{j})\| < \infty\,.
\]
The absolute convergence of the previous sum allows us to define, for every finite subset  $X \subset \Lambda$, the following Hermitian operator supported on $X$:
\[ 
\widetilde{\Phi}_{X}^{L, \beta} := \sum_{n=1}^{\infty}  \,\,
\sum_{\substack{(\gamma_{1}, \ldots, \gamma_{n}) \in \mathbb{P}(\Lambda)^{n}\colon \\ \supp(\gamma_{1} \vee \ldots \vee \gamma_{n}) = X}} \phi(\gamma_{1}, \ldots, \gamma_{n}) \prod_{i=1}^{n} \mathbf{w}_{\beta}(\gamma_{i})\,. 
\]
Note that the preceding definition is independent of $\Lambda \supset X$, namely, we can replace in the above sum the indexing $(\gamma_{1}, \ldots, \gamma_{n}) \in \mathbb{P}(\Lambda)^{n}$ with  $(\gamma_{1}, \ldots, \gamma_{n}) \in \mathbb{P}^{n}$ due to the additional condition $\supp(\gamma_{1} \vee \ldots \vee \gamma_{n}) = X$. Thus, since $\Lambda$ is arbitrary, we have defined a local interaction \mbox{$\widetilde{\Phi}^{L, \beta}: \mathcal{P}_{f}(V) \to \mathbb{R}$} that satisfies for every finite $\Lambda \subset V$
\[
\mathbb{E}_{L}[e^{-\beta H_{\Lambda}}]= \exp\left(\sum_{n=1}^{\infty}  \sum_{(\gamma_{1}, \ldots, \gamma_{n}) \in \mathbb{P}^{n}(\Lambda)} \phi(\gamma_{1}, \ldots, \gamma_{n}) \prod_{i=1}^{n} \mathbf{w}_{\beta}(\gamma_{i}) \right) = \exp\left( \sum_{X \subset \Lambda} \widetilde{\Phi}^{L,\beta}_{X} \right)\,.
\]
Let us check that $\widetilde{\Phi}^{L, \beta}$ is a local interaction satisfying (i) and (ii). For every finite subset $X \subset V$, it is clear that $\widetilde{\Phi}_{X}^{L, \beta}$ is self-adjoint, since the conditional expectation preserves Hermiticity. Moreover, it  satisfies (i) since if $Q$ is supported in $X$, then $\mathbb{E}_{L}[Q]$ is supported in $L \cap X$; and also (ii) since $\mathbb{E}_{L}[Q] = \mathbb{E}_{L'}[Q]$ if $Q$ is supported in $X$ and $L' \cap X = L \cap X$.

To verify the decay condition \eqref{equa:effectiveInteractionDecayCondition}, we use first that $\mathbf{b}$ is subadditive to estimate
\[ 
\sum_{X \ni x}\|\widetilde{\Phi}_{X}^{L, \beta} \| e^{\mathbf{b}(X)} 
\leq 
\sum_{X \ni x} \, \sum_{n=1}^{\infty}  \,\,
\sum_{\substack{(\gamma_{1}, \ldots, \gamma_{n}) \in \mathbb{P}^{n} \colon \\ \supp(\gamma_{1} \vee \ldots \vee \gamma_{n}) = X}} |\phi(\gamma_{1}, \ldots, \gamma_{n})| \, \prod_{i=1}^{n} \|\mathbf{w}_{\beta}(\gamma_{i})
\| e^{\mathbf{b}(\gamma_{i})}
\]
Then, we apply
inequality \eqref{equa:convergenceAbstractPolymersAux2} from Theorem \ref{thm:convergenceAbstractPolymers}, so that 
\begin{align*}
 \sum_{X \ni x}\|\widetilde{\Phi}_{X}^{L, \beta} \| e^{\mathbf{b}(X)}
& \leq 
\sum_{n=1}^{\infty} \sum_{\substack{(\gamma_{1}, \ldots, \gamma_{n}) \in \mathbb{P}^{n} \colon \\ x \in \supp(\gamma_{1} \vee \ldots \vee \gamma_{n})}} |\phi(\gamma_{1}, \ldots, \gamma_{n}) | \, \prod_{i=1}^{n} \|\mathbf{w}_{\beta}(\gamma_{i})
\| e^{\mathbf{b}(\gamma_{i})}\\
& \leq 
\sum_{n=1}^{\infty} \sum_{(\gamma_{1}, \ldots, \gamma_{n}) \in \mathbb{P}^{n}} \left( \sum_{i=1}^{n} \,|\xi([\{ x\}], \gamma_{i})|  \right) |\phi(\gamma_{1}, \ldots, \gamma_{n})| \prod_{i=1}^{n} \|\mathbf{w}_{\beta}(\gamma_{i})
\| e^{\mathbf{b}(\gamma_{i})}\\[1.5mm]
& \leq
\mathbf{a}(\{ x\})\,.
\end{align*}
This finishes the proof of (iv). Note that conditions (i) - (iii) also hold, as per the discussion above on the weights.
\end{proof}

In the following subsections, we will apply the previous result to prove Theorems \ref{thm:effectiveTemperatureFiniteDegree} (finite-range interactions), 
for which the argument is simpler, and later Theorem \ref{thm:GeneralLocalEffectiveInteractionExistence} (exponentially decaying interactions).

\subsubsection{Finite-degree case: Proof of Theorem \ref{thm:effectiveTemperatureFiniteDegree}}

We will need the following auxiliary result that appears in \cite[Lemma 1]{wild2023classical}, which is a reformulation of \cite[Proposition 3.6]{haah2021optimal}.

\begin{prop}\label{Prop:countingPolymersFixedComponent}
Let us assume that the quantum spin system has finite degree $\mathfrak{d}$. Then, for every $X \in \mathcal{S}$ and $m \in \mathbb{N}$, the number of polymers $\gamma=[X_{1}, \ldots, X_{k}] \in \mathbb{P}_{k}$ such that $X \in \gamma$ is at most $(e \mathfrak{d})^{k}$. 
\end{prop}

Next, we can prove the main result that establishes the existence of an effective Hamiltonian for high temperatures under the specified conditions.

\begin{proof}[Proof of Theorem \ref{thm:effectiveTemperatureFiniteDegree}]
We have to prove that the hypotheses of Theorem \ref{thm:generalExistenceEffectiveHamiltonianCommuting} are satisfied for this choice of $\mathbf{a}$ and $\mathbf{b}$. To check that inequality \eqref{equa:generalExistenceEffectiveHamiltonianCommutingAux1} holds, we simply use that for each set $X \in \mathcal{S}$ we can estimate $\| \Phi_{X}\|e^{\mathbf{a}(X) + \mathbf{b}(X)} \leq \| \Phi\|_{\mathbf{a} + \mathbf{b}}$, so that
\begin{align*}
\sum_{k=1}^{\infty} \,\,\sum_{[X_{1}, \ldots, X_{k}] \in \mathbb{P}_{k}} \chi(Z,X_{1}, \ldots, X_{k}) \, & \prod_{j=1}^{k} |\beta| \| \Phi_{X_{j}}\| e^{\mathbf{a}(X_{j}) + \mathbf{b}(X_{j})}\\ 
& \leq \sum_{k=1}^{\infty} \, |\beta|^{k} \| \Phi\|_{\mathbf{a} + \mathbf{b}}^{k}\,\sum_{[X_{1}, \ldots, X_{k}] \in \mathbb{P}_{k}} \chi(Z,X_{1}, \ldots, X_{k}) \,.
\end{align*}
Since each $\gamma=[X_{1}, \ldots, X_{k}] \in \mathbb{P}_{k}$ is connected,  $\chi(Z,X_{1}, \ldots, X_{k}) = 1$ if and only if (at least) one of the sets $X_{j}$ satisfies $X_{j} \cap Z \neq \emptyset$. Then, we can estimate for each $k \in \mathbb{N}$
\[
\sum_{[X_{1}, \ldots, X_{k}] \in \mathbb{P}_{k}} \chi(Z,X_{1}, \ldots, X_{k})
\leq  
\sum_{X \in \mathcal{S}}\chi(Z, X) \sum_{\gamma \in \mathbb{P}_{k} \colon \gamma \ni X} 1 \leq \sum_{X \in \mathcal{S}} \chi(Z, X) (\mathfrak{d}e)^{k} \leq \mathfrak{d} (\mathfrak{d}e)^{k}\,. 
\]
Thus, applying this estimate in the above inequality
\[
\sum_{k=1}^{\infty} \,\,\sum_{[X_{1}, \ldots, X_{k}] \in \mathbb{P}_{k}} \chi(Z,X_{1}, \ldots, X_{k}) \,  \prod_{j=1}^{k} |\beta| \| \Phi_{X_{j}}\| e^{\mathbf{a}(X_{j}) + \mathbf{b}(X_{j})} \leq \sum_{k=1}^{\infty} \mathfrak{d}(|\beta|\, \|\Phi \|_{\mathbf{a} + \mathbf{b}} \mathfrak{d} e)^{k}\,.
\]
Taking $\mathbf{a}$ as a constant function, $\mathbf{a}(X) = a>0$ for every $X \in \mathcal{P}_{f}(V)$, which is obviously subadditive, we have that $\| \Phi\|_{\mathbf{a} + \mathbf{b}} = e^{a} \| \Phi\|_{\mathbf{b}}$, so that
\[
\sum_{k=1}^{\infty} \,\,\sum_{[X_{1}, \ldots, X_{k}] \in \mathbb{P}_{k}} \chi(Z,X_{1}, \ldots, X_{k}) \,  \prod_{j=1}^{k} |\beta| \| \Phi_{X_{j}}\| e^{\mathbf{a}(X_{j}) + \mathbf{b}(X_{j})}
\leq
\sum_{k=1}^{\infty} \mathfrak{d}(|\beta|\mathfrak{d} e^{1+a} \|\Phi\|_{\mathbf{b}})^{k}
\]
Therefore, if $|\beta|<\frac{1}{\mathfrak{d} e^{1+a}\| \Phi\|_{\mathbf{b}}}$, then we can estimate
\[ \sum_{k=1}^{\infty} \,\,\sum_{[X_{1}, \ldots, X_{k}] \in \mathbb{P}_{k}} \chi(Z,X_{1}, \ldots, X_{k}) \,  \prod_{j=1}^{k} |\beta| \| \Phi_{X_{j}}\| e^{\mathbf{a}(X_{j}) + \mathbf{b}(X_{j})} \leq \frac{|\beta| \mathfrak{d}^{2} e^{1+a}  \| \Phi\|_{\mathbf{b}}}{1-|\beta| \mathfrak{d} e^{1+a}  \| \Phi\|_{\mathbf{b}}} \,. \]
If we moreover impose $|\beta| \leq \frac{a}{a+\mathfrak{d}} \frac{1}{\mathfrak{d}e^{1+a}\| \Phi\|_{\mathbf{b}}}$, {using that $x \mapsto \frac{x}{1-x}$ is increasing on $[0,1)$} we have that
\[ 
\frac{|\beta| \mathfrak{d}^{2} e^{1+a}  \| \Phi\|_{\mathbf{b}}}{1-|\beta| \mathfrak{d} e^{1+a}  \| \Phi\|_{\mathbf{b}}} \leq a = \mathbf{a}(Z)\,,
\]
which means that \eqref{equa:generalExistenceEffectiveHamiltonianCommutingAux1} is satisfied. Observe that the choice of $a>0$ is arbitrary. Actually we could try to optimize the map $a \mapsto \frac{a}{a+\mathfrak{d}} e^{-a}$, obtaining that it reaches an absolute maximum on $[0, \infty)$ at $a= \frac{\sqrt{\mathfrak{d}^{2} + 4 \mathfrak{d}} - \mathfrak{d}}{2} = \frac{2\mathfrak{d}}{\mathfrak{d}+\sqrt{\mathfrak{d}^{2} + 4\mathfrak{d}}} = \frac{2}{1+\sqrt{1 + 4/\mathfrak{d}} }$. Since the expression is rather complicated, we can take $a=1$, obtaing the expression that appears in the statement of the theorem.
\end{proof}

\subsubsection{Exponentially-decaying interactions: Proof of Theorem \ref{thm:GeneralLocalEffectiveInteractionExistence}}

We will need the following auxiliary result. It is based on the proof of \cite[Theorem 5.4]{FriedliVelenik2018}.

\begin{lem}\label{lem:toolConvergenceGeneral}
Let $\mathbf{c}, \mathbf{u}:\mathcal{P}_{f}(V) \to [0,\infty)$ such that for every $Z \in \mathcal{S}$
\begin{equation}\label{equa:lemToolConvergenceGeneralHypo} 
\sum_{X \in \mathcal{S}}\chi(Z, X)\, \mathbf{u}(X) e^{\mathbf{c}(X)} \leq \mathbf{c}(Z)\,. 
\end{equation}
Then,
\begin{equation}\label{equa:lemToolConvergenceGeneralMain}
\sum_{k=1}^{\infty}\sum_{[X_{1}, \ldots, X_{k}] \in \mathbb{S}_{k}} \chi(Z, X_{1}, \ldots, X_{k}) \prod_{j=1}^{k} \mathbf{u}(X_j) \leq e^{\mathbf{c}(Z)}-1\,. 
\end{equation}
\begin{equation}\label{equa:lemToolConvergenceGeneralMain2}
\sum_{k=1}^{\infty}\sum_{[X_{1}, \ldots, X_{k}] \in \mathbb{P}_{k}} \chi(Z, X_{1}, \ldots, X_{k}) \prod_{j=1}^{k} \mathbf{u}(X_j) \leq \mathbf{c}(Z)\,. 
\end{equation}
\end{lem}

\begin{proof}
Let us denote $\mathbb{P}_{\leq m} = \cup_{j \leq m} \mathbb{P}_{j}$ and $\mathbb{S}_{\leq m} = \cup_{j \leq m} \mathbb{S}_{j}$ for every $m \in \mathbb{N}$. We are going to prove by induction on $m$ that for every $Z \in \mathcal{S}$ we have 
\begin{equation}\label{equa:lemToolConvergenceGeneralAux1} 
\sum_{\gamma \in \mathbb{P}_{\leq m}} \chi(Z,\gamma) \mathbf{u}^{\gamma} \leq \mathbf{c}(Z) \quad \text{and} \quad \sum_{\gamma \in \mathbb{S}_{\leq m}} \chi(Z,\gamma) \mathbf{u}^{\gamma} \leq e^{\mathbf{c}(Z)}-1\, ,  
\end{equation}
for $\mathbf{u}^{\gamma}= \prod_{X \in \gamma} \mathbf{u} (X)$. Then, taking limit when $m$ tends to infinity we will get that both \eqref{equa:lemToolConvergenceGeneralMain} and \eqref{equa:lemToolConvergenceGeneralMain2} hold. For $m=1$, both inequalities are a simple consequence of the hypothesis \eqref{equa:lemToolConvergenceGeneralHypo}, as
\[ 
\sum_{[X_{1}] \in \mathbb{P}_{1}} \chi(Z, X_{1}) \, \mathbf{u}(X_{1}) = \sum_{[X_{1}] \in \mathbb{S}_{1}} \chi(Z, X_{1}) \, \mathbf{u}(X_{1}) = \sum_{X \in \mathcal{S}} \chi(Z,X) \, \mathbf{u}(X) \leq \mathbf{c}(Z) \leq e^{\mathbf{c}(Z)}-1\,. \]
Let us next assume that \eqref{equa:lemToolConvergenceGeneralAux1} holds for $m$. To see that it holds for $m+1$, let us start by noticing that since polymers $\gamma \in \mathbb{P}$ are simply the elements of $\mathbb{S}$ satisfying $\chi(\gamma) \neq 0$ and thus equal to one, we can rewrite 
\[ \sum_{\gamma \in \mathbb{P}_{\leq m+1}} \chi(Z, \gamma)  \, \mathbf{u}^{\gamma}  = \sum_{\gamma \in \mathbb{S}_{\leq m+1}} \chi(Z, \gamma) \, \chi(\gamma) \, \mathbf{u}^{\gamma}\,.  \]
The condition $\chi(Z,\gamma) = 1$ yields that $X \cap Z \neq \emptyset$ for some $X \in \gamma$. Therefore, we can split $\gamma = \gamma' \vee [X]$ and estimate 
\begin{align*}
\sum_{\gamma \in \mathbb{S}_{\leq m+1}} \chi(Z, \gamma) \, \chi(\gamma) \, \mathbf{u}^{\gamma} 
& \leq
\sum_{X \in \mathcal{S}}   \sum_{\gamma' \in \mathbb{S}_{\leq m}}\chi(Z,X) \, \chi(X,\gamma') \,\mathbf{u}(X)\, \mathbf{u}^{\gamma'} \\[1.5mm]
& = \sum_{X \in \mathcal{S}}  \chi(Z,X) \,\mathbf{u}(X) \sum_{\gamma' \in \mathbb{S}_{\leq m}} \chi(X,\gamma') \, \mathbf{u}^{\gamma'}\,.
\end{align*}
Using then the induction hypothesis in the previous expression, and subsequently the condition of the statement of the lemma, we can further get
\[
\sum_{\gamma \in \mathbb{S}_{\leq m+1}} \chi(Z, \gamma) \, \chi(\gamma) \, \mathbf{u}^{\gamma} 
\leq \sum_{X \in \mathcal{S}} \chi(Z,X) \mathbf{u}(X) e^{\mathbf{c}(Z)}
\leq \mathbf{c}(Z)\,.
\]
This shows that the left hand-side inequality of \eqref{equa:lemToolConvergenceGeneralAux1} holds for $m+1$. To show that the inequality of the right hand-side also holds, recall that given any multiset $\gamma \in \mathbb{S}_{k}$, we know that it admits a unique (except for reordering) decomposition as a sum of connected multisets (i.e. polymers) 
\[ 
\gamma = \gamma_{1} \vee \ldots \vee \gamma_{l} \quad \text{such that} \quad \gamma_{j} \wedge \gamma_{k} = \emptyset \quad \text{whenever} \quad j \neq k\,,
\]
The fact that each $\gamma_{j}$ is connected and that $\gamma_{j} \wedge \gamma_{k} = \emptyset$ whenever $j \neq k$, ensures that 
\[ \chi(Z, \gamma) \mathbf{u}^{\gamma} = \prod_{j=1}^{l} \chi(Z, \gamma_{j}) \chi(\gamma_{j}) \mathbf{u}^{\gamma_{j}}\,. \]
Therefore, we can upper estimate 
\begin{equation}\label{equa:lem:toolConvergenceGeneralAux1}
\sum_{\gamma \in \mathbb{S}_{\leq m+1}} \chi(Z, \gamma) \mathbf{u}^{\gamma} \leq \sum_{l=1}^{m}\frac{1}{l!} \left( \sum_{\gamma \in \mathbb{P}_{\leq m+1}} \chi(Z, \gamma) \,  \, \mathbf{u}^{\gamma}\right)^{l}\,. 
\end{equation}
 Next, using the inequality of the left hand-side of \eqref{equa:lemToolConvergenceGeneralAux1}  for $m+1$ that we just proved, we conclude that
\[
\sum_{\gamma \in \mathbb{S}_{\leq m+1}} \chi(Z, \gamma) \mathbf{u}^{\gamma} \leq \sum_{l=1}^{m}\frac{1}{l!} \mathbf{c}(Z)^{l} \leq e^{\mathbf{c}(Z)}\,. 
\]
This finishes the proof by induction.
\end{proof}

We can now prove the main result.

{
\begin{proof}[Proof of Theorem \ref{thm:GeneralLocalEffectiveInteractionExistence}]
We are going to apply Theorem \ref{thm:generalExistenceEffectiveHamiltonianCommuting}. Let us consider $\mathbf{a}, \mathbf{b}:\mathcal{S} \to \mathbb{R}$ by $\mathbf{a}(X) = (\varepsilon/2) |X|$ and an arbitrary $\mathbf{b}$. We have to prove that
\[
\sum_{k=1}^{\infty} \, |\beta|^{k}\,\sum_{[X_{1}, \ldots, X_{k}] \in \mathbb{P}_{k}} \chi(Z,X_{1}, \ldots, X_{k}) \, \prod_{j=1}^{k}  \| \Phi_{X_{j}}\| e^{\mathbf{a}(X_{j}) + \mathbf{b}(X_{j})} \leq \mathbf{a}(Z) \, .
\]
For that, we will use the previous Lemma \ref{lem:toolConvergenceGeneral} with the maps $\mathbf{c},\mathbf{u} : \mathcal{S} \to [0, \infty)$ given by $\mathbf{c}=\mathbf{a}$ and $\mathbf{u}(X) = |\beta|\, \| \Phi_{X}\| e^{\mathbf{a}(X) +\mathbf{b}(X)}$. With this choice, note that the previous estimate corresponds to Eq. \eqref{equa:lemToolConvergenceGeneralMain2}.  Thus, we just have to check the hypotheses of the Lemma. But this can be easily verified, since
\[ \sum_{X \in \mathcal{S}}\chi(Z,X) \, |\beta|\, \| \Phi_{X}\| e^{\mathbf{b}(X)}e^{2\mathbf{a}(X)} \leq  |\beta|\, \sum_{x \in Z} \sum_{X \ni x} \| \Phi_{X}\|e^{\varepsilon |X| + \mathbf{b}(X)} \leq |\beta| \, |Z|\, \| \Phi\|_{\varepsilon, \mathbf{b}} \leq \frac{\varepsilon}{2} |Z| = \mathbf{a}(Z) \,  \]
if we take $|\beta| \leq \varepsilon/(2 \| \Phi\|_{\varepsilon, \mathbf{b}})$, so we conclude the result.
\end{proof}
}
\section{Mixing condition via strong effective Hamiltonians}\label{sec:stronf_effHam_implies_mixing_condition}

To finalize this section, we are going to show that, under the existence of a strong local effective Hamiltonian as the one described above, we can show that exponential decay of covariance can be lifted to the mixing condition.

\begin{theo}[Strong form]\label{prop:mixingConditionStrongEffectiveHamiltonian} Let us assume that $\Phi$ is an interaction on $V$ satisfying the \emph{strong} local effective Hamiltonian property at $\beta >0$, and assume that there is a uniform bound $\Delta >0$ such that, for every $L \subset V$, the local interaction $\widetilde{\Phi}^{L, \beta}$ satisfies
\[ \|\widetilde{\Phi}^{L, \beta}\| = \sup_{x \in V} \sum_{X \ni x}\|\Phi_{X}\|e^{\lambda |X| + \mu \operatorname{diam}(X)} \leq \Delta \,. \]
Then, for every $\Lambda \in \mathcal{P}_{f}(V)$ and every pair of disjoint subsets $A,C \subset \Lambda$, the local Gibbs state $\rho = \rho^{\Lambda}_{\beta}$ satisfies whenever $\beta < \lambda/(2\Delta)$
\[
\left\| \rho_{AC} \rho_{A}^{-1} \otimes \rho_{C}^{-1} - \mathbbm{1} \right\| \leq   \exp\left( \frac{4 \Delta \lambda \beta}{\lambda - 2 \Delta \beta} \sum_{x \in A} e^{-\mu \operatorname{dist}(x,C)}\right) \, \frac{4 \Delta \lambda \beta}{\lambda - 2 \Delta \beta} \sum_{x \in A} e^{-\mu \operatorname{dist}(x,C)} \, .
\]
\end{theo}

\begin{rem}
If $V = \mathbb{Z}^{g}$, then we can use the notation from Remark \ref{rem:onion} and rewrite the above estimation as
\[ \| \rho_{AC} \rho_{A}^{-1} \otimes \rho_{C}^{-1} - \mathbbm{1} \| \leq \exp \{ 4\beta K  |\partial A| e^{- (\mu/2) \operatorname{dist}(A,C) }  \} \cdot  4\beta K |\partial A| e^{- (\mu/2) \operatorname{dist}(A,C) }  \,.\]
Exchanging the roles of $A$ and $C$, we could write $|\partial C|$ instead of $|\partial A|$, taking the minimum of both values to minimize the expression. In any case, we have an exponential decay on the distance between $A$ and $C$.
\end{rem}

\begin{proof}
Let us drop the dependence on $\Lambda$ and $\beta$ in $\rho$ to ease notation. Consider $\Lambda$ and $A, C \subset \Lambda$ as in the statement of the proposition. Then,
\[
\rho_{AC} \rho_{A}^{-1} \rho_{C}^{-1} = \mathbb{E}_{AC}[e^{-\beta H_{\Lambda}}] \mathbb{E}_{A}[e^{-\beta H_{\Lambda}}]^{-1} \mathbb{E}_{C}[e^{-\beta H_{\Lambda}}]^{-1} \mathbb{E}_{\emptyset}[e^{-\beta H_{\Lambda}}]\,. 
\]
Using the hypothesis, there are local interactions $\widetilde{\Phi}^{A,\beta}, \widetilde{\Phi}^{C,\beta}, \widetilde{\Phi}^{AC,\beta}$ and $\widetilde{\Phi}^{\emptyset,\beta}$ such that for every $L \in \{ A,C,AC, \emptyset\}$ 
\[  \widetilde{H}^{L, \beta}_{\Lambda}:=-\frac{1}{\beta} \log\left( \mathbb{E}_{L}[e^{-\beta H_\Lambda}] \right) =  \sum_{X \subset \Lambda} \widetilde{\Phi}_{X}^{L, \beta}\,. \]
For the remainder of the proof, we will omit the superscript $\beta$ and subscript $\Lambda$ from the local interactions $\widetilde{\Phi}^{L} = \widetilde{\Phi}^{L, \beta}$ and the effective Hamiltonian as well. Thus, we can rewrite 
\[ \rho_{AC} \rho_{A}^{-1} \rho_{C}^{-1} = 
e^{-\beta \widetilde{H}^{AC}} 
e^{\beta \widetilde{H}^{A}} 
e^{\beta \widetilde{H}^{C}} 
e^{-\beta \widetilde{H}^{\emptyset}} 
\,. 
\]
Since $\widetilde{H}^{\emptyset}$ is a multiple of the identity, it commutes with every operator. Moreover, $\widetilde{H}^{A}$ and $\widetilde{H}^{C}$ commute with each other since every $\Phi_{X}^{A}$ is supported in $A$, every $\Phi_{X}^{C}$ is supported in $C$, and $A \cap C= \emptyset$. Therefore, we can rearrange
\[ \rho_{AC} \rho_{A}^{-1} \rho_{C}^{-1} = 
e^{-\beta \widetilde{H}^{AC} } 
e^{\beta (\widetilde{H}^{A}+ \widetilde{H}^{C} - \widetilde{H}^{\emptyset})} 
\,. 
\]
Let us recall that for every pair of operators $Q$ and $W$ in $\mathfrak{A}_{\Lambda}$ we have the following expansion in terms of the time-evolution operator (see \cite[Eqs.\ (5.5), (5.9)]{Araki1969})
\begin{equation}\label{equa:mixingviaStrongEffectiveAux2} 
e^{-\beta Q}e^{\beta(Q+W)} = \sum_{m=0}^{\infty} \,\, \int_{0}^{\beta}dt_{1} \int_{0}^{t_{1}}dt_{2} \ldots \int_{0}^{t_{m-1}} dt_{m} \,\, \Gamma_{Q}^{i  t_{1}}(W) \cdot \ldots \cdot \Gamma_{Q}^{i t_{m}}(W)
\end{equation}
 where  $\Gamma_Q^{it}(W) = e^{-tQ}We^{tQ}$. Thus, applying this identity with
 \[ Q:= \widetilde{H}^{AC, \beta}_{\Lambda}  \quad \mbox{and} \quad W = \widetilde{H}^{A, \beta}_{\Lambda} + \widetilde{H}^{C, \beta}_{\Lambda} - \widetilde{H}^{\emptyset, \beta}_{\Lambda} - \widetilde{H}^{AC, \beta}_{\Lambda}\,, \]
 we can estimate
\begin{equation}\label{eq:estimate_prod_exp_eff_ham_minus_identity_first}
    \left\| \rho_{AC} \rho_{A}^{-1} \rho_{C}^{-1} - \mathbbm{1} \right\| =  \left\| e^{-\beta Q}e^{\beta(Q+W)} - \mathbbm{1}\right\| \leq \sum_{m=1}^{\infty} \frac{|\beta|^{m}}{m!} \, \left(\sup_{0 \leq s \leq \beta}\| \Gamma^{is}_{Q}(W)\|\right)^{m} \, .
\end{equation}
Next we are going to expand $W$ in terms of the local interactions and note some cancellations between summands:
\[ W = \sum_{X \subset \Lambda} \widetilde{\Phi}^{A}_{X} + \widetilde{\Phi}^{C}_{X}  - \widetilde{\Phi}^{AC}_{X} - \widetilde{\Phi}^{\emptyset}_{X}\,.  \]
We claim that if $X \cap C = \emptyset$ or $X \cap A = \emptyset$, then 
\[ \widetilde{\Phi}^{A}_{X} + \widetilde{\Phi}^{C}_{X}  - \widetilde{\Phi}^{AC}_{X} - \widetilde{\Phi}^{\emptyset}_{X} =0\,. \]
Indeed,  if $X \cap C = \emptyset$, then $X \cap AC = X \cap A$ and $X \cap C = X \cap \emptyset$, which respectively yield that $\widetilde{\Phi}^{AC}_{X} = \widetilde{\Phi}^{A}_{X}$ and $\widetilde{\Phi}^{C}_{X} = \widetilde{\Phi}^{\emptyset}_{X}$, by property (iii) in Definition \ref{defi:localityEffectiveHamiltonian}. Thus, we get a zero summand in this case. The argument when $X \cap A=\emptyset$ is analogous exchanging the roles of $A$ and $C$. Therefore, the expression for $W$ can be simplified to
\[ W = \sum_{X \subset \Lambda \colon X \cap A \neq \emptyset, X \cap C \neq \emptyset}  \widetilde{\Phi}^{A}_{X}  + \widetilde{\Phi}^{C}_{X} - \widetilde{\Phi}^{AC}_{X} - \widetilde{\Phi}^{\emptyset}_{X}\,.  \]
Next, we can estimate
\[ \| \Gamma_{Q}^{is}(W)\| \leq \sum_{X \subset \Lambda \colon X \cap A \neq \emptyset, X \cap C \neq \emptyset} ( \|\Gamma_{Q}^{is}(\widetilde{\Phi}^{A}_{X}) \| +\|\Gamma_{Q}^{is}(\widetilde{\Phi}^{C}_{X}) \|
+\|\Gamma_{Q}^{is}(\widetilde{\Phi}^{AC}_{X}) \|
+\|\widetilde{\Phi}^{\emptyset}_{X} \|) \,, \]
where in the last summand we have used that $\widetilde{\Phi}^{\emptyset, \beta}_{X}$ is a multiple of the identity. Applying Proposition \ref{theo:localityEstimates}, we can estimate since $|s| \leq \beta  < \lambda/(2\Delta)$
\begin{align*}
\| \Gamma_{Q}^{is}(W)\| 
& \leq   \,  \frac{\lambda}{\lambda-2\Delta|s|}  \sum_{\substack{X \subset \Lambda \colon\\[2mm] X \cap A \neq \emptyset, X \cap C \neq \emptyset}} e^{\lambda |X|}(\|\widetilde{\Phi}^{AC}_{X} \| + \|\widetilde{\Phi}^{A}_{X} \| + \|\widetilde{\Phi}^{C}_{X} \| +\|\widetilde{\Phi}^{\emptyset}_{X} \|)\\[2mm]
& \leq   \, \frac{\lambda}{\lambda-2 \Delta |s|}   \sum_{x \in A} \sum_{\substack{X \ni x \colon\\[1mm] X \cap C \neq \emptyset}} e^{\lambda |X|}(\|\widetilde{\Phi}^{AC}_{X} \| + \|\widetilde{\Phi}^{A}_{X} \| + \|\widetilde{\Phi}^{C}_{X} \|+\|\widetilde{\Phi}^{\emptyset}_{X} \|)\\[2mm]
& =   \, \frac{\lambda}{\lambda-2\Delta|s|}   \sum_{x \in A} \sum_{\substack{X \ni x \colon\\[1mm] X \cap C \neq \emptyset}} e^{-\mu \diam(X)} e^{\lambda |X|+ \mu \diam(X)}(\|\widetilde{\Phi}^{AC}_{X} \| + \|\widetilde{\Phi}^{A}_{X} \| + \|\widetilde{\Phi}^{C}_{X} \|+\|\widetilde{\Phi}^{\emptyset}_{X} \|)\\[2mm]
& \leq   \, \frac{\lambda}{\lambda-2 \Delta |s|}   \sum_{x \in A}  e^{-\mu \operatorname{dist}(x,C)} \sum_{\substack{X \ni x \colon\\[1mm] X \cap C \neq \emptyset}}  e^{\lambda |X|+ \mu \diam(X)}(\|\widetilde{\Phi}^{AC}_{X} \| + \|\widetilde{\Phi}^{A}_{X} \| + \|\widetilde{\Phi}^{C}_{X} \|+\|\widetilde{\Phi}^{\emptyset}_{X} \|)\\[2mm]
& \leq   \, \frac{4 \Delta \lambda}{\lambda-2 \Delta |s|}    \sum_{x \in A}  e^{-\mu \operatorname{dist}(x,C)} \,.
\end{align*}
Next we use this inequality in Eq.\ \eqref{eq:estimate_prod_exp_eff_ham_minus_identity_first}, taking into account that $\sum_{k=1}^{\infty} x^{k}/k! = e^{x} - 1 \leq xe^{x}$ whenever $x \geq 0$, which allows to estimate 
\[
 \left\|\rho_{AC} \rho_{A}^{-1} \rho_{C}^{-1}- \mathbbm{1} \right\| \leq   \exp\left( \frac{4\Delta\lambda|\beta|}{\lambda - 2\Delta |s|} \sum_{x \in A} e^{-\mu \operatorname{dist}(x,C)}\right) \, \frac{4\Delta\lambda|\beta|}{\lambda - 2\Delta |s|} \sum_{x \in A} e^{-\mu \operatorname{dist}(x,C)} \, .
\]

\end{proof}

It is natural to ask if an analgous theorem holds assuming the weak local effective Hamiltonian property. However, in this case an additional term arises that must be controlled in other ways to have a suitable decay of the mixing condition. We will discuss this in detail in Section \ref{mixingviaweakeffective}.

\section{Local indistinguishability}\label{sec:local_indistinguishability}

This section is dedicated to elucidating the concept of local indistinguishability for Gibbs states of short-range, local Hamiltonians and its validity under the assumption of uniform clustering of correlations. While this property was previously discussed in \cite{Kastoryano2013} for finite-range interactions, it is worth noting that their proof contains certain issues pertaining to normalization factors. A similar approach to the one presented in this section, based on the so-called \emph{Quantum Belief Propagation} (QBP), has recently been explored in \cite{OnoratiRouzeStilckfrancaWatson2023LearningStatesPhases} for finite-range interactions, and has been extended to short-range interactions in full detail in \cite{CapelMoscolariTeufelWessel-LPPL-2023}. Here, we provide a concise and rigorous presentation of QBP for short-range interactions.

\subsection{Quantum Belief Propagation}\label{subsubsec:quantum_belief_propagation}
The technique of QBP was introduced in \cite{Hastings2007, leifer2008quantum}. We follow the presentation in \cite{Kato2019, kim2012perturbative}, which developed the method further for finite-range interactions, and \cite{CapelMoscolariTeufelWessel-LPPL-2023}, which adapted QBP to the setting of short-range interactions. We will only present here the necessary ingredients for the proof of local indistinguishability of Gibbs states and refer the interested reader to \cite{CapelMoscolariTeufelWessel-LPPL-2023} for the details on the proof. 

In this section, we will fix $(V,E) = \mathbb Z^g$ for some $g \in \mathbb N$ together with the Euclidean distance, as this is the setting of \cite{CapelMoscolariTeufelWessel-LPPL-2023}, although the results should carry over to more general graphs. First, let $\Phi$ be a finite-range interaction generating a Hamiltonian $H$ and let $W$ be a bounded Hermitian operator on some finite subset of the lattice $\mathbb Z^g$. QBP \cite{Hastings2007, kim2012perturbative,CapelMoscolariTeufelWessel-LPPL-2023} allows us to rewrite 
\begin{equation}\label{eq:QBPgeneralformulation}
    e^{-\beta(H+W)} = \eta(W) e^{-\beta H} \eta(W)^{\ast} \, ,
\end{equation}
where $\eta(W)$ inherits the locality properties of $W$ by means of the (real) Lieb-Robinson bounds. Let us describe $\eta(W)$ more explicitly. Denote $H(s) = H+sW$ and 
\[ \mho^{s}_{H}(W) := \int_{\mathbb{R}} dt \, f_\beta(t) \, e^{-iH(s)t} W e^{iH(s)t} \, , \]
where $f_\beta \in L^{1}(\mathbb{R})$, see \cite{Kato2019, Ejima2019} for an explicit description. Then,
\begin{equation} \label{eq:def-OW}
     \eta (W) = \sum_{m=0}^{\infty} \left(\frac{-\beta}{2} \right)^m \int_{0}^{1} ds_{1} \int_{0}^{s_{1}} d s_{2} \ldots \int_{0}^{s_{m-1}}ds_{m} \, \mho^{s_{m}}(W) \cdot \ldots \cdot \mho^{s_{1}}(W)\,.
\end{equation}
It can be shown that $\|\mho^{s}_{H}(W)\| \leq \|W\|$ and thus 
\begin{equation}\label{eq:bound-on-norm-of-OW}
    \|\eta (W) \|\leq e^{\beta \|W\|/2}.
\end{equation}
Moreover, the operator $\eta(W)$ can be well-approximated by an operator $\eta_\ell(W)$ supported on a ball of radius $\ell$ around $\mathcal W := \operatorname{supp}W$ in Euclidean distance \cite{kim2012perturbative, Kato2019}. We will call this set $\mathcal W_\ell$. In fact, $\eta_\ell(W)$ is defined as $\eta(W)$ in Eq.\ \eqref{eq:def-OW}, but with $\mho^{s}_{H}(W)$ replaced by $\mathbb{E}_{\mathcal W_\ell}[\mho^{s}_{H}(W)]$. This also shows
\begin{equation*}
     \|\eta_\ell(W)\|\leq e^{\beta \|W\|/2}.
\end{equation*}
Additionally, the distance between $\eta(W)$ and $\eta_\ell(W)$ can be estimated as 
\begin{equation}\label{eq:QBP-localisation}
   \|\eta(W) - \eta_\ell(W)\| \leq e^{c_1 \|W\|} e^{-c_2 \ell} \, .
\end{equation}
Here, $c_1$ and $c_2$ depend on $\mathcal W$ as can be seen from the proof in \cite{kim2012perturbative, Kato2019} based on Lieb-Robinson bounds, but they do not depend on the support of the Hamiltonian $H$. Something similar can be proven for short-range interactions and for the $\eta$ that appears in the normalized version of Eq.\ \eqref{eq:QBPgeneralformulation}, namely one involving Gibbs states instead of exponentials. This was done explicitly in \cite{CapelMoscolariTeufelWessel-LPPL-2023}, from which we extract the following result.

\begin{prop} \label{prop:QBP} Let $\Lambda$ be a finite lattice $\Lambda \subset \mathbb{Z}^g$, and let $H$ be a self-adjoint operator on $\mathcal{B}(\mathcal{H}_\Lambda)$ generated by short-range interactions, with $W \in \mathfrak{A}(\mathcal{H}_X)$ for $X \subset \Lambda$. Consider the path of Hamiltonians $H(s):= H + s W$ for $s\in [0,1]$.  Then,
\begin{enumerate}
    \item There exists $s \mapsto \eta (W,s)$ such that
    \begin{equation}\label{eq:prop_QBP_1}
        e^{-\beta H(s)} = \eta(W,s) e^{-\beta H} \eta(W,s)^{\ast}  \quad , \qquad \norm{\eta(W,s)} \leq e^{\frac{\beta}{2}s \norm{W}} \, .
    \end{equation}
    \item There exists $s \mapsto \tilde \eta (W,s)$ such that
     \begin{equation}\label{eq:prop_QBP_2}
        \rho_{\beta}(s) = \tilde \eta(W,s) \rho_{\beta}(0) \tilde \eta(W,s)^{\ast}  \quad , \qquad \norm{\tilde \eta(W,s)} \leq e^{{\beta} s \norm{W}} \, ,
    \end{equation}
    where $\rho_{\beta}(s)$ is the Gibbs state for $H(s)$. Moreover,
    \begin{equation}\label{eq:prop_QBP_3}
        \norm{ \rho_{\beta}(s) - \rho_{\beta}(0)}_1 \leq e^{2 \beta s \norm{W}}-1 \leq s (e^{2 \beta \norm{W}}-1) \, .
    \end{equation}
    \item There exist constants $\kappa , \gamma >0$ such that
    \begin{equation}\label{eq:prop_QBP_4}
        \norm{\tilde \eta (W,s) - \tilde \eta_\ell(W,s)} \leq \kappa \beta s  |X| \norm{W} e^{\beta s \norm{W}} e^{-\gamma \ell} \, .
    \end{equation}
    $\kappa $ and $\gamma$ are based on Lieb-Robinson bounds, and their explicit expression can be found in \cite{CapelMoscolariTeufelWessel-LPPL-2023}. {Again, $\tilde \eta_\ell(W,s)$ is supported on $\mathcal W_{\ell}$.}
\end{enumerate}
\end{prop}

In the next section, we will make use of the results collected above. Whenever the perturbation $W$ is clear from the context, we will drop the dependence of $\eta$ on it.

\subsection{Local indistinguishability}\label{subsec:local_indistinguishability}

 In this section, we discuss the notion of local indistinguishability of Gibbs states and how it arises from decay of correlations. For that, we reprove some of the main findings of \cite{CapelMoscolariTeufelWessel-LPPL-2023}, extended from \cite[Theorem 5]{Brandao2019} from finite-range to short-range interactions, in order to keep track of the constants and terms involved.  Note that the proof of the former is similar in spirit to the latter, but notably more technical. Here, we will limit ourselves to the regular lattice $\mathbb Z^g$ for some $g \in \mathbb N$, although the results could be extended to more general lattices.

Let us recall the notion of \textit{operator correlation function}, also known as \textit{covariance}. Given a finite subset $\Lambda \subset \subset \mathbb{Z}^g$, $\rho \in \mathcal{D}(\mathcal{H}_\Lambda)$ a state on $\Lambda$, subsets $A,B \subset  \Lambda$, and $O_A \in \mathfrak{A}_A$, $O_B \in \mathfrak{A}_B$, we define the covariance of $\rho$ between $A$ and $B$ as
\begin{equation*}
    \operatorname{Cov}_\rho(A,B) = \underset{\norm{O_A}=\norm{O_B}=1}{\sup} |\mathrm{Tr}[\rho \, O_A O_B] - \mathrm{Tr}[\rho \, O_A] \mathrm{Tr}[\rho \, O_B]| \, .
\end{equation*}
In the main result of this section, namely the local indistinguishability of Gibbs states, we will consider a region $\Lambda$ split into $ABC$ as in Figure \ref{fig:local_indist_construction} (left-hand side) and we will show that the effect of an observable $O_A\in \mathfrak{A}_A$ traced with respect to the Gibbs state in $ABC$, and with respect to $AB$ is almost indistinguishable {if $A$ and $C$ are sufficiently far apart}. For that, we will remove the sites of $C$ (i.e. the interactions acting on each site) one by one, and we will show that the change performed at each step is almost negligible. However, this requires the assumption that correlations between spatially separated regions decay fast enough, not only for the Gibbs state of the global Hamiltonian in $\Lambda$, but also for the Hamiltonian on each of the intermediate steps until having completely removed $C$. For simplicity, we assume a more general condition, which is inspired by \cite{Brandao2019} and contains exponential uniform clustering of covariance as a special case. 

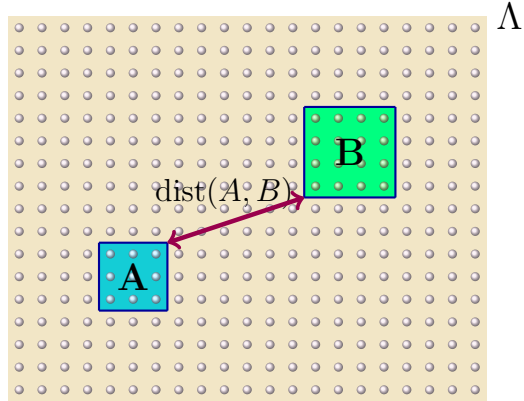
\begin{figure}[ht]
\begin{center}

\begin{tikzpicture}[scale=0.3]

\fill [darkyellow!50!orange!25!white] (-0.5,-0.5) rectangle (20.5,16.5);
\foreach \n in {1,...,21}{
\foreach \m in {1,...,17}{
\shade[shading=ball, ball color=darkred!10!white] (\n-1,\m-1) circle (0.2);
}
 }
 \begin{scope}[xshift=4cm,yshift=4cm]
 \fill [cyan!80!lightblue] (-0.5,-0.5) rectangle (2.5,2.5);
  \foreach \n in {1,...,3}{
\foreach \m in {1,...,3}{
\shade[shading=ball, ball color=lightblue!50!white] (\n-1,\m-1) circle (0.2);
}
 }
 \end{scope}
  \begin{scope}[xshift=13cm,yshift=9cm]
 \fill [lightgreen] (-0.5,-0.5) rectangle (3.5,3.5);
  \foreach \n in {1,...,4}{
\foreach \m in {1,...,4}{
\shade[shading=ball, ball color=midgreen!40!white] (\n-1,\m-1) circle (0.2);
}
 }
 \end{scope}

\node at (5,5) {\Large $\text{\textbf{A}}$};
\node at (14.5,10.5) {\Large $\text{\textbf{B}}$};
\node at (21.5,16.5) {\Large $\Lambda$};
\node at (9,8.7) {\large $\mathrm{dist}(A,B)$};
\draw [darkblue,thick](3.5,3.5) -- (3.5,6.5) ;
\draw [darkblue,thick](6.5,3.5) -- (6.5,6.5) ;
\draw [darkblue,thick](3.5,3.5) -- (6.5,3.5) ;
\draw [darkblue,thick](3.5,6.5) -- (6.5,6.5) ;
\draw [darkblue,thick](12.5,8.5) -- (12.5,12.5) ;
\draw [darkblue,thick](12.5,8.5) -- (16.5,8.5) ;
\draw [darkblue,thick](16.5,8.5) -- (16.5,12.5) ;
\draw [darkblue,thick](12.5,12.5) -- (16.5,12.5) ;
\draw [<->,darkred,thick,line width=0.6mm](6.5,6.5) -- (12.5,8.5) ;

\end{tikzpicture}

 \caption{Display of two sublattices $A$, $B$ of $\Lambda$ such that $\mathrm{dist}(A,B) \geq \ell $.}
  \label{fig:local_indist_mult}
  \end{center}
\end{figure}

\begin{defi}\label{def:uniform_clustering}
Let $\Phi$ be a local, short-range interaction, i.e. such that $\norm{\Phi}_{\lambda, \mu} < \infty$ for certain $\lambda, \mu >0$ (cf. Eq.\ \eqref{eq:norm_interaction}). Fix an inverse temperature $\beta > 0$, and for any finite $\Lambda \subset \subset \mathbb{Z}^g$, let $\rho^\Lambda_\beta$ be the Gibbs state of $H_\Lambda$ at  $\beta > 0$, defined from $\Phi$. We say that $H = (H_\Lambda)_{\Lambda \subset \subset \mathbb{Z}^g}$ is \emph{$\epsilon(\ell)$-uniform clustering} if for all $\Lambda \subset \mathbb Z^g$, and all $O_A \in \mathfrak{A}_A$, $O_B \in \mathfrak{A}_B$, where $A$, $B \subseteq \Lambda$ such that $\operatorname{dist}(A,B) \geq \ell$ (cf. Figure \ref{fig:local_indist_mult}), 
\begin{equation*}
    \operatorname{Cov}_{\rho^\Lambda}(O_A,O_B) \leq 
 {f(A,B)} \norm{O_A} \norm{O_B} \epsilon(\ell) \, .
\end{equation*}
Here, $f(A,B) \leq g(A) |B|^b$ for some $b \in \mathbb N$ and some function $g$.
\end{defi}

\begin{rem}\label{rem:discussion_clustering}
    In our definition of $\epsilon(\ell)$-uniform clustering, we leave on purpose open the dependence of the function $f$ on $A$ and $B$. That is, because our aim is to show that uniform clustering implies the mixing condition, not to prove uniform clustering. As examples, the review \cite{alhambra2022quantum} considers $f(A,B) = {\mathrm{min}\{|\partial A|, |\partial B|\}}$. For finite-range interactions, \cite{Kliesch2014} proves that form of clustering for any finite interaction hypergraph.  The article \cite{Brandao2019} limits their attention to uniform clustering with $f(A,B) = 1$, such as one-dimensional systems, for which uniform clustering was shown to hold in \cite{Araki1969} for finite-range interactions and subsequently extended to short-range interactions in \cite{Perez2020} above a threshold temperature. This is also the case for commuting Hamiltonians associated to gapped Davies Lindbladians by \cite{Kastoryano2013}.  Additionally, for short-range interactions, \cite[Theorem 3.2]{frohlich2015some} seems to show uniform clustering with $f(A,B) = \mathcal O(|A| |B|)$ for high dimensions. To unify all these different notions of uniform clustering, we hence chose to retain the freedom of choosing $f(A,B)$ appropriately. In the following proofs, however, we will need to control the growth of $f$ if one of the sets in its arguments is a ball of radius $\ell$. Therefore, we assume for simplicity in Definition \ref{def:uniform_clustering} that the second argument of $f$ behaves like a power. This is similar to the treatment in \cite{CapelMoscolariTeufelWessel-LPPL-2023} and covers the aforementioned examples.
\end{rem}

\begin{figure}[ht]
\begin{center}

\begin{tikzpicture}[scale=0.4]

\fill [darkyellow!50!orange!25!white] (-0.5,-0.5) rectangle (12.5,12.5);
\foreach \n in {1,...,13}{
\foreach \m in {1,...,13}{
\shade[shading=ball, ball color=darkred!10!white] (\n-1,\m-1) circle (0.2);
}
 }
   \begin{scope}[xshift=1cm,yshift=1cm]
 \fill [lightgreen] (-0.5,-0.5) rectangle (8.5,8.5);
  \foreach \n in {1,...,9}{
\foreach \m in {1,...,9}{
\shade[shading=ball, ball color=midgreen!40!white] (\n-1,\m-1) circle (0.2);
}
 }
 \end{scope}
 \begin{scope}[xshift=4cm,yshift=4cm]
 \fill [cyan!80!lightblue] (-0.5,-0.5) rectangle (2.5,2.5);
  \foreach \n in {1,...,3}{
\foreach \m in {1,...,3}{
\shade[shading=ball, ball color=lightblue!50!white] (\n-1,\m-1) circle (0.2);
}
 }
 \end{scope}

   \begin{scope}[xshift=19cm,yshift=0cm]
 \fill [darkyellow!50!orange!25!white] (-0.5,-0.5) rectangle (12.5,12.5);
 \fill [white] (3.5,9.5) rectangle (9.5,10.5);
  \fill [white] (9.5,10.5) rectangle (10.5,0.5);
  \foreach \n in {1,...,13}{
\foreach \m in {1,...,13}{
\shade[shading=ball, ball color=darkred!10!white] (\n-1,\m-1) circle (0.2);
}
 }
   \fill [white,opacity=0.6] (3.5,9.5) rectangle (9.5,10.5);
  \fill [white,opacity=0.6] (9.5,10.5) rectangle (10.5,0.5);
 \end{scope}
    \begin{scope}[xshift=20cm,yshift=1cm]
 \fill [lightgreen] (-0.5,-0.5) rectangle (8.5,8.5);
  \foreach \n in {1,...,9}{
\foreach \m in {1,...,9}{
\shade[shading=ball, ball color=midgreen!40!white] (\n-1,\m-1) circle (0.2);
}
 }
 \end{scope}
 \begin{scope}[xshift=23cm,yshift=4cm]
 \fill [cyan!80!lightblue] (-0.5,-0.5) rectangle (2.5,2.5);
  \foreach \n in {1,...,3}{
\foreach \m in {1,...,3}{
\shade[shading=ball, ball color=lightblue!50!white] (\n-1,\m-1) circle (0.2);
}
 }
 \end{scope}

\node at (5,5) {\huge $\text{\textbf{A}}$};
\node at (2,11) {\huge $\text{\textbf{C}}$};
\node at (3.5,8) {\huge $\text{\textbf{B}}$};
\node at (24,5) {\huge $\text{\textbf{A}}$};
\node at (24.1,11.5) {\Large $\text{\textbf{C}}\setminus \{ c_1 \cup \ldots \cup c_j\}$};
\node at (22.5,8) {\huge $\text{\textbf{B}}$};

\node at (23.1,10) {\large $c_j$};
\node at (29.1,1.2) {\large $c_1$};
\node at (29,4.2) {\Huge $\vdots$};
\node at (29,7.2) {\Huge $\vdots$};
\node at (26,10.1) {\huge $\ldots$};
\node at (8.1,5.4) {\begin{small}
$\mathrm{dist}(A,C)$
\end{small} };
\draw [darkblue,thick](3.5,3.5) -- (3.5,6.5) ;
\draw [darkblue,thick](6.5,3.5) -- (6.5,6.5) ;
\draw [darkblue,thick](3.5,3.5) -- (6.5,3.5) ;
\draw [darkblue,thick](3.5,6.5) -- (6.5,6.5) ;
\draw [darkblue,thick](0.5,0.5) -- (0.5,9.5) ;
\draw [darkblue,thick](0.5,0.5) -- (9.5,0.5) ;
\draw [darkblue,thick](9.5,0.5) -- (9.5,9.5) ;
\draw [darkblue,thick](9.5,9.5) -- (0.5,9.5) ;
\draw [darkblue,thick](-0.5,-0.5) -- (-0.5,12.5) ;
\draw [darkblue,thick](-0.5,-0.5) -- (12.5,-0.5) ;
\draw [darkblue,thick](12.5,-0.5) -- (12.5,12.5) ;
\draw [darkblue,thick](12.5,12.5) -- (-0.5,12.5) ;

\draw [darkblue,thick](22.5,3.5) -- (22.5,6.5) ;
\draw [darkblue,thick](25.5,3.5) -- (25.5,6.5) ;
\draw [darkblue,thick](22.5,3.5) -- (25.5,3.5) ;
\draw [darkblue,thick](22.5,6.5) -- (25.5,6.5) ;
\draw [darkblue,thick](19.5,0.5) -- (19.5,9.5) ;
\draw [darkblue,thick](19.5,0.5) -- (28.5,0.5) ;
\draw [darkblue,thick](28.5,0.5) -- (28.5,9.5) ;
\draw [darkblue,thick](28.5,9.5) -- (19.5,9.5) ;
\draw [darkblue,thick](18.5,-0.5) -- (18.5,12.5) ;
\draw [darkblue,thick](18.5,-0.5) -- (31.5,-0.5) ;
\draw [darkblue,thick](31.5,-0.5) -- (31.5,12.5) ;
\draw [darkblue,thick](31.5,12.5) -- (18.5,12.5) ;

\draw [darkblue,thick](22.5,10.5) -- (29.5,10.5) ;
\draw [darkblue,thick](29.5,0.5) -- (29.5,10.5) ;
\draw [darkblue,thick](28.5,0.5) -- (29.5,0.5) ;
\draw [darkblue,thick](22.5,10.5) -- (22.5,9.5) ;
\draw [<->,darkred,thick,line width=0.6mm](6.5,4.5) -- (9.5,4.5) ;

\end{tikzpicture}

 \caption{On the left-hand side, we split $\Lambda$ into $ABC$ so that $B$ is a ring around $A$, shielding it from $C$. On the right-hand side, we show the support of $\Lambda_j$, the lattice after $j$ steps of removing sites in an ordered way, one by one. }
  \label{fig:local_indist_construction}
  \end{center}
\end{figure}
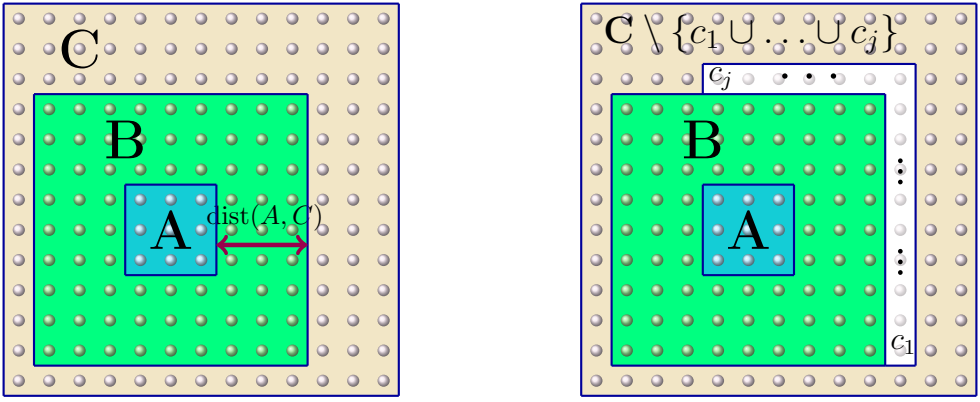

This allows us to prove a version of the following theorem for hypercubic lattices, inspired by \cite[Theorem 5]{Brandao2019} for finite-range interactions, and extended to short-range interactions in \cite{CapelMoscolariTeufelWessel-LPPL-2023}. Its proof can be derived as a combination of Theorem 14, Lemma 19 and Theorem 20 of \cite{CapelMoscolariTeufelWessel-LPPL-2023} for the particular case considered here. However, we include here a self-contained and simplified proof with easier notation for completeness. 

\begin{thm}\label{thm:local_indistinguishability_kastoryanobrandao}
Let { $V = \mathbb{Z}^{g}$ and} let $\Phi$ be a short-range interaction with $\norm{\Phi}_{\lambda, \mu} < \infty$ for certain $\lambda, \mu >0$. For any $\Lambda \subset \subset V$, let $\rho^\Lambda_\beta$ be the Gibbs state of $H_\Lambda$ at  $\beta > 0$, defined from $\Phi$. Consider a splitting of $\Lambda$ as $\Lambda=ABC$ and $\ell \in \mathbb{N}$ such that $\operatorname{dist}(A,C) \geq 2 \ell + 1 $ (see Figure \ref{fig:local_indist_construction}, left-hand side), with $H=(H_\Lambda)_{\Lambda \subset \subset V}$ being $\epsilon(\ell)$-clustering. Then, we have
\begin{align}\label{eq:local_indistinguishability_kastoryanobrandao_1norm}
    & \left| \Tr[\rho^\Lambda_\beta O_A] -  \Tr[\rho^{AB}_\beta O_A ] \right| \nonumber \\
    & \hspace{2.2cm} \leq |C| g(A) \ell^{db} K e^{2 \beta (\norm{\Phi}_{0,0}+ \norm{\Phi}_{\mu,\lambda} e^{-\mu \ell})}  \left(  4 \kappa' \beta \norm{\Phi}_{0,0} \ell^{d} e^{-\gamma \ell}  + \epsilon(\ell) +  \norm{\Phi}_{\mu,\lambda} e^{-\mu \ell} \right) \norm{O_A} \, ,
\end{align}
where  $K$, $\kappa'$, $\gamma >0 $ are constants. Thus, in particular if $\epsilon(\ell)$ is at least exponentially decreasing, the above can be simplified to
\begin{align}\label{eq:local_indistinguishability_simplified}
     \left| \Tr[\rho^\Lambda_\beta O_A] -  \Tr[\rho^{AB}_\beta O_A ] \right| \leq |C| g(A) \mathcal{K (\beta)}  \norm{O_A}  e^{-\alpha \ell} \, ,
\end{align}
for certain constants $\alpha , \mathcal{K}(\beta)>0$. 
\end{thm} 

\begin{rem}
    The factor $|C|$ in Eq.\ \eqref{eq:local_indistinguishability_kastoryanobrandao_1norm} is due to the number of steps performed to remove all interactions with support intersecting $C$. In the case of finite-range interactions with range $r$, to prove local indistinguishability it is enough to decouple the interactions with support in $AB$ from those with a disjoint support. Therefore, in that case it is enough to remove the sites of $|\partial_r C|$, since 
    \begin{equation*}
        \Tr[\rho^{AB}_\beta\otimes \rho_\beta^{C \setminus \partial_r C} O_A] =  \Tr[\rho_\beta^{AB} O_A]
    \end{equation*}
    for every $O_A \in \mathfrak{A}_A$, and thus the dependence on $C$ in Eq.\ \eqref{eq:local_indistinguishability_kastoryanobrandao_1norm} can be tightened to $|\partial_r C|$. 
\end{rem}

\begin{proof}[Proof of {Theorem \ref{thm:local_indistinguishability_kastoryanobrandao}}]
Let us drop hereafter the subscript $\beta$ in $\rho$ to ease notation. Let us denote $m=|C|$ and let us enumerate by $c_j$, for $j \in [m]$, the sites in $ C$. The idea is to remove in an ordered way these sites, one by one, from the interactions in the Hamiltonian, and to use QBP and uniform clustering to show that the change when doing this is small. We will write $\Lambda_k$ for the remaining lattice after removing $k$ sites, i.e. $\Lambda_k = \Lambda \setminus \bigcup_{j=1}^k \{c_j\}$ (see Figure \ref{fig:local_indist_construction}, right-hand side). In particular, $\Lambda_0 = \Lambda$ and $\Lambda_m = AB$.  Then, for any $O_A \in \mathfrak{A}_A$, we have 
\begin{align}\label{eq:triangle_sum_local_ind}
    \left| \Tr[\rho^\Lambda O_A ] - \Tr[\rho^{AB} O_A ] \right| & \leq \sum_{j = 0}^{m-1} \left| \Tr[\rho^{\Lambda_j} O_A ] - \Tr[\rho^{\Lambda_{j+1}} O_A ] \right| ,
\end{align}
For each of these terms, we denote:
\begin{equation*}
    W_j := H_{\Lambda_{j+1}} - H_{\Lambda_j}  \, ,
\end{equation*}
and we split it as
\begin{equation*}
     W_{j,1} := -\underset{\substack{c_j \in Z \subset \Lambda \\\text{diam}(Z)\leq R}}{\sum} \, \Phi(Z) \quad , \quad \quad  W_{j,2} := -\underset{\substack{c_j \in Z \subset \Lambda \\\text{diam}(Z) > R}}{\sum} \, \Phi(Z) \, ,
\end{equation*}
for $R:= 1/2 \lfloor {\mathrm{dist}(A,C)}-1\rfloor$. Since $W_{j,1}$ is supported in an $R$-ball centered at $\{c_j\}$, we have 
\begin{equation*}
    \norm{W_{j,1}} \leq \norm{\Phi}_{0,0} \, .
\end{equation*}
Additionally, for the remaining part, we know that it is small, namely
\begin{equation*}
    \norm{W_{j,2}} \leq \norm{\Phi}_{\mu,\lambda} e^{-\mu R} \, .
\end{equation*}
Therefore, $\norm{W_j} \leq  \norm{\Phi}_{0,0} + \norm{\Phi}_{\mu,\lambda} e^{-\mu R}$. We note that denoting by $\hat \rho^{\Lambda_{j}}$ the Gibbs state on $\Lambda_j$, but with the Hamiltonian $H_{\Lambda_j}+W_{j,1}$, we can bound
\begin{equation*}
     \left| \Tr[\rho^{\Lambda_j} O_A ] - \Tr[\rho^{\Lambda_{j+1}} O_A ] \right| \leq \left| \Tr[\rho^{\Lambda_j} O_A ] - \Tr[\hat \rho^{\Lambda_{j}} O_A ] \right| + 2 \beta \|O_A\| \|W_{j,2} \| e^{2 \beta \|W_{j,2}\|} \, ,
\end{equation*}
where we have used Eq.\ \eqref{eq:prop_QBP_3} and $e^x-1 \leq x e^x$ for $x \geq 0$.

Moreover, by Eq.\ \eqref{eq:prop_QBP_2}, and denoting  $\tilde \eta^j :=\tilde \eta (W_{j,1},1) $ and $\tilde \eta^j_R :=\tilde \eta_R (W_{j,1},1) $, each term in the right-hand side of Eq.\ \eqref{eq:triangle_sum_local_ind} can be bounded as
\begin{align*}
     \left| \Tr[\rho^{\Lambda_j} O_A ] - \Tr[\hat \rho^{\Lambda_{j}} O_A ] \right|& =  \left| \operatorname{Tr}[ \tilde \eta^j \rho^{\Lambda_{j}} \tilde \eta^{j,*}  O_A] - \Tr[\rho^{\Lambda_{j}} O_A ] \right| \\[1mm] 
    & \leq  \underbrace{ \left|\operatorname{Tr}[ (\tilde \eta^j - \tilde \eta^j_R)\rho^{\Lambda_{j}} \tilde \eta^{j,*}  O_A]  - \operatorname{Tr}[ \tilde \eta^j_{{R}} \rho^{\Lambda_{j}} (\tilde \eta^{j,*} - \tilde \eta^{j,*}_R) O_A] \right|}_{I} \\
    &\quad + \underbrace{\left|\operatorname{Tr}[\rho^{\Lambda_{j}} \tilde \eta^{j,*}_R \tilde \eta^{j}_R O_A ] -  \operatorname{Tr}[\rho^{\Lambda_j} O_A ] \right|}_{II} \, . 
\end{align*}
For $I$, note that
\begin{equation*}
    I \leq \big\| \tilde \eta^j - \tilde \eta^j_R \big\| \norm{O_A} \left( \norm{ \tilde \eta^j} + \big\| \tilde \eta^j_R \big\| \right) \, , 
\end{equation*}
and for $II$, we have
\begin{align*}
    II&  = \operatorname{Tr}[\rho^{\Lambda_{j}} \tilde \eta^{j,*}_R \tilde \eta^{j}_R O_A ] -  \operatorname{Tr}[\rho^{\Lambda_{j}} \tilde \eta^{j,*}_R \tilde \eta^{j}_R] \operatorname{Tr}[\rho^{\Lambda_{j}}  O_A ] +  \operatorname{Tr}[\rho^{\Lambda_{j}} \tilde \eta^{j,*}_R \tilde \eta^{j}_R] \operatorname{Tr}[\rho^{\Lambda_{j}}  O_A ]- \operatorname{Tr}[\rho^{\Lambda_j} O_A ]\\[1mm] 
    & \leq \text{Cov}_{\rho^{\Lambda_{j}}}\left(\tilde \eta^{j,*}_R \tilde \eta^{j}_R , O_A \right) + \left( \operatorname{Tr}[\rho^{\Lambda_{j}} \tilde \eta^{j,*}_R \tilde \eta^{j}_R]  - 1 \right) \operatorname{Tr}[\rho^{\Lambda_j} O_A ] \\[1mm] 
    & \leq\text{Cov}_{\rho^{\Lambda_{j}}}\left(\tilde \eta^{j,*}_R \tilde \eta^{j}_R , O_A \right) + \big\| \tilde \eta^j - \tilde \eta^j_R \big\| \norm{O_A} \left( \norm{ \tilde \eta^j} + \big\| \tilde \eta^j_R \big\| \right) \, .
\end{align*}
{In the last inequality, the second estimate comes from a similar argument as in the bound on $I$.} Let $\mathcal W^{1}_{j}$ be the support of $W_{j,1}$ and $\mathcal W^{1}_{j,\ell}$ be the set of points with distance at most $\ell$ to $\mathcal W^{1}_{j}$. Combining the estimates of $I$ and $II$, the $\epsilon(\ell)$-clustering with $\ell= R$ and Eq.\ \eqref{eq:prop_QBP_4}, we get 
\begin{align*}
     \left| \Tr[\rho^{\Lambda_j} O_A ] - \Tr[\hat \rho^{\Lambda_{j}} O_A ] \right| & \leq \text{Cov}_{\rho^{\Lambda_{j}}}\left(\tilde \eta^{j,*}_R \tilde \eta^{j}_R , O_A \right) + 2 \big\| \tilde \eta^j - \tilde \eta^j_R \big\| \norm{O_A} \left( \norm{ \tilde \eta^j} + \big\| \tilde \eta^j_R \big\| \right) \\[1mm] 
     & \leq f(\mathcal{W}^1_{j,R} , A) \big\| \tilde \eta^{j,*}_R \tilde \eta^{j}_R \big\| \norm{O_A}  \epsilon(\ell) + 4 \kappa \beta |\mathcal{W}^1_{j,R}| \norm{W_{j,1}}  e^{2 \beta \norm{W_{j,1}}}  \norm{O_A}e^{-\gamma \ell} \\[1mm] 
     & \leq  \max \{ f(\mathcal{W}^1_{j,R} , A) ,   1 \} \, e^{2 \beta \norm{W_{j,1}}} \norm{O_A} \left(  4 \kappa \beta | \mathcal{W}^1_{j,R}| \norm{W_{j,1}} e^{-\gamma \ell}  + \epsilon(\ell) \right) \, .
\end{align*}
Finally, summing this over all sites of $C$, we get in Eq.\ \eqref{eq:triangle_sum_local_ind}:
\begin{align*}
     &\left| \Tr[\rho^\Lambda O_A ] - \Tr[\rho^{AB} O_A ] \right| \leq \\& \quad \max_j |C| \max \{ f(\mathcal{W}^1_{j,R} , A) ,   1 \} \, e^{2 \beta (\norm{W_{j,1}}+\norm{W_{j,2}})} \norm{O_A} \left(  4 \kappa \beta | \mathcal{W}^1_{j,R}| \norm{W_{j,1}} e^{-\gamma \ell}  + \epsilon(\ell) + \| W_{j,2}\| \right) \, .
\end{align*}
We conclude by noting the explicit bound for $f(\mathcal{W}^1_{j,R} , A)$ from Definition \ref{def:uniform_clustering} and recalling $R=\ell$:
\begin{equation*}
    f(\mathcal{W}^1_{j,R} , A)\leq g(A) |\mathcal{W}^1_{j,R}|^b 
\end{equation*}
as well as the fact that, as $\mathcal W^1_j$ is contained in an $\ell$-ball centered at $c_j$,
\begin{equation*}
    |\mathcal{W}^1_{j,R}|^b = \left( \frac{\pi^{d/2}}{\Gamma \left( \frac{d}{2} +1\right)} (2\ell)^d \right)^b \, ,
\end{equation*}
where $\Gamma$ is Euler's gamma function. 
\end{proof}

In the next section, we will use the previous theorem to prove a mixing condition for the Gibbs state assuming exponential uniform decay of correlations. 

{
\section{Mixing condition via weak effective Hamiltonian }\label{mixingviaweakeffective}

 The main result of this section is a mixing condition for Gibbs states of local, short-range Hamiltonians under the assumption of a weak local effective Hamiltonian such as the one presented in Section \ref{sec:effective_hamiltonian}.

 \begin{prop}[Weak form]\label{prop:mixingConditionWeakEffectiveHamiltonian}
Let us assume that $\Phi$ is an interaction on $V$ satisfying the \emph{weak} local effective Hamiltonian property at $\beta >0$ (cf. \Cref{defi:weakeffHamiltonian}), and assume that there is a uniform bound $\Delta >0$ and $\lambda, \mu >0$ such that, for every $L \subset V$, the local interaction $\widehat{\Phi}^{L, \beta}$ satisfies
\[ \|\widehat{\Phi}^{L, \beta}\|_{\lambda, \mu} = \sup_{x \in V} \sum_{X \ni x}\|\Phi_{X}\|e^{\lambda |X| + \mu \operatorname{diam}(X)} \leq \Delta \,. \]
Then, for every $\Lambda \in \mathcal{P}_{f}(V)$ split into three disjoint subsets $\Lambda =ABC$, the local Gibbs state $\rho = \rho^{\Lambda}_{\beta}$ satisfies whenever $\beta < \lambda/(2\Delta)$
\[
\left\| \rho_{AC} \rho_{A}^{-1} \otimes \rho_{C}^{-1} - \mathbbm{1} \right\| \leq   \exp\left( \frac{3\Delta\lambda \beta}{\lambda - 2\Delta \beta} \sum_{x \in A} e^{-\mu \operatorname{dist}(x,C)}\right) \, \left( \frac{3\Delta\lambda \beta}{\lambda - 2\Delta \beta} \sum_{x \in A} e^{-\mu \operatorname{dist}(x,C)}  + \abs{\kappa_{ABC} - 1 }\right)\,,
\]
where 
\begin{equation}\label{equa:kappaConstant}
\kappa_{ABC} = \kappa_{ABC}(\beta) = \frac{Z_{B} \cdot Z_{\Lambda}}{Z_{BC} \cdot Z_{AB}}\,. 
\end{equation}
\end{prop}

\begin{rem}
Comparing to Theorem \ref{prop:mixingConditionStrongEffectiveHamiltonian}, observe that we have an extra summand $|\kappa_{ABC} -1|$ as a consequence of the weak local effective Hamiltonian assumption.
\end{rem}

\begin{rem}\label{rema:mixingConditionStrongEffectiveHamiltonianR1}
If $V = \mathbb{Z}^{g}$, then we can use the notation from Remark \ref{rem:onion} and rewrite the above estimation as
\[ \| \rho_{AC} \rho_{A}^{-1} \otimes \rho_{C}^{-1} - \mathbbm{1} \| \leq \exp \{ 3\beta K  |\partial A| e^{- (\mu/2) \operatorname{dist}(A,C) }  \} \cdot  \left( 3\beta K |\partial A| e^{- (\mu/2) \operatorname{dist}(A,C) } + |\kappa_{ABC} - 1| \right)  \,,\]
where $K = \frac{\Delta \lambda \nu}{\lambda - 2 \Delta \beta}  $. Exchanging the roles of $A$ and $C$, we could write $|\partial C|$ instead of $|\partial A|$, taking the minimum of both values to minimize the expression. In any case, we have an exponential decay on the distance between $A$ and $C$.
\end{rem}

\begin{proof}[Proof of Proposition \ref{prop:mixingConditionWeakEffectiveHamiltonian}]
Given $\Lambda$ and $A, C \subset \Lambda$ as in the statement of the proposition, let us denote $B = \Lambda \setminus (A \cup C)$ we omit the subscript $\Lambda$ from the local effective interactions $\widehat{\Phi}^{L} = \widehat{\Phi}^{L, \beta}$ and the effective Hamiltonian, so that (see Remark \ref{rem:relation_effrho_usualrho})
\begin{equation*}
    \rho_{AC} \, \rho^{-1}_A \otimes \rho^{-1}_C  = \operatorname{e}^{- \beta \widehat{H}^{AC}}\operatorname{e}^{\beta \widehat{  H}^{A}}\operatorname{e}^{\beta \widehat{  H}^{C}} \frac{Z_{B} \cdot Z_{\Lambda}}{Z_{BC} \cdot Z_{AB}}\ = \operatorname{e}^{- \beta \widehat{H}^{AC}}\operatorname{e}^{\beta \widehat{  H}^{A}}\operatorname{e}^{\beta \widehat{  H}^{C}} \kappa_{ABC} \, .
\end{equation*}
Then, we clearly have 
\begin{equation}\label{equa:auxMixingviaWeakEffective}
\begin{split}
    \norm{\rho_{AC} \, \rho_{A}^{-1} \otimes \rho_{C}^{-1} - \mathbbm{1}_{AC}}  & =\norm{ \operatorname{e}^{- \beta \widehat{H}^{AC}}\operatorname{e}^{\beta ( \widehat{  H}^{A} + \widehat{  H}^{C}) }\kappa_{ABC} - \mathbbm{1}_{AC} } \\[2mm]
    & \hspace{-8mm} \leq \norm{ \operatorname{e}^{- \beta \widehat{H}^{AC}}\operatorname{e}^{\beta ( \widehat{  H}^{A} + \widehat{  H}^{C}) }} \abs{\kappa_{ABC} - 1 } + \norm{ \operatorname{e}^{- \beta \widehat{H}^{AC}}\operatorname{e}^{\beta ( \widehat{  H}^{A} + \widehat{  H}^{C}) } - \mathbbm{1}_{AC}}  \, .
    \end{split}
\end{equation}
To deal with the product of exponential, we will argue as in the proof of Theorem \ref{prop:mixingConditionStrongEffectiveHamiltonian} and use Eq. \eqref{equa:mixingviaStrongEffectiveAux2}. Indeed, we can rewrite
\[ \operatorname{e}^{- \beta \widehat{H}^{AC}}\operatorname{e}^{\beta ( \widehat{  H}^{A} + \widehat{  H}^{C}) } = e^{-\beta Q} e^{\beta(Q+W)} \]
where $Q = \widehat{H}^{AC} 
$ and
\[ W= \widehat{H}^{A} + \widehat{H}^{C}  - \widehat{H}^{AC} =  \sum_{X \cap A \neq \emptyset} \widehat{\Phi}_{X}^{A} +  \sum_{X \cap C \neq \emptyset} \widehat{\Phi}_{X}^{C} -  \sum_{X \cap AC \neq \emptyset} \widehat{\Phi}_{X}^{AC} \,.\]
We claim that if $X \cap A = \emptyset$ or $X \cap C = \emptyset$, then $\widehat{\Phi}^{A}_{X} + \widehat{\Phi}^{C}_{X}- \widehat{\Phi}^{AC}_{X} = 0$. Let us prove the claim when $X \cap C = \emptyset$. In this case, $X \cap A = X \cap AC$, and therefore $\widehat{\Phi}^{A}_{X} = \widehat{\Phi}_{X}^{AC}$ and $\widehat{\Phi}^{C}_{X} = 0$, so the sum is obviously zero. The case $X \cap A = \emptyset$ is analogous.   As a consequence of this claim, only summands over $X$ with $X \cap A \neq \emptyset$ and $X \cap C \neq \emptyset$ will survive:
\[ W = \sum_{X \subset \Lambda \colon X \cap A \neq  \emptyset \,, \, X \cap C \neq \emptyset} \widehat{\Phi}^{A}_{X} + \widehat{\Phi}^{C}_{X} - \widehat{\Phi}^{AC}_{X}\,.\]
Thus, we can estimate by Proposition \ref{theo:localityEstimates} for $|s| \leq \beta  < \lambda/(2\Delta)$
\begin{align*}
\| \Gamma_{Q}^{is}(W)\| & \leq \sum_{X \subset \Lambda \colon X \cap A \neq \emptyset, X \cap C \neq \emptyset} ( \|\Gamma_{Q}^{is}(\widehat{\Phi}^{A}_{X}) \| +\|\Gamma_{Q}^{is}(\widehat{\Phi}^{C}_{X}) \|
+\|\Gamma_{Q}^{is} (\widehat{\Phi}^{AC}_{X}) \| ) \\
& \leq \frac{3\Delta\lambda}{\lambda -2\Delta|s|} \sum_{x \in A} e^{-\mu \operatorname{dist}(x,C)} \,, 
\end{align*}
Therefore,
 using that $\sum_{k=1}^{\infty} \frac{x^k}{k!} =e^{x}-1 \leq xe^{x}$ for $x>0$,  we can estimate
\begin{align*}\label{eq:estimate_prod_exp_eff_ham_minus_identity_first2}
     \norm{ \operatorname{e}^{- \beta \widehat{H}^{AC}}\operatorname{e}^{\beta ( \widehat{  H}^{A} + \widehat{  H}^{C}) } - \mathbbm{1}} =  \left\| e^{-\beta Q}e^{\beta(Q+W)} - \mathbbm{1}\right\| \leq \sum_{m=1}^{\infty} \frac{|\beta|^{m}}{m!} \, \left(\sup_{0 \leq s \leq \beta}\| \Gamma^{is}_{Q}(W)\|\right)^{m} \\ 
     \leq \exp\left( \frac{3\Delta\lambda|\beta|}{\lambda - 2\Delta |s|} \sum_{x \in A} e^{-\mu \operatorname{dist}(x,C)}\right) \, \frac{3\Delta\lambda|\beta|}{\lambda - 2\Delta |s|} \sum_{x \in A} e^{-\mu \operatorname{dist}(x,C)}\,,
\end{align*}
and also
\[  \norm{ \operatorname{e}^{- \beta \widehat{H}^{AC}}\operatorname{e}^{\beta ( \widehat{  H}^{A} + \widehat{  H}^{C}) } } \leq \sum_{m=0}^{\infty} \frac{|\beta|^{m}}{m!} \, \left(\sup_{0 \leq s \leq \beta}\| \Gamma^{is}_{Q}(W)\|\right)^{m} \leq \exp\left( \frac{3\Delta\lambda|\beta|}{\lambda - 2\Delta |s|}  \sum_{x \in A} e^{-\mu \operatorname{dist}(x,C)}\right)\,.  \]
 Applying these estimates on \eqref{equa:auxMixingviaWeakEffective}, we conclude the result.
\end{proof}

 To deal with the summand $|\kappa_{ABC} - 1|$ in Proposition \ref{prop:mixingConditionStrongEffectiveHamiltonian}, we prove the following lemma. In the proof of this result, we will use the technique of Quantum Belief Propagation recalled in Section \ref{sec:local_indistinguishability}, as well as the assumption that uniform clustering of correlations as in Definition \ref{def:uniform_clustering} holds with exponential decay.

\begin{lem}\label{lem:clusteringImpliesPartitionFunction}
Let $V = \mathbb{Z}^{g}$ and let $\Phi$ be a local interaction on $V$ satisfying, for some $\lambda, \mu >0$
\[ \|\Phi\|_{\lambda, \mu} = \sup_{x \in V} \sum_{X \ni x} \| \Phi_{X}\| e^{\lambda |X| + \mu \diam(X)}\,. \]
Let us assume that  for the inverse temperature $0<\beta< \lambda/(2\|\Phi\|_{\lambda, \mu})$, the family of Hamiltonians $H = (H_{\Lambda})_{\Lambda \subset \subset \mathbb{Z}^{g}}$ is $\epsilon(\ell)$-clustering for an exponentially decaying function $\epsilon(\ell)$.

 Then, there exist constants $\widehat{K},c>0$, such that for every subset $\Lambda \subset \subset \mathbb{Z}^{g}$ and every pair of disjoint subsets $A,C \subset \Lambda$, we have that the constant $\kappa_{ABC} =\kappa_{ABC}(\beta)  $ introduced in \eqref{equa:kappaConstant} satisfies:
\begin{equation}\label{eq:step_lambda-1}
    \abs{\kappa_{ABC}-1}  \leq  \widehat{K} \operatorname{e}^{- c \, \mathrm{dist}(A,C)} \,.
\end{equation}

\end{lem}

\begin{proof}
Here, we follow similar steps as those in the proof of \cite[Theorem 8.2]{Bluhm2021exponential}. Let us denote $\Lambda = ABC$, so that $B$ shields $A$ from $C$, see Figure \ref{fig:step2new}. First, note that we can rewrite $\kappa_{ABC}$ as 
\begin{align*} 
\kappa_{ABC} & = \, \Tr_{ABC}\left[ e^{-\beta H_{ABC}} \right] \, \Tr_{BC}\left[ e^{-\beta H_{BC}} \right]^{-1} \, \Tr_{AB}\left[ e^{-\beta H_{AB}} \right]^{-1} \, \Tr_{B}\Big[e^{-\beta H_{B}}\Big]\\[2mm] 
& = \, \Tr_{ABC}\big[\rho_\beta^{ABC} \, e^{\beta H_{ABC}} \, e^{-\beta ( H_{A} +H_{BC}) }\big]^{-1} \, \operatorname{Tr}_{AB}\left[ \rho_\beta^{AB} \, e^{\beta H_{AB}} \, e^{-\beta (H_{A} + H_{B})}\right]\, \\[2mm]
& = \, \Tr_{ABC}\big[\rho_\beta^{ABC} E_{A, BC}^{\ast \, -1}\big]^{-1} \, \Tr_{AB}\big[\rho_\beta^{AB} \, E_{A,B}^{\ast \, -1}\big] \, , 
\end{align*}
where we recall that we are denoting $\rho^{ABC}=e^{-\beta H_{ABC}}/Z_{ABC}$ and $\rho^{AB}=e^{-\beta H_{AB}}/Z_{AB}$. Note that, since we are using the same $\beta>0$ throughout the whole proof, we are dropping the explicit dependence of $\rho^X$, $E_{X,Y}$ and $Z_X$ on it, for every $X,Y \subset ABC$. Denoting $\ell := \lfloor \frac{\mathrm{dist}(A,C)}{2} \rfloor$, let us split $B$ into $B_1 B_2$ so that:
\begin{itemize}
    \item $B_2$ shields $C$ from $A B_1$. 
    \item dist$(A, B_2) \geq \ell$.
\end{itemize}
A possible construction for $B_1 B_2 $ is shown in Figure \ref{fig:step2new}.

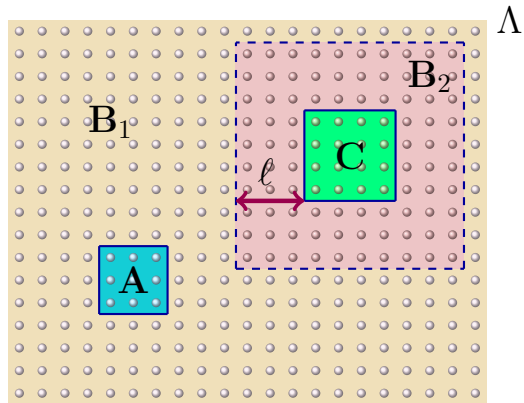
\begin{figure}[ht]
\begin{center}

\begin{tikzpicture}[scale=0.3]

\fill [darkyellow!50!orange!30!white] (-0.5,-0.5) rectangle (20.5,16.5);
\foreach \n in {1,...,21}{
\foreach \m in {1,...,17}{
\shade[shading=ball, ball color=darkred!10!white] (\n-1,\m-1) circle (0.2);
}
 }
 \begin{scope}[xshift=4cm,yshift=4cm]
 \fill [cyan!80!lightblue] (-0.5,-0.5) rectangle (2.5,2.5);
  \foreach \n in {1,...,3}{
\foreach \m in {1,...,3}{
\shade[shading=ball, ball color=lightblue!50!white] (\n-1,\m-1) circle (0.2);
}
 }
 \end{scope}
 \begin{scope}[xshift=10cm,yshift=6cm]
 \fill [darkred!60!orange!30!white] (-0.5,-0.5) rectangle (9.5,9.5);
 \foreach \n in {1,...,10}{
\foreach \m in {1,...,10}{
\shade[shading=ball, ball color=darkred!60!orange!40!white] (\n-1,\m-1) circle (0.2);
}
 }
 \end{scope}

  \begin{scope}[xshift=13cm,yshift=9cm]
 \fill [lightgreen] (-0.5,-0.5) rectangle (3.5,3.5);
  \foreach \n in {1,...,4}{
\foreach \m in {1,...,4}{
\shade[shading=ball, ball color=midgreen!40!white] (\n-1,\m-1) circle (0.2);
}
 }
 \end{scope}

\node at (5,5) {\Large $\text{\textbf{A}}$};
\node at (14.5,10.5) {\Large $\text{\textbf{C}}$};
\node at (18,14) {\Large $\text{\textbf{B}}_2$};
\node at (4,12) {\Large $\text{\textbf{B}}_1$};
\node at (21.5,16.5) {\Large $\Lambda$};
\node at (10.8,9.8) {\Large $\ell$};
\draw [dashed,darkblue,thick](9.5,5.5) -- (9.5,15.5) ;
\draw [dashed,darkblue,thick](9.5,5.5) -- (19.5,5.5) ;
\draw [dashed,darkblue,thick](19.5,5.5) -- (19.5,15.5) ;
\draw [dashed,darkblue,thick](19.5,15.5) -- (9.5,15.5) ;
\draw [darkblue,thick](3.5,3.5) -- (3.5,6.5) ;
\draw [darkblue,thick](6.5,3.5) -- (6.5,6.5) ;
\draw [darkblue,thick](3.5,3.5) -- (6.5,3.5) ;
\draw [darkblue,thick](3.5,6.5) -- (6.5,6.5) ;
\draw [darkblue,thick](12.5,8.5) -- (12.5,12.5) ;
\draw [darkblue,thick](12.5,8.5) -- (16.5,8.5) ;
\draw [darkblue,thick](16.5,8.5) -- (16.5,12.5) ;
\draw [darkblue,thick](12.5,12.5) -- (16.5,12.5) ;
\draw [<->,darkred,thick,line width=0.6mm](9.5,8.5) -- (12.5,8.5) ;
\end{tikzpicture}
 \caption{Splitting of $B$ into $B_1$ and $B_2$ devised for the proof of  \Cref{lem:clusteringImpliesPartitionFunction}.} 
  \label{fig:step2new}
  \end{center}
\end{figure}
Next, note that $|\kappa_{ABC}-1|$ can be bounded by
\begin{align*}
      \abs{\kappa_{ABC}-1} & = \abs{\Tr_{ABC}\big[\rho^{ABC} E_{A, BC}^{\ast \, -1}\big]^{-1} \, \Tr_{AB}\big[\rho^{AB} \, E_{A,B}^{\ast \, -1}\big] - 1} \\
    & \leq \abs{\Tr_{ABC}\big[\rho^{ABC} E_{A, BC}^{\ast \, -1}\big]^{-1} } \abs{\Tr_{AB}\big[\rho^{AB} \, E_{A,B}^{\ast \, -1}\big] - \Tr_{ABC}\big[\rho^{ABC} E_{A, BC}^{\ast \, -1}\big] } \\
    & \leq \| E_{A, BC}^{\ast \, -1}\| \abs{\Tr_{AB}\big[\rho^{AB} \, E_{A,B}^{\ast \, -1}\big] - \Tr_{ABC}\big[\rho^{ABC} E_{A, BC}^{\ast \, -1}\big] } \\
    & \leq e^{\beta K |\partial A|  } \abs{\Tr_{AB}\big[\rho^{AB} \, E_{A,B}^{\ast \, -1}\big] - \Tr_{ABC}\big[\rho^{ABC} E_{A, BC}^{\ast \, -1}\big] } \, ,
\end{align*}
where we have used Corollary \ref{cor:estimate_expansional}. Next, we add and subtract some intermediate terms in the previous difference, which allows us to bound:
\begin{align*}
 \left| \Tr_{AB}\big[\rho^{AB} E_{A,B}^{\ast \,-1}\big] \, - \, \Tr_{ABC}\big[\rho^{ABC} E_{A,BC}^{\ast \,-1}\big] \right|&  \leq \,\, \left| \Tr_{AB}\big[\rho^{AB} E_{A,B}^{\ast \,-1}\big] \, - \, \Tr_{AB}\big[\rho^{AB} E_{A,B_{1}}^{\ast \,-1}\big] \right| \, \\[1mm]
&  \quad + \left| \Tr_{AB}\big[\rho^{AB} E_{A,B_{1}}^{\ast \,-1}\big] \, - \, \Tr_{ABC}\big[\rho^{ABC} E_{A,B_{1}}^{\ast \,-1}\big] \right|  \\[1mm]
& \quad  + \left| \Tr_{ABC}\big[\rho^{ABC} E_{A,B_{1}}^{\ast \,-1}\big] \, - \, \Tr_{ABC}\big[\rho^{ABC} E_{A,BC}^{\ast \,-1}\big] \right|\,.
\end{align*}
The first and third terms are bounded using estimates for the expansionals. Indeed, by Hölder's inequality and the simplified bound of \Cref{prop:estimates_expansionals_normal} from Eq. \eqref{eq:simplified_estimates_difference_expansionals}, note that 
\begin{align*}
    \left| \Tr_{AB}\big[\rho^{AB} E_{A,B}^{\ast \,-1}\big] \, - \, \Tr_{AB}\big[\rho^{AB} E_{A,B_{1}}^{\ast \,-1}\big] \right| & \leq \norm{\rho^{AB}}_1 \norm{E_{A,B}^{\ast \,-1}\, - \, E_{A,B_{1}}^{\ast \,-1} } \\[2mm]
    & \leq  e^{ |\beta| K |\partial A| } K ' |\partial A| e^{- (\mu/2) \operatorname{dist}(A,B_{2}) } \, ,
\end{align*}
and analogously 
\begin{align*}
    &\left| \Tr_{ABC}\big[\rho^{ABC} E_{A,B_{1}}^{\ast \,-1}\big] \, - \, \Tr_{ABC}\big[\rho^{ABC} E_{A,BC}^{\ast \,-1}\big] \right| \leq e^{ |\beta| K |\partial A| } K ' |\partial A| e^{- (\mu/2) \operatorname{dist}(A,B_{2}) } \, .
\end{align*}
Let us bound the remaining term using Theorem \ref{thm:local_indistinguishability_kastoryanobrandao} and Eq.\ \eqref{eq:simplified_estimates_expansionals}. For that, since we are assuming uniform exponential decay of correlations, there exist constants $\alpha>0$ and $\mathcal{K}(\beta)>0$ such that
\begin{equation*}
    \left| \Tr_{AB}\big[\rho^{AB} E_{A,B_{1}}^{\ast \,-1}\big] \, - \, \Tr_{ABC}\big[\rho^{ABC} E_{A,B_{1}}^{\ast \,-1}\big] \right| \leq |C| g(A)  \mathcal{K}(\beta)  \operatorname{e}^{-\alpha \mathrm{dist}(A,B_2)} e^{\beta  K  |\partial A|   } \, .
\end{equation*}
Putting together these three estimates, and taking $c = \min\{ \mu/2, \alpha\}$ we get that
\[ \abs{\kappa_{ABC}-1} \leq 3 e^{2 \beta K |\partial A|} \left( 2K' |\partial A| + |C| g(A) \mathcal{K}(\beta)\right)  e^{-c \operatorname{dist}(A, B_{2})}\,.\]
This finishes the proof.
\end{proof}

Combining Proposition \ref{prop:mixingConditionStrongEffectiveHamiltonian}, see Remark \ref{rema:mixingConditionStrongEffectiveHamiltonianR1}, and Lemma \ref{lem:clusteringImpliesPartitionFunction}, we conclude the following main result of the section.

\begin{theo}\label{thm:Weakimpliesmixingcondition}
Let $\Phi$ be a local interaction on $V = \mathbb{Z}^{g}$ satisfying for some $\lambda, \mu, \Delta > 0$
\[ \| \Phi\|_{\lambda, \mu} = \sup_{x \in V} \sum_{X \ni x} \| \Phi_{X} \| e^{\lambda |X| + \mu \diam(X)} \leq \Delta\,.  \]
Moreover, let $0 <\beta < \lambda/(2 \Delta)$ be an inverse temperature such that:
\begin{itemize}
\item There is a weak local effective Hamiltonian at temperature $\beta>0$, and for every $L \subset V$, the local interaction $\widehat{\Phi}^{L, \beta}$ satisfies
\[ \|\widehat{\Phi}^{L, \beta}\|_{\lambda, \mu} = \sup_{x \in V} \sum_{X \ni x} \| \widehat{\Phi}_{X}^{L, \beta}\| e^{\lambda |X| + \mu \diam(X)} \leq \Delta\,.\]
\item $\Phi$ satisfies $\epsilon(\ell)$-clustering property.
\end{itemize}
Then, there exists constants $\widehat{K}', c'>0$  such that for every $\Lambda \in \mathcal{P}_{f}(V)$ and every pair of disjoint subsets $A,C \subset \Lambda$, the local Gibbs state $\rho = \rho^{\Lambda}_{\beta}$ satisfies 
\[
\left\| \rho_{AC} \rho_{A}^{-1} \otimes \rho_{C}^{-1} - \mathbbm{1} \right\| \leq   \widehat{K}' e^{-c' \operatorname{dist}(A,C)}\,.
\]
Moreover, $\widehat{K}'= \mathcal{O}(\mathrm{min} \{ e^{|\partial A |}( |\partial A| + |C| g(A)), e^{|\partial C |}( |\partial C| + |A| g(C)) \} ).$ 
\end{theo}
}

\section{Discussion}\label{sec:outlook}

Let us conclude this article with a discussion of the equivalence of different notions of decay of correlations in quantum many-body systems. In this work, we have reviewed the notions of exponential uniform decay of covariance, exponential uniform decay of mutual information, the uniform mixing condition and uniform local indistinguishability, which all quantify in some sense that correlations in a quantum Gibbs state decay with the distance between spatially separated regions. 

In \cite{Bluhm2021exponential}, the present authors proved that, for Gibbs states of finite-range interactions and one-dimensional quantum spin chains at any positive temperature, all these notions of decay of correlations are equivalent. The current manuscript together with previous work shows that, under certain conditions, all these notions of decay of correlations hold also for higher-dimensional systems with short-range interactions above a critical temperature. In that sense, the present work can be seen as an extension of \cite{Bluhm2021exponential}. 

\begin{figure}[h!]
\centering
\includegraphics[scale=0.35]{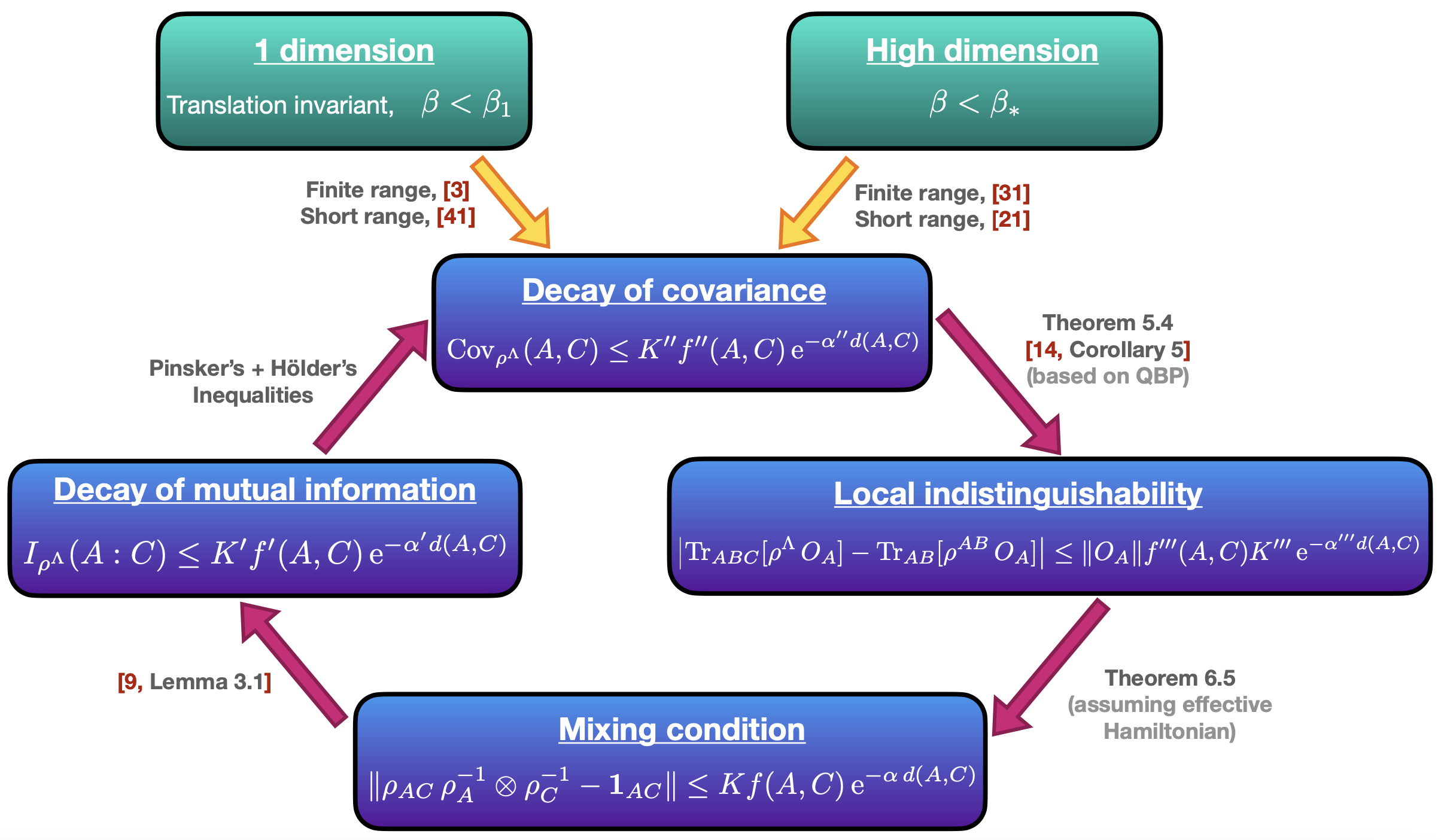}
\caption{Summary of the main results contained in this paper. We consider a positive function $f$ on finite sets that can possibly be different for each type of correlation decay. The equivalence between the four notions of decay of correlations is valid for short-range interactions. For one-dimensional spin systems, $\beta_1$ reduces to $\infty$ for finite-range interactions. Beyond one dimension, we have this equivalence only assuming the existence of an effective Hamiltonian.  }
\label{fig:diagram1}
\end{figure}

On the other hand, contrary to the one-dimensional case, in this work we have to assume the existence of a weak local effective Hamiltonian as in Section \ref{sec:effective_hamiltonian}, motivated by the cluster expansion techniques in \cite{Kuwahara2019}. This seems a quite strong assumption, and actually its strong version is already sufficient to prove the mixing condition, and thus all the different notions of decay of correlations we discussed above. Therefore, we cannot claim that we have shown the equivalence of these different notions of decay of correlations also beyond the one-dimensional case. However, note that the existence of a local effective Hamiltonian is only needed in Step 3 of the proof outline in Section \ref{sec:proof-outline}.

In future work, we will explore whether the existence of a local effective Hamiltonian is equivalent to other notions of decay of correlations, or, failing that, whether we can prove equivalence of different versions of decay of correlations without having to assume the existence of a local effective Hamiltonian. There is hope for that, since this is the case for commuting Hamiltonians, for which the aforementioned equivalence has recently been shown in \cite{kochanowski2023MLSI} without the use of the effective Hamiltonian, building up on previous work from \cite{CapelRouzeStilckFranca-MLSIcommuting-2020}.

\vspace{\baselineskip}
\textbf{Acknowledgements:} AB acknowledges financial support from the European Research Council (ERC Grant Agreement No. 81876) and VILLUM FONDEN via the QMATH Centre of Excellence (Grant No.10059). Moreover, AB is supported by the French National Research Agency in the framework of the ``France 2030” program (ANR-11-LABX-0025-01) for the LabEx PERSYVAL. AC  acknowledges the support of the Deutsche Forschungsgemeinschaft (DFG, German Research Foundation) – Project-ID 470903074 – TRR 352. APH is partially supported by the Escuela T\'{e}cnica Superior de Ingenieros Industriales (UNED) of Spain, project 2023-ETSII-UNED-01, as well as by the Spanish Ministerio de Ciencia e Innovación project PID2020-113523GB-I00 and by Comunidad de Madrid project QUITEMAD-CM P2018/TCS4342. AC and AB are grateful for the hospitality of Perimeter Institute where part of this work was carried out. Research at Perimeter Institute is supported in part by the Government of Canada through the Department of
Innovation, Science and Economic Development and by the Province of Ontario through the Ministry of
Colleges and Universities. This research was also supported in part by the Simons Foundation through the
Simons Foundation Emmy Noether Fellows Program at Perimeter Institute.

\bibliographystyle{plain}
\bibliography{lit}

\end{document}